\documentclass[hidelinks,12pt]{article}
\usepackage{amsmath}
\usepackage{amssymb}
\usepackage{amsthm}
\usepackage{relsize}
\usepackage{mathtools}
\usepackage{url}
\usepackage[usenames,svgnames]{xcolor}
\usepackage{enumitem}
\usepackage{comment}
\usepackage{mathdots}
\usepackage{graphicx}
\usepackage{float}
\usepackage[symbol]{footmisc}
\usepackage[pagebackref=false]{hyperref} 
\usepackage{tikz-cd} 
\usepackage{stackrel} 
\usepackage{lipsum,adjustbox} 
\usepackage{youngtab}

\theoremstyle{plain}
\theoremstyle{definition}
\newtheorem{theorem}{Theorem}[section]
\newtheorem{lemma}[theorem]{Lemma}
\newtheorem{definition-theorem}[theorem]{Definition-Theorem}
\newtheorem{definition-proposition}[theorem]{Definition-Proposition}

\newtheorem{proposition}[theorem]{Proposition}
\newtheorem{corollary}[theorem]{Corollary}
\newtheorem{example}{Example}[section]
\newtheorem{examples}{Example}[subsection]

\newtheorem{remark}{Remark}[section]

\newtheorem{definition}{Definition}[section]

\numberwithin{equation}{section} 

\DeclareMathOperator{\End}{End}
\DeclareMathOperator{\sgn}{sgn}

\newcommand{\medwedge}[1]{\scalebox{.95}{\bigwedge}}

\def\ra{{\rightarrow}}

\def\hra{{\, \hookrightarrow\, }}
\def\mt{{\mapsto}}

\def\det{\mathrm {det}}
\def\Det{\mathrm {Det}}

\def\Spec{\mathrm {Spec}}

\def\deg{\mathrm {deg}}

\def\End{\mathrm {End}}

\def\span{\mathrm {span}}

\def\Gr{\mathrm {Gr}}
\def\res{\mathop{\mathrm {res}}\limits}
\def\swedge{\mathsmaller{\mathsmaller{\wedge\,}}}

\def\GL{\mathrm{GL}}

\def\Sp{\mathrm{Sp}}

\def\res{\mathop{\mathrm{res}}\limits}

\def\&{&{\hskip -20pt}}

\DeclarePairedDelimiter{\no}{:}{:}

\def\be{\begin{equation}}
\def\ee{\end{equation}}

\def\bea{\begin{eqnarray}}
\def\eea{\end{eqnarray}}

\def\bt{\begin{theorem}}
\def\et{\end{theorem}}

\def\bex{\begin{example}\small \rm}
\def\eex{\end{example}}

\def\bexs{\begin{examples}\small \rm}
\def\eexs{\end{examples}}

\def\ra{\rightarrow}
\def \ss {\subset}

\def\br{\begin{remark}\small \rm \em}
\def\er{\end{remark}}

\def\alphab{{\boldsymbol\alpha}}
\def\betab{{\boldsymbol\beta}}

\def\FF {{\mathcal F}}

\def\HH{{\mathcal H}}
\def\II {{\mathcal I}}
\def\JJ{{\mathcal J}}

\def\LL{{\mathcal L}}

\def\OO{{\mathcal O}}

\def\RR {{\mathcal R}}
\def\SS{{\mathcal S}}
\def\TT {{\mathcal T}}

\def\Ab{\mathbf{A}}
\def\Bb{\mathbf{B}}
\def\Cb{\mathbf{C}}

\def\Hb{\mathbf{H}}

\def\Nb{\mathbf{N}}

\def\Pb{\mathbf{P}}

\def\Zb{\mathbf{Z}}
\def\ab{\mathbf{a}}
\def\bb{\mathbf{b}}

\def\hb{\mathbf{h}}

\def\mb{\mathbf{m}}
\def\nb{\mathbf{n}}

\def\sb{\mathbf{s}}
\def\tb{\mathbf{t}}

\def\vb{\mathbf{v}}

\def\xb{\mathbf{x}}
\def\yb{\mathbf{y}}

\def\0b{\boldsymbol{0}}

\def\Cbb{\mathbb{C}}

\def\Fbb{\mathbb{F}}

\def\Hbb{\mathbb{H}}

\def\Hbb{\mathbb{H}}

\def\Nbb{\mathbb{N}}

\def\Zbb{\mathbb{Z}}

\def\grG{\mathfrak{G}} \def\grg{\mathfrak{g}}
 
\def\grI{\mathfrak{I}}

 \def\grl{\mathfrak{l}}

\def\grP{\mathfrak{P}} \def\grp{\mathfrak{p}}

 \def\grs{\mathfrak{s}}

\def\Hc{\check{H}}
\def\Hbc{\check{\Hb}}

\def\hc{\check{h}}
\def\hbc{\check{\hb}}

\makeatletter
\@addtoreset{equation}{section}
\makeatother

\makeatletter
\@addtoreset{equation}{subsection}
\makeatother

\begin{document}
\baselineskip 16pt

\medskip
\begin{center}
\begin{Large}\fontfamily{cmss}
\fontsize{17pt}{27pt}
\selectfont
	\textbf{Tau functions, infinite Grassmannians and lattice recurrences}
	\end{Large}
	
\bigskip \bigskip
\begin{large}
S. Arthamonov$^{1, 2}$\footnote[1]{e-mail:artamono@crm.umontreal.ca},
J. Harnad$^{1, 2}$\footnote[2]{e-mail:harnad@crm.umontreal.ca}
and J. Hurtubise$^{1, 3}$\footnote[3]{e-mail:jacques.hurtubise@mcgill.ca}
 \end{large}
 \\
\bigskip
\begin{small}
$^{1}${\em Centre de recherches math\'ematiques, Universit\'e de Montr\'eal, \\C.~P.~6128, succ. centre ville, Montr\'eal, QC H3C 3J7  Canada}\\
$^{2}${\em Department of Mathematics and Statistics, Concordia University\\ 1455 de Maisonneuve Blvd.~W.~Montreal, QC H3G 1M8  Canada}\\
$^{3}${\em Department of Mathematics and Statistics, McGill University, \\ 805 Sherbrooke St.~W.~Montreal, QC  H3A 0B9 Canada }
\end{small}
 \end{center}
\medskip
\begin{abstract}
\smaller{
The addition formulae for KP $\tau$-functions, when evaluated at  lattice points
in the KP flow group orbits in the infinite dimensional Sato-Segal-Wilson Grassmannian, 
give  infinite parametric families of solutions to discretizations of the KP hierarchy.
The CKP hierarchy may similarly be viewed as commuting flows on the Lagrangian sub-Grassmannian of 
maximal isotropic subspaces with respect to a suitably defined symplectic form. Evaluating  the $\tau$-functions at
a sublattice of points within the KP orbit, the resulting discretization gives solutions both to
 the hyperdeterminantal relations (or Kashaev recurrence) and the hexahedron (or Kenyon-Pemantle) recurrence.}
 \end{abstract}
\break

 \section{Introduction}
 \label{sec:introduction}

 The purpose of this work is to formalize and extend the well-known result  \cite{Mi, SaSa} that evaluation of any KP $\tau$-function 
 $\tau^{KP}_w(\tb)$,  corresponding to an element $w \in \Gr_{\HH_+}(\HH)$ of the Sato-Segal-Wilson \cite{Sa, SaSa, SW}
 infinite dimensional Grassmannian, at suitably defined parametric families of lattice points within the KP flow group
 orbit,  leads to lattice discretizations of the KP hierarchy.
 The CKP hierarchy, viewed as  a restriction to the subgroup of odd parameter KP flows
acting upon the Lagrangian sub-Grassmannian \cite{DJKM1, AHH}  $\Gr^\LL_{\HH_+}(\HH, \omega)\ss \Gr_{\HH_+}(\HH)$
 consisting of maximal isotropic subspaces of the Hilbert space $\HH$,
 with respect to a  symplectic form $\omega$, may be similarly discretized by evaluation of the corresponding 
 KP $\tau$-function on a suitably defined parametric family of infinite sublattices \cite{AHH} within the orbit of any given Lagrangian 
Grassmannian element $w^0\in \Gr^\LL_{\HH_+}(\HH, \omega)$. 

The equations of the discretized KP hierarchy are equivalent to the infinite set of 
bilinear relations satisfied by the Pl\"ucker coordinates of any Grassmannian element $w$.
The fact that these are also satisfied by evaluations of KP $\tau$-functions at suitably defined lattice points 
on the KP flow group orbit follows from the addition formulae (\cite{SaSa, Shig}, Sec. 3.10  of  \cite{HB}).
Viewed as solutions of lattice recurrence relations, they are generated by
 the octahedron relations, or discrete Hirota equations \cite{Hir, Mi}, which correspond to
  the {\em short} (three-term) Pl\"ucker relations.
  
 The discretized CKP hierarchy, on the other hand, consists of an infinite system of
 quartic equations \cite{Sch}, equivalent to the {\em hyperdeterminantal relations} \cite{HoSt, Oed}), which are known 
 to be satisfied by the principal minors of any symmetric matrix. 
 These were also introduced by Kashaev  \cite{Ka} in his study of  the star-triangle relations
satisfied by Boltzmann weights in the Ising model. 
A further set of recurrence relations, the hexahedron relations, was  introduced by Kenyon and Pemantle \cite{KePe1, KePe2}
in their study of double dimer covers and rhombus tilings. These are known to be satisfied by certain 
of the Pl\"ucker coordinates of any Lagrangian Grassmannian element \cite{AHH} or, equivalently, on the big cell,
by the principal and next-to-principal minors of any symmetric matrix. 

  The addition formulae for KP $\tau$-functions (\cite{SaSa, Shig},  \cite{HB}, Sec.~3.10) may be expressed in a
form that is equivalent to the full (infinite) set of  Pl\"ucker relations, which generically follow from the short  ones (\cite{HB}, App. C.9).
Since these are also satisfied by suitably normalized evaluations of  the corresponding KP $\tau$-function at lattice values, 
this provides infinite families of solutions to the octahedron recurrence relations.  A similar result holds 
for the quartic hyperdeterminantal relations \cite{AHH}, which are satisfied both by a suitably chosen subset 
of the Pl\"ucker coordinates of an element of the infinite Lagrangian Grassmannian,
and by correctly normalized lattice evaluations of the corresponding $\tau$-functions of CKP type.

 Subsection \ref{sec:KP_tau_inf_grassmannians} gives a brief review of the Sato-Segal-Wilson formulation
 of the KP hierarchy as abelian flows on an infinite dimensional Grassmannian.
 The fermionic vacuum expectation value  (VEV) representation of the $\tau$-function
 and its Schur function expansion is recalled in Subsection \ref{sec:fermionic_rep_tau}. The KP addition theorem
 \cite{SaSa, Shig} is recalled in Subsection \ref{sec:add_thm_lattice_eval}, together with Sato's interpretation \cite{SaSa}
 of the latter in terms of Pl\"ucker relations. The CKP reduction of the KP hierarchy, which involves  restriction
 to the subgroup of odd flows on the Lagrangian Grassmannian \cite{DJKM1, AHH, KZ} is summarized 
 in Subsection \ref{sec:CKP_lagrange_grassmann}.
 
  The main new results,  stated in Subsections \ref{sec:lattice_lagrange_hyperdet} and \ref{sec:hexahedron},
and proved in Sections  \ref{sec:lattice_embeddings_lagrangian} and \ref{sec:applications_lattice},  consist of
showing that suitable lattice evaluations of KP $\tau$-functions of the CKP type determine solutions of the Kashaev recursions
 (or ``core'' hyperdeterminantal relations) and the Kenyon-Pemantle hexahedron relations. 

 \begin{remark}
 An analogous reduction to the sub-Grassmannian $\Gr^0_{\HH_+ }(\HH, Q)$ of maximal isotropic subspaces with respect to
 a quadratic form $Q$ leads similarly to the BKP (or DKP) hierarchy, but the corresponding $\tau$-functions are  defined
as Fredholm Pfaffians rather than determinants (see \cite{BHH}, and Sec. 7.1 and App. E of \cite{HB}). These 
satisfy both the BKP analogue of the Hirota bilinear residue equations \cite{DJKM1, DJKM3, JM1} and an analogous set of addition 
formulae (\cite{Shig} and \cite{HB}, Sec. 7.2). They may again be related to bilinear equations (the {\em Cartan equations}
 \cite{Ca}, \cite{HB}, App. E.2) satisfied by the expansion coefficients in the natural basis for the corresponding 
 projectivized exterior space  into which $\Gr^0_{\HH_+ }(\HH, Q)$ is mapped under  the {\em Cartan embedding},
 which is the orthogonal group analogue of the Pl\"ucker embedding (see  \cite{Ca} and \cite{HB}, Sec. 7.1 and App.~E).
  The addition formulae  may similarly be interpreted as  discretizations of the BKP (or DKP) hierarchy, 
  and solutions  obtained by replacing the Cartan  coefficients by suitably  normalized evaluations of the BKP $\tau$-functions at lattice points within the BKP flow group orbit.  The minimal generating set for these relations are the {\em short Cartan relations},  which may also be 
 interpreted as a system of recurrence relations, the {\em Miwa relations}  \cite{Mi}, or {\em cube recursions} 
 (see, e.g.,  \cite{HB}, Sec.~7.7.2). These will not be further considered here,  but will form the subject of a subsequent work.
 \end{remark}

 \subsection{KP $\tau$-functions and infinite Grassmannians }
\label{sec:KP_tau_inf_grassmannians}

 \subsubsection{Hirota residue equation}
\label{sec:Hirota_res_equation}

    Solutions of the KP hierarchy of integrable PDE's  are given in terms of
an associated $\tau$-function $\tau^{KP}_w(\tb)$, which depends on an infinite set of KP flow variables
\be
\tb:=(t_1, t_2, \cdots)
\ee
and satisfies the formal Hirota bilinear residue equation \cite{Sa, SW, JM1,  JM2, HB}
\be
\res_\infty\left(e^{\xi(\delta \tb, z)} \tau_w^{KP}(\tb -[z^{-1}]) \tau_w^{KP}(\tb + \delta\tb +[z^{-1}])\right) =0,
\label{eq:hirota_res_eq}
\ee
where
\be
\delta \tb := (\delta t_1, \delta t_2, \dots), \quad [z]:= (z, \frac{z^2}{2}, \dots , \frac{z^i}{i}, \dots), \quad
\xi(\tb, z) := \sum_{i=1}^\infty {t_i z^i}.
\ee
For any formal power series
\be
f(z) = \sum_{i\in \Zbb} f_i z^i,
\ee
the formal residue appearing  in (\ref{eq:hirota_res_eq}) is defined as
\be
\res_\infty(f):= f_{-1}
\ee
 and eq.~(\ref{eq:hirota_res_eq}) is understood as satisfied identically in the parameters
$\delta \tb = (\delta t_1, \delta t_2, \dots)$.

 \subsubsection{Infinite  Grassmannians}
\label{sec:infinite_grassmannian}

There are two different approaches to the definition of the infinite Grassmannian appearing in relation  to
$\tau$-functions: one due to Sato \cite{Sa,SaSa},  the other to Segal and Wilson \cite{SW}.
The formulation of Segal and Wilson \cite{SW}, is functional-analytic and geometric in flavour, 
with the Grassmannian  $\Gr_{\HH_+}(\HH)$ defined as a Banach manifold, and the KP flow variables  
viewed as additive coordinates on an infinite dimensional abelian transformation group $\Gamma_+$
acting on it.
  In the approach of Sato  \cite{Sa, SaSa, JM1, DJKM1} the infinite Grassmannian is identified 
  as a direct limit of finite dimensional Grassmannians. The group is not viewed as acting continuously,  but determines the dynamics
 via a formal expansion of the $\tau$-function as a linear combination of Schur functions, with coefficients
 equal to the Pl\"ucker coordinates of the initial point $w$.   (The formal series interpretation is of particular 
 relevance when $\tau$-functions are used as generating functions for enumerative combinatorial invariants 
 (see, e.g. \cite{HB}, Chapts. 13 and 14)).) 
 
  In the present work, we mainly use the Sato approach, viewing $\tau$-functions as formal series in the flow parameters, 
  which are interpreted as indeterminates, and emphasize the coordinate rings, rather than the geometry of flows 
  on the Grassmannian. But for the purposes of this introduction, we begin with the geometrical viewpoint, which 
  corresponds more readily to the intuitive notion of  group actions on Grassmannians.

KP $\tau$-functions $\tau^{KP}_w(\tb)$ are parametrized \cite{Sa, SaSa,  SW} by elements $w$ of an infinite dimensional
Grassmann manifold $\Gr_{\HH_+}(\HH)$,  which are subspaces $w\ss\HH$ of an infinite dimensional vector space $\HH$,
with denumerable basis $\{e_i\}_{i\in \Zbb}$ and polarization
\be
\HH= \HH_+ + \HH_-, \quad \HH_+:=\span\{e_{-i}\}_{i\in \Nbb^+}, \quad \HH_-:=\span\{e_i\}_{i\in \Nbb},
\label{H_pm}
\ee
that are {\em commensurable} with $\HH_+\ss\HH$ (in the sense defined in \cite{SW}). $\Gr_{\HH_+}(\HH)$ is a homogeneous
space of the infinite dimensional group $\GL(\HH)$ of  invertible linear transformations of $\HH$ preserving
{\em admissible} frames \cite{Sa, SaSa, SW}.  (Here,  $\Zbb$, $\Nbb$ and $\Nbb^+$ 
denote the set of all integers,  nonnegative integers and positive integers, respectively.)
 The KP dynamics may be understood geometrically as the action $\Gamma_+ \times  \Gr_{\HH_+}(\HH) \ra \Gr_{\HH_+}(\HH)$
of an infinite abelian subgroup $\Gamma_+\ss \GL(\HH)$ on $\Gr_{\HH_+}(\HH)$.

\begin{remark}
For precise definitions  of: {\em subspaces} $w\ss \HH$ {\em commensurable}  with $\HH_+$, {\em admissible frames}
for $w\in \Gr_{\HH_+}(\HH)$, the {\em infinite group} $\GL(\HH)$ of invertible linear transformation of $\HH$,
 its action on the bundle of frames over $\Gr_{\HH_+}(\HH)$ and how the determinant (\ref{tau_fn_det_def}) is defined in terms of
lifts of the $\Gamma_+$-action to the {\em dual determinantal line bundle} $\Det^*\ra  \Gr_{\HH_+}(\HH)$,
see \cite{SW} and \cite{HB}, Chapt.~3.

In order to relate the above to the definition of  infinite dimensional Grassmannians in terms of inverse limits
of  finite coordinate sheaves used in Section \ref{sec:Infinite_Grassmannians_KP_tau},
 the spaces $\HH, \HH_+, \HH_-$  must be viewed as suitably defined completions.\
 \ \end{remark}

 \subsubsection{Fermionic Fock space and Pl\"ucker map}
\label{sec:fermionic_fock_space}

We may  embed $ \Gr_{\HH_+}(\HH)$,  into the projectivization of an associated fermionic Fock space,
which is viewed as a graded space consisting of the semi-infinite wedge product  of $\HH$ with itself,
\be
\FF := \scalebox{.95}{$\bigwedge$}^{\infty/2}\,\HH = \bigoplus_{n\in \Zbb} \FF_n,
\ee
via the infinite dimensional analogue of the  Pl\"ucker map \cite{GH},
\bea
\grP: \Gr_{\HH_+}(\HH) &\&\ra \Pb(\FF) \cr
\grP: \span\{w_1, w_2, \dots \} &\&\mapsto [ w_1 \wedge  w_2, \wedge \cdots] \in  \Pb(\FF).
\label{eq:inf_plucker_map}
\eea
Here $\FF_n \ss \FF$ is the graded sector with fermionic charge $n$, which is spanned by
orthonormal basis elements
\be
|\lambda;n\rangle := e_{l_1}\wedge e_{l_2} \wedge \cdots
\label{fermionic_basis_els}
\ee
labelled by pairs $(\lambda, n)$  of an integer partition $\lambda = (\lambda_1, \lambda_2, \dots)$ and an integer
$n\in \Zbb$, where the indices $(l_1, l_2, \dots)$ are the {\em particle locations} (see \cite{JM1},  \cite{SW} and \cite{HB}, Chapt. 5):
\be
l_i := \lambda_i -i +n,
\ee
which form a strictly decreasing sequence of integers that eventually saturate at all consecutive decreasing
integers. The vacuum vector in each sector $\FF_n$ corresponds to the null partition and  is denoted
\be
|n\rangle := |\emptyset; n\rangle.
\ee

 For $w\in \Gr_{\HH_+}(\HH)$ of {\em virtual dimension}  $n$ (i.e. for which the projection map
 \be
 \Pi_+: w \ra \HH_+
 \ee along $\HH_-$ is Fredholm, with index $n$ \cite{SW}),
the image  $\grP(w) \in \Pb(\FF_n),$ of the Pl\"ucker map may be expressed as a linear combination of the basis elements
\be
\grP(w) =\big[ \sum_{\lambda} \pi_\lambda(w) |\lambda;n\rangle \big] \in \Pb(\FF_n),
\label{eq:plucker_exp_sector_n}
\ee
where the coefficients $\pi_\lambda(w)$ are (within projectivization) the Pl\"ucker coordinates of $w$.
Bose-Fermi equivalence (\cite{HB}, Chapt.~5.6) gives the corresponding expansion of $\tau^{KP}_w(\tb)$
 in a basis of Schur functions \cite{Mac, SaSa}
\be
\tau_w^{KP}(\tb) = \sum_\lambda \pi_\lambda(w) s_\lambda(\tb),
\label{eq:schur_fn_exp_tau}
\ee
where the $t_i$'s are interpreted as normalized power sums in a set of auxiliary variables $\{x_a\}_{a\in \Nb}$
\be
t_i= \frac{p_i}{i} := \frac{1}{i}\sum_{a\in \Nb} x_a^i.
\label{eq:normal_power_sums}
\ee
(The Hirota residue equation (\ref{eq:hirota_res_eq}) is then equivalent to the fact that the $\pi_\lambda(w)$'s
satisfy the full set of Pl\"ucker relations  (\ref{eq:PluckerRelationIndices}).
(See eqs.~(\ref{eq:piLDefDet}), (\ref{eq:FiniteLThroughLambda})
 and  (\ref{eq:piTildepiLambda}) for  detailed definitions of the notations used.)

Equivalently, let
\be
\Lambda : e_i  \mapsto e_{i-1}, \quad  i\in \Zb,
\ee
be the  shift map $\Lambda : \HH \ra \HH$ and denote
the infinite abelian group of KP {\em shift flows} $\Gamma_+ \ss \GL(\HH)$ by
\bea
\Gamma_+ &\&:= \{\gamma_+(\tb):=e^{\sum_{i=1}^\infty t_i \Lambda^i }\}\cr
\tb&\&= (t_1, t_2, \dots ), \quad t_i\in \Cb, \ i \in \Nbb^+.
\eea
Applying the elements $\gamma_+(\tb)\in \Gamma_+$ multiplicatively to the basis elements $\{w_i\}_{i\in \Nbb^+}$
spanning $w \in \Gr_{\HH_+}(\HH)$ induces an action of $\Gamma_+$ 
as commuting flows on the Grassmannian,
\bea
\Gamma_+:\Gr_{\HH_+}(\HH) &\& \ra \Gr_{\HH_+}(\HH) \cr
\gamma_+(\tb) :w &\&\ \mt  =w(\tb):=\span\{w_i(\tb) := \gamma_+(\tb)w_i\}
\eea
Denoting the projection  map  from $w(\tb)$ to $\HH_+$ along $\HH_-$  as $\pi_+: w(\tb) \ra \HH_+$
and choosing an {\em admissible basis} \cite{SW} for $w(\tb) \in  \Gr_{\HH_+}(\HH)$,  allowing it to
be identified isomorphically with $\HH_+$, the KP $\tau$-function may be defined as the
determinant of the projection map
\be
\tau_w^{KP}(\tb) := \det(\Pi_+: w(\tb) \ra \HH_+).
\label{tau_fn_det_def}
\ee


\subsection{Fermionic representation of KP $\tau$-functions}
\label{sec:fermionic_rep_tau}

The fermionic representation of the Clifford algebra on $\HH\oplus \HH^*$ with respect to the natural scalar product
   \be
   Q(e_i, e^*_j) = \delta_{ij}, \quad Q(e_i, e_j) = Q(e^*_i, e^*_j)=0 ,
   \ee
   where $\{e^*_j\}_{j\in \Zbb}$ is the basis for $\HH^*$ dual to the basis $\{e_j\}_{j\in \Hbb}$ for $\HH$
as endomorphisms of $\FF$,   is generated by the creation and annihilation operators
   \be
   \psi_i := e_i\wedge \in \End(\FF), \quad \psi^\dag_i := i_{e^*_i} \in \End(\FF), \quad i \in \Zbb
   \ee
   defined, respectively, as outer and inner products with respect to the basis elements.
   These  satisfy the usual fermionic anticommutation relations
   \be
   [\psi_i, \psi_j]_+ = 0, \quad  [\psi^\dag_i, \psi^\dag_j]_+ =0, \quad  [\psi_i, \psi^\dag_j]_+ = \delta_{ij}, \quad i,j \in \Zbb
   \ee
   and vacuum annihilation conditions
   \be
   \psi_{-i} | 0\rangle =0 , \quad \psi^\dag_{i-1} |0\rangle =0, \quad \forall \  i \in \Nbb^+.
   \ee
In terms of the Frobenius  index notation for partitions $\lambda$ (i.e., the ``arm'' and ``leg'' lengths
in the corresponding Young diagrams \cite{Mac})
\be
\lambda = (\ab | \bb) = (a_1, \dots, a_r | b_1, \dots, b_r),
\ee
 the basis vectors (\ref{fermionic_basis_els}) are (see \cite{JM1, JM2}, or \cite{HB}, Chapt. 5.1)
\be
|\lambda; n\rangle = (-1)^{\sum_{i=1}^r b_r} \psi_{a_i+n} \psi^\dag _{-b_i +n-1}|n \rangle.
\ee

For an exponential element $g \in \GL(\HH)$
\be
g = e^A, \quad  A:= \sum_{ij \in \Zbb} A_{ij} E_{ij} \in \grg\grl(\HH), \quad E_{ij}(e_k) = e_i \delta_{jk} ,
\ee
we have the Clifford representation
\be
\hat{g} = e^{\sum_{i,j \in \Zbb} A_{ij}\no{\psi_i \psi^\dag_j} },
\ee
where
\be
\no{\psi_i \psi^\dag_j} := \psi_i \psi^\dag_j - \langle  0| \psi_i \psi^\dag_j| 0 \rangle
\ee
is the normal ordered product. As indicated in (\ref{eq:plucker_exp_sector_n}), the  Pl\"ucker
image of an element $w \in \Gr_{\HH_+}(\HH)$  of  {\em virtual dimension} $n$ is in the fermionic
charge sector $\FF_n$. In particular, the subspace
\be
\HH^n_+ :=  \span \{ e_{n-i}, \dots \}_{i\in \Nbb^+} 
\ee
has virtual dimension $n$, and its Pl\"ucker image is (the projectivization of) the vacuum vector
\be
[|n \rangle ]= \grP(\HH^n_+ ) \in\Pb(\FF_n)
\ee
in the sector $\FF_n$. Any element $w\in \Gr_{\HH_+}(\HH)$ with virtual dimension $n$ may be expressed as
\be
w = g(\HH^n_+),
\label{w_g_HH_^n}
\ee
where $g\in \GL^0(\HH)$ is in the identity component of $\GL(\HH)$, and is determined up
to right multiplication by an element of the stabilizer of $\HH^n_+$.

The fermionic representation of the KP flow group $\Gamma_+$ on $\FF$  is given by
\be
\gamma_+(\tb)\ \mt \ \hat{\gamma}_+(\tb) = e^{\sum_{i=1}^\infty {t_i J_i}} \in \End(\FF),
\ee
where
\be
J_i := \sum_{j\in \Zbb} \psi_j \psi^\dag_{j+i}, \quad i\in \Nbb^+
\ee
are the negative Fourier components of the   current operator.
\bea
J(z) &\&:= \no{\psi(z) \psi^\dag(z)}= \sum_{i\in \Zbb} J_i z^{-i} , \\
\psi(z) &\&:= \sum_{i\in \Zbb} \psi_i z^i,\quad    \psi^\dag(z) := \sum_{i\in \Zbb} \psi^\dag_i z^{-i}.
\eea
From the equivariance of the Pl\"ucker map (\ref{eq:inf_plucker_map}),
the  image  $\grP(w(\tb))$ of the element
\be
w(\tb) = \gamma_+(\tb)w
\ee
 in the $\Gamma_+$ orbit of $w\in \Gr_{\HH_+}(\HH)$ is
\be
\grP(w(\tb)) = [\hat{\gamma}_+(\tb) |w\rangle ]
\ee
where
\be
[|w\rangle] := \grP(w).
\ee

It follows from the determinantal interpretation of the Pl\"ucker coordinates of 
  $\grP(w(\tb))$ and definition (\ref{tau_fn_det_def}),
 that $\tau^{KP}_w(\tb)$ is expressible, within a
projective normalization factor, as the fermionic vacuum expectation value
\be
\tau^{KP}_w(\tb) = \langle n | \gamma_+(\tb) |w\rangle.
\ee
If $w$ is expressed as in (\ref{w_g_HH_^n}), we have $\grP(w)\in \Pb(\FF_n)$ and
this becomes
\be
\tau^{KP}_w(\tb) = \langle n | \gamma_+(\tb) \hat{g} |n\rangle.
\label{tau_g_fermionic}
\ee

The Pl\"ucker coordinates are  the fermionic matrix elements
\be
\pi_\lambda(w) = \langle \lambda; n | \hat{g}| n \rangle,
\ee
so (\ref{tau_fn_det_def}) is equivalent to
\be
\tau_w^{KP}(\tb)= \pi_\emptyset(\grP(w(\tb)).
\label{tau_null_plucker}
\ee


\subsection{KP addition formulae, lattice evaluations and discrete KP}
\label{sec:add_thm_lattice_eval}

A standard result in the theory of KP $\tau$-functions is that the Hirota residue
equation (\ref{eq:hirota_res_eq}) is also equivalent to an infinite set of multi-parametrically
defined {\em addition formulae}  (\cite{SaSa, Shig}, \cite{HB}, Chapt. 3)  satisfied by evaluations of $\tau^{KP}_w(\tb)$ at lattice of points
in the $\Gamma_+$ orbit  $\{w(\tb)\}$ of $w$.
In the notation of \cite{HB}, for a finite set of $n$ parameters $(x_1, \dots, x_n)$, let
\be
\zeta_n(\tb,  x_1, \dots, x_n):= \prod_{i<j}(x_i-x_j)\tau_w^{KP} (\tb + \sum_{i=1}^n[ x_i]).
\label{eq:zeta_xb_tb_def}
\ee
We then have  (\cite{HB}, Sec.~3.10).
\begin{theorem}[KP addition formulae]
\label{th:addition_formula}
 If $\tau({\bf t})$ satisfies the Hirota bilinear equations (\ref{eq:hirota_res_eq}), then for  $k \in \Nbb^+$ and any
$2k$ distinct parameters $(x_1, \dots, x_{k-1}; y_1, \dots , y_{k+1})$, the following addition formulae
hold:
\be
\sum_{j=1}^{k+1} (-1)^j \zeta_k({\bf t}; x_1, \dots, x_{k-1}, y_j)  \zeta_k({\bf t}; y_1, \dots, \hat{y}_j, \dots,
y_{k+1}) = 0,
\label{tau_k_addition_formula}
\ee
where $(y_1, \dots, \hat{y}_j, \dots, y_{k+1})$ denotes the sequence of $k$ parameters obtained by omitting $y_j$
from the sequence $(y_1,  \dots, y_{k+1})$.
\end{theorem}

For an infinite sequence of parameters $\xb=(x_1, x_2, \dots)$,
we define (cf.~eq.~(\ref{eq:NormalizedEvaluationOfTauFunction})), Subsection \ref{sec:lattice_eval_KP_tau})
the following  normalized  evaluations of $\tau^{KP}_w(\tb)$ at lattice points,
\be
 H_w^{\nb}(\tb):=\prod_{i<j}(x_i-x_j)^{n_in_j}\;\tau_w^{KP}\Big(\tb+\sum_{i=-\infty}^{+\infty}n_i [x_i]\Big),
 \label{H_w_def_tb}
\ee
where
\be
\nb = \sum_{i\in \Zbb} n_i \alphab_i \in (\Zbb)^\infty
\ee
is any infinite integer lattice element with a finite number of nonzero components $\{n_i\}_{i\in \Zbb}$ and $\{\alphab_i\}_{i\in \Zbb}$
are the unit basis elements having $1$  in position $i$ and $0$ elsewhere.
For any choice of lattice element $\nb$,  initial flow group element $\tb_0$,
and partition $\lambda$ with Frobenius indices $(a_1, \dots, a_r | b_1, \dots, b_r)$,
define $ H_{w}^{\nb,\lambda}(\tb_0)$ as
\be
   H_{w}^{\nb,\lambda}(\tb_0):=H_w^{\nb+\sum_{i=1}^r \alphab_{a_i}-\sum_{i=1}^r\alphab_{-b_i-1} }(\tb_0).
\ee
As pointed out by Sato \cite{SaSa},  the addition formulae (Theorem \ref{th:addition_formula})
 then imply, for any fixed choice of lattice element $\nb$ and initial value $\tb_0$,
that the $H_{w}^{\nb,\lambda}(\tb_0)$'s satisfy the same infinite set of Pl\"ucker relations (\ref{eq:PluckerRelationIndices})
as the Pl\"ucker coordinates $\{\pi_\lambda(w)\}$. This is  stated and proved inductively in Theorem
\ref{thm:plucker_lattice_homomorph}.

In particular, these  include the discrete Hirota equations \cite{Hir, Zab}  (or octahedron recurrences
\cite{Mi, Sp}, or $T$-relations \cite{KNS}), which define a discretization of the KP hierarchy,
and may be expressed \cite{Mi, Nim, Zab}  as
\be
a(b-c) \tau_{l+1, m,n} \tau_{l, m+1, n+1} + b(c-a) \tau_{l, m+1, n} \tau_{l,m+1, n+1}
+ c(a-b) \tau_{l, m, n+1} \tau_{l+1, m+1, n} =0,
\label{eq:octahedron_rels}
\ee
where
\be
\tau_{l, m, n} := \tau^{KP}_w(\tb_0 + l[a] + m [b] + n[c]), \quad l, m,n \in \Zbb.
\label{eq:tau_lmn_def}
\ee

\begin{remark} The recurrences associated to these discretizations take the same form  as the short Pl\"ucker relations
 satisfied by the Pl\"ucker coordinates of a single element of the Grassmannian $\Gr_{\HH_+}(\HH)$.
  But the relations are associated here to a  group of discrete commutative flows on a line bundle over the Grassmannian.

\end{remark}

 \subsection{Lagrangian Grassmannians and the CKP hierarchy}
\label{sec:CKP_lagrange_grassmann}

A symplectic form $\omega$ may be defined on $\HH$  by  evaluation on the basis elements $\{e_i\}_{i\in \Zbb}$
\be
\omega(e_i, e_j) =-\omega(e_j, e_i) := (-1)^i \delta_{i, -j -1},  \quad i, j \in \Zbb.
\label{symplectic _form}
\ee
The symplectic subalgebra  $\grs\grp(\HH,\omega)  \ss\grg\grl(\HH)$ consists of the elements $X \in \grg\grl(\HH)$ satisfying
\be
\omega(X u, v) + \omega(u, X v) =0, \quad \forall \ u, v \in \HH.
\label{sp_sympl_invar}
\ee
The symplectic subgroup $\Sp(\HH, \omega) \ss \GL(\HH)$, with Lie algebra $\grs\grp(\HH, \omega)$,
consists of elements $g^0 \in \GL(\HH)$ that preserve $\omega$:
\be
\Sp(\HH, \omega) = \{ g^0 \in \GL(\HH)\  |\  \omega(g^0 u, g^0 v) = \omega(u, v), \quad \forall \ u, v \in \HH\}.
\label{Sp_group_invar}
\ee

The Lagrangian Grassmannian $\Gr^\LL_{\HH_+}(\HH, \omega) \ss \Gr_{\HH_+} (\HH)$
is the sub-Grassmannian consisting of maximal isotropic subspaces $w^0\in Gr_{\HH_+} (\HH)$;
i.e. those on which the restriction of $\omega$ vanishes
\be
\omega \vert_{w^0} =0,  \quad \forall \ w^0\in \Gr^\LL_{\HH_+}(\HH, \omega).
\label{inf_lagrangian_cond}
\ee
This is a homogeneous space of the symplectic group $\Sp(\HH, \omega)$,
 and may be viewed as the orbit of  the element $\HH_+\ss \HH$
\be
\Gr^\LL_{\HH_+}(\HH, \omega) = \Sp(\HH, \omega) (\HH_+).
\ee
Restricting to elements $w^0 = g^0(\HH_+)$ of the Lagrangian Grassmannian,
where $g^0\in \Sp(\HH, \omega)$,  determines a KP $\tau$-function with fermionic representation
\be
\tau^{KP}_{w^0}(\tb) = \langle 0 | \hat{\gamma}_+(\tb) \hat{g}^0 | 0 \rangle
\ee
which, as shown in \cite{DJKM1, DJKM3, KZ},  must satisfy the following symmetry condition
\be
\tau^{KP}_{w^0}(\tb) = \tau^{KP}_{w^0}(\tilde{\tb}),
\label{eq:LagrangianIsotropyConditionTau}
\ee
where
\be
\tilde{\tb}:=(t_1, -t_2, t_3, -t_4, \dots).
\ee

The full KP flow group $\Gamma_+$ does not leave $\Gr^\LL_{\HH_+}(\HH, \omega)$ invariant,
but the subgroup $\Gamma'_+\ss\Gamma_+$ consisting of purely odd flows
\be
\Gamma'_+:= \Gamma_+\cap \Sp(\HH, \omega) = \{\gamma_+(\tb'), \quad \tb':=(t_1, 0,  t_3, 0, \cdots)\}
\ee
does. The corresponding Baker function, restricted to vanishing even values
\be
\tb = \tb' := (t_1, 0, t_3, 0, \dots)
\ee
of the KP flow variables is given by the Sato formula \cite{Sa, SaSa}
\be
\Psi_{w^0}(z, \tb') := e^{\xi(z, \tb')} \frac{\tau^{KP}_{w_0}(\tb' - [z^{-1})]')}{\tau^{KP}_{w^0}(\tb')},
\ee
where
\be
\xi(z, \tb') := \sum_{j=1}^\infty t_{2j -1} z^{2j -1}, \quad
[z^{-1}]' := \left(z^{-1}, 0, {1\over 3} z^{-3}, {1\over 5}  z^{-5},0,  \dots \right),
\ee
This satisfies the Hirota bilinear residue equation in the form \cite{DJKM1, DJKM3}
\be
\res_{z=\infty}\left(\Psi_{w^0}(z, \tb')  \Psi_{w^0}(-z,\tb'+\delta \tb' )
\right)dz= 0
\label{hirota_bilinear_tau_res_CKP}
\ee
identically in
\be
\delta\tb' = (\delta t_1, 0, \delta t_3, 0 \dots)
\ee
which gives  the reduction of the KP to the CKP hierarchy \cite{DJKM1, DJKM2, DJKM3, KZ}.

 \subsection{Lattice evaluations of $\tau$-functions and hyperdeterminantal relations}
\label{sec:lattice_lagrange_hyperdet}
Let $\Bb$  be the infinite sublattice  of $(\Zbb)^{\Nbb}$ consisting of infinite sequences $\nb' =(\dots,  n'_i, n'_{i+1}, \dots)$, $n'_i \in \Zbb$,
where only a finite number of the entries  are nonzero,  and these satisfy
\be
n'_{-i} = - n'_{i-1} \quad \forall \ i \in \Nbb^+.
\ee
Choosing an element $w^0\in \Gr_{\HH_+}^\LL(\HH, \omega)$ and an infinite set of parameters $\{y_i\}_{i\in \Nbb^+}$
we define, similarly to (\ref{H_w_def_tb}),  a function $\hb$ on $\Bb$  whose values are
\begin{equation}
\begin{aligned}
h^{\nb'}
:=&\prod_{1\leq i< j}(y_i-y_j)^{n'_{i-1}n'_{j-1}+n'_{-i}n'_{-j}} \prod_{i,j=1}^{\infty}(-1)^{n'_{-i}n'_{j-1}}(y_i+y_j)^{n'_{-i}n'_{j-1}}
\\ &\qquad\times
\tau_{w^0}^{KP}\left(\tb'+\sum_{i=1}^{\infty}\big(n'_{i-1}[y_i]+n'_{-i}[-y_i]\big)\right),
\end{aligned}
\label{eq:Lagrangian_valuation_tau}
\end{equation}
(cf.~eqs. (\ref{eq:LagrangianEvaluationTauFunction}), (\ref{eq:hb_bold_n_prime_def}),
(\ref{eq:hb_bold_n_prime_v_def})).

Let $\{\betab_i\}_{i\in \Nbb^+}$ be the generating basis for $\Bb$ consisting of elements $\betab_i$ with just two nonvanishing components 
\be
n'_j=1, \quad n'_{-j-1} = -1\  \text{ for }  j= i-1, \quad \text{and }  n_j =0 \ \text{ for all other } j \in \Zbb.
\ee
For all distinct positive integer $l$-tuples $(i_1, \dots i_l)$, denote  as $\hb_{\betab_{i_1}+\dots+\betab_{i_l}}$ 
the corresponding lattice functions on $\Bb$ with values
\bea
h^{\nb'}_{\betab_{i_1}+\dots+\betab_{i_l}}:=h^{\nb' +\betab_{i_1}+\dots+\betab_{i_l}},  \quad \nb'\in \Bb,
\label{eq:Lagrangian_valuation_tau_transl}
\eea
 given by shifting the argument  by such a sum of generators.
Proposition \ref{prop:HyperdeterminantalRecurrence} says that, choosing $l=0,1,2$ or $3$ in
 (\ref{eq:Lagrangian_valuation_tau_transl}), for all  triples $i,j,k \in \Nbb^+$, 
  $i<j<k$,
 these provide solutions to the system 
\begin{equation}
\begin{aligned}
&\hb^2\hb_{\betab_i+\betab_j+\betab_k}^2 +\hb_{\betab_i}^2\hb_{\betab_j+\betab_k}^2 +\hb_{\betab_j}^2\hb_{\betab_i+\betab_k}^2 +\hb_{\betab_k}^2\hb_{\betab_i+\betab_j}^2\\
&-2\hb\hb_{\betab_i+\betab_j+\betab_k}\left(\hb_{\betab_i}\hb_{\betab_j+\betab_k} +\hb_{\betab_j}\hb_{\betab_i+\betab_k} +\hb_{\betab_k}\hb_{\betab_i+\betab_j}\right)\\
&-2\left(\hb_{\betab_i}\hb_{\betab_j}\hb_{\betab_i+\betab_k}\hb_{\betab_j+\betab_k} +\hb_{\betab_j}\hb_{\betab_k}\hb_{\betab_i+\betab_j}\hb_{\betab_i+\betab_k} +\hb_{\betab_i}\hb_{\betab_k}\hb_{\betab_i+\betab_j}\hb_{\betab_j+\betab_k}\right)\\
&+4\hb\hb_{\betab_i+\betab_j}\hb_{\betab_i+\betab_k}\hb_{\betab_j+\betab_k} +4\hb_{\betab_i}\hb_{\betab_j}\hb_{\betab_k}\hb_{\betab_i+\betab_j+\betab_k} =0, 
\end{aligned}
\label{eq:HyperdeterminantalRecurrenceIntroduction}
\end{equation}
 known as the $2 \times 2 \times 2$ {\em hyperdeterminantal relations}, which are  satisfied by the principal minors of any symmetric $N \times N$ matrix \cite{HoSt, Oed}. In the theory of integrable systems, they arise as the Kashaev recurrence \cite{Ka}, and are also  interpretable as a discretization of the CKP hierarchy \cite{Sch, BobSch}. Proposition \ref{hexahedron_Lagrangian_36} derives these as a consequence
of Theorem \ref{th:LatticeOfPointsInInfiniteLagrangianGrassmannian}, which gives a homomorphism, as commutative
algebras, from the coordinate algebra of  the infinite dimensional Lagrangian Grassmannian $\Gr^\LL_{\HH_+}(\HH, \omega)$
to a localization of the ring  of formal power series in the infinite sequence of odd KP flow variables
 $\tb':=(t_1, 0, t_3, 0, \cdots)$ to which the $\tau$-function is restricted.

This provides  an infinite lattice extension of Corollary 3.20 of \cite{AHH}, where the hyperdeterminantal relations (\ref{eq:HyperdeterminantalRecurrenceIntroduction}) are shown satisfied by 
 finite lattice evaluations of the $\tau$-function. 
 (Cf.  eq.~(3.6.43) of \cite{AHH}, where $\sigma,\sigma_i,\sigma_{ij},\sigma_{ijk}$ (for $i, j, k \in \{1, \dots, N\}$) 
 are essentially the same as $\hb,\hb_{\betab_i}, \hb_{\betab_i+\betab_j}, \hb_{\betab_i+\betab_j+\betab_k}$, up to an overall
 rational normalization factor in the parameters).

\subsection{Lagrangian Pl\"ucker coordinates hexahedron relations}
\label{sec:hexahedron}

Kenyon and Pemantle \cite{KePe1, KePe2} introduced a system of  relations
that they called the {\em hexahedron recurrence}, and applied it to the study of double dimer covers and rhombus tilings. 
These are defined as follows. Let $\hbc,\hbc^{(x)},\hbc^{(y)},\hbc^{(z)}$ be four functions on the  3-dimensional lattice
\begin{equation}
\Bb_3:=\mathrm{span}_{\Zbb}\{\betab_1,\betab_2,\betab_3\}
\end{equation}
with basis vectors $\betab_1,\betab_2,\betab_3$,  and values at any $\nb'\in \Bb_3$ denoted
$\hc^{\nb'},\hc^{\nb', (x)},\hc^{\nb',(y)},\hc^{\nb',(z)}$, respectively. 
For any $\vb'\in \Bb_3$, define four additional functions
 $\hbc_{\vb'},\hbc^{(x)}_{\vb'},\hbc^{(y)}_{\vb'},\hbc^{(z)}_{\vb'}$ with values
\be
\hc^{\nb'}_{\vb'}:= \hc^{\nb' +\vb'}, \quad \hc^{\nb', (x)}_{\vb'}:= \hc^{\nb'+\vb', (x)},\quad \hc^{\nb',(y)}_{\vb'}:= \hc^{\nb' +\vb',(y)}
\quad \hc^{\nb',(z)}_{\vb'}:= \hc^{\nb' +\vb',(z)},
\ee
 obtained by translating the evaluation point by $\vb'$.
 These  satisfy the hexahedron recurrence if\;\footnote{To compare with the notation of \cite{KePe2}, we have the
following correspondence:
\bea
&\& \hbc \leftrightarrow h, \quad \hbc^{(x)} \leftrightarrow h^{(x)} , \quad \hbc^{(y)} \leftrightarrow h^{(y)} ,  \quad \hbc^{(z)} \leftrightarrow h^{(z)} , 
\quad \hbc^{(x)}_{\betab_1} \leftrightarrow h^{(x)}_1 ,    \quad \hbc^{(y)} _{\betab_2}\leftrightarrow h^{(y)}_2 ,  \quad \hbc^{(z)}_{\betab_3} \leftrightarrow h^{(z)}_3\cr
&\&  \hbc_{\betab_1} \leftrightarrow h_1, 
   \ \hbc _{\betab_2}\leftrightarrow h_{2},  \ \hbc_{\betab_3} \leftrightarrow h_{3} ,
   \ \hbc _{\betab_1+\betab_2}\leftrightarrow h_{12},  \  \hbc_{\betab_1 + \betab_3} \leftrightarrow h_{13} ,
   \ \hbc_{\betab_2 + \betab_3} \leftrightarrow h_{23},
   \  \hbc_{\betab_1 + \betab_2 +\betab_3} \leftrightarrow h_{123}.
   \nonumber
\eea
} 
\begin{subequations}
\bea
\hbc\hbc^{(x)}\hbc_{\betab_1}^{(x)}&\&=\hbc\hbc_{\betab_1}\hbc_{\betab_2+\betab_3} +\hbc_{\betab_1}\hbc_{\betab_2}
\hbc_{\betab_3}+\hbc^{(x)}\hbc^{(y)}\hbc^{(z)},
\\
\hbc\hbc^{(y)}\hbc_{\betab_2}^{(y)}&\&=\hbc\hbc_{\betab_2}\hbc_{\betab_1+\betab_3} +\hbc_{\betab_1}\hbc_{\betab_2}\hbc_{\betab_3}+\hbc^{(x)}\hbc^{(y)}\hbc^{(z)},
\\
\hbc\hbc^{(z)}\hbc^{(z)}_{\betab_3}&\&=\hbc
_{\betab_3}\hbc_{\betab_1+\betab_2} +\hbc_{\betab_1}\hbc_{\betab_2}\hbc_{\betab_3}+\hbc^{(x)}\hbc^{(y)}\hbc^{(z)},\\
&\&{\hskip -46 pt}  \hbc^2\hbc^{(z)}\hbc^{(x)}\hbc^{(y)}\hbc_{\betab_1+\betab_2+\betab_3} = (\hbc^{(z)}\hbc^{(x)}\hbc^{(y)})^2\cr
&\&{\hskip -46 pt} + (\hbc^{(z)}\hbc^{(x)}\hbc^{(y)})
\left(2\hbc_{\betab_1}\hbc_{\betab_2}\hbc_{\betab_3}+\hbc(\hbc_{\betab_1} \hbc_{\betab_2+\betab_3}+\hbc_{\betab_2}\hbc_{\betab_1+\betab_3} +\hbc_{\betab_3}\hbc_{\betab_1+\betab_2})\right)\cr
&\&{\hskip -46 pt} +(\hbc_{\betab_1}\hbc_{\betab_2} +\hbc\hbc_{\betab_1+\betab_2})(\hbc_{\betab_1} \hbc_{\betab_3} +\hbc\hbc_{\betab_1+\betab_3})(\hbc_{\betab_2}\hbc_{\betab_3} +\hbc\hbc_{\betab_2+\betab_3}).\cr
&\&
\eea
\label{eq:hexahedronRecurrenceLagrangianGr36_intro}
\end{subequations}

By combining the Pl\"ucker relations  (\ref{eq:PluckerRelationIndices})  with the  Lagrangian condition (\ref{inf_lagrangian_cond}), 
the hexahedron relations  were shown in \cite{AHH}  to hold for a certain subset of the Pl\"ucker coordinates of any element 
of the Lagrangian Grassmannian  $\Gr^\LL_{\HH_+}(\HH, \omega)$.  The homomorphism of Theorem \ref{th:LatticeOfPointsInInfiniteLagrangianGrassmannian} implies  Corollary  \ref{cor:hexahedron_tau_eval}, which says that, 
for every element $w^0\in\Gr_{\HH_+}^\LL(\HH,\omega)$ of the infinite Lagrangian Grassmannian we get a solution of (\ref{eq:hexahedronRecurrenceLagrangianGr36_intro}) by normalized evaluation
of the $\tau$-function of CKP type on a suitably  defined sublattice, as localized formal power series in the parameters. 
Whenever this power series can be evaluated at a given set of values of the parameters, we obtain a 
 complex  valued solution of the hexahedron recurrence system (\ref{eq:hexahedronRecurrenceLagrangianGr36_intro}).

 \section{Nested finite dimensional Grassmannians}
 \label{sec:Finite-dimensional_Grassmannians}

\subsection{Finite dimensional Grassmannians and Pl\"ucker relations}
\label{sec:PluckerRelationsIndices}

Let
\be
 V _1 \ss V_2 \ss \cdots V _k \ss V_{k+1} \ss \cdots , \quad k\in \Nbb^+
\ee
be a nested sequence of $k$--dimensional  complex vector spaces, with $V_k$  spanned
by basis vectors $(e_{-1},\dots, e_{-k})$,  and $V_0$ being the zero dimension space containing as sole element the zero vector $\bf {0}$.
Denoting the corresponding dual spaces as $V_k^*$, the dual basis vectors as $(e_0,\dots,e_{k-1})$, for $k\ge 1$,
 with dual pairing
\be
e_{i} (e_{-j}) = (-1)^i\delta_{i, j-1} \quad i = 0, \dots, k-1, \ j = 1, \dots, k,
\label{eq:dualizing_basis}
\ee
we have
\be
 V ^*_1 \ss V^*_2 \ss \cdots V^* _k \ss V^*_{k+1} \ss \cdots. 
\ee
 For $n \in \Nbb$, $n\ge k$, denote  the direct sum of $V_k$ and $V^*_{n-k}$ by
\be
\HH_{k,n} := V_k \oplus V_{n-k}^*
\ee
and the Grassmannian of $k$-dimensional  subspaces of $\HH_{k,n}$; 
i.e., the homogeneous space of the general linear group $\GL(\HH_{k, n})$ consisting of the orbit of $V_k\ss \HH_{k,n}$,
by  $\Gr_{V_k}(\HH_{k,n})$.

Let $W$ be the $n \times k$ homogeneous coordinate matrix of an element $w \in \Gr_{V_k}(\HH_{k,n})$, with entries
$\{W_{ij}\}_{1\le i \le n \atop 1\le j \le k}$, whose column vectors determine
a basis $\{w_j \}_{1\le j \le k}$ for $w$ defined by
\be
w_j :=\sum_{i=1}^n W_{ij} e_{i-k-1}, \quad j=1, \dots, k.
\ee
For $k\in \Nbb^+$ and any $k$-tuple of indices
\be
L= (L_1, \dots, L_k), \quad-k\le L_i  \le n-k-1, \quad i=1, \dots, k. \quad L_i \in \Zbb
\ee
(not necessarily ordered or even distinct).
\begin{definition}
If all the $L_i$'s are distinct, let $\sgn(L)$ be the sign of the permutation that puts the $k$-tuple $L=(L_1,\dots,L_k)$
into increasing order and $\sgn(L):=0$ if $L$ contains repeated indices.
\end{definition}
Denote by $W_L$ the $k \times k$ submatrix whose $j$th row is the $(L_j+k+1)$th row of $W$, and define
\be
\tilde{\pi}_L(w) := \det(W_L)
\label{eq:piLDefDet}
\ee 
as its determinant. These are the {\em Pl\"ucker coordinates} of  $w \in \Gr_{V_k}(\HH_{k,n})$. They are defined
only within projective equivalence, since any change of basis multiplies them by the determinant of the basis change matrix.

The quantities $\{\tilde{\pi}_L(w) \}$ defined in (\ref{eq:piLDefDet}) satisfy the skew-symmetry conditions
\be
\widetilde{\pi}_{L_1,\dots,L_k} =\sgn(\sigma)\, \widetilde{\pi}_{L_{\sigma_1},\dots,L_{\sigma_k}},\quad\forall \ \sigma  \in S_k,
\label{eq:PluckerSkewSymmetryIndices}
\ee
and  hence, when $L$ has repeated indices,  $\widetilde{\pi}_{L}$ vanishes.
It follows from (\ref{eq:PluckerSkewSymmetryIndices}) that,  for  all $k-1$-tuples $(I_1,\dots,I_{k-1})$ and  $k+1$-tuples $(J_1,\dots,J_{k+1})$
of integer indices between $-k$ and $n-k-1$, $\{\widetilde\pi_L\}$ satisfy the Pl\"ucker relations \cite{GH}
\be
 p_{I,J}:= \sum_{j=1}^{k+1}(-1)^j \widetilde{\pi}_{I_1,\dots,I_{k-1},J_j} \widetilde{\pi}_{J_1,\dots,J_{j-1},\widehat{J}_j,J_{j+1},\dots J_{k+1}}=0,
\label{eq:PluckerRelationIndices}
\ee
where $\hat{J}_j$ denotes omission of the element $J_j$ from the sequence. The $p_{I,J}$'s will be  called the {\em Pl\"ucker 
quadratic forms}.

Alternatively, viewing the symbols $\{\tilde{\pi}_L \}$  as indeterminates, define for all pairs $k, n  \in \Nbb^+$,
$k < n$, the free polynomial ring
\begin{equation}
\widetilde{\SS}_{k,n}:=\Cbb[\widetilde{\pi}_{L_1,\dots,L_k}\;|\;- k \le L_i\le n-k-1, \ i\in \{1, \dots , k\}],
\end{equation}
in $n^k$ indeterminates $\{\widetilde{\pi}_{L}\}$ labeled by $k$-element multi-indices $L=(L_1, \dots, L_k)$ 
(not necessarily ordered or distinct) with elements in the range indicated.
Let $\widetilde{\II}_{k,n}\subset\widetilde{\SS}_{k,n}$ denote the ideal generated
by the skew-symmetry conditions (\ref{eq:PluckerSkewSymmetryIndices}) and
the Pl\"ucker relations (\ref{eq:PluckerRelationIndices}).
Define a $\Nbb\times\Nbb$ bi-grading on $\widetilde{\SS}_{k,n}:$ by
\be
\deg(\widetilde \pi_{L_1,\dots,L_k}):=\Big(1,\,\sum_{i=1}^kL_i+\frac{k(k+1)}2\Big).
\label{eq:BihomogeneousDegreePiIndices}
\ee
For a monomial in the generators $\{\tilde{\pi}_L\}$, the first component in (\ref{eq:BihomogeneousDegreePiIndices}) 
is its degree, while the second is the weight $|\lambda(L)|$ of the associated partition $\lambda(L)$ defined 
in eq.~(\ref{eq:FiniteLThroughLambda}) below. It follows from the skew-symmetry relations (\ref{eq:PluckerSkewSymmetryIndices}) and
the Pl\"ucker relations (\ref{eq:PluckerRelationIndices}) that $\widetilde{\II}_{k,n}$ is bihomogeneous with respect 
to this  bi-grading. The coordinate ring of the Grassmannian $\Gr_{V_k}(\HH_{k,n})$ thus has the following presentation
in terms of generators and relations
\be
\widetilde{\SS}_{k,n}(\Gr_{V_k}(\HH_{k,n})):=\widetilde{\SS}_{k,n}\Big/\,\widetilde{\II}_{k,n},
\ee
and the quotient ring is also bi-graded.

\subsection{Pl\"ucker map and coordinates $\pi_\lambda$ labelled by partitions}
\label{plucker_coords_partitions}

Geometrically,  the Pl\"ucker map is defined by
\bea
\grP^n_k: \Gr_{V_k}(\HH_{k,n}) &\&\ra \Pb(\Lambda^k\HH_{k,n}) \cr
\grP^n_k: w &\&\mapsto [w_1 \wedge \cdots \wedge w_k]
\label{eq:plucker_map_kn}
\eea
(where $[\upsilon]$ denotes the projective equivalence class of $\upsilon\in \Lambda^k(\HH_{k,n})$). This embeds
$\Gr_{V_k}(\HH_{k,n})$ equivariantly, as a projective variety, into the projectivization $\Pb(\Lambda^k(\HH_{k,n})) $ of
the $k$th exterior power of $\HH_{k,n}$.

 Because of the skew-symmetry condition (\ref{eq:PluckerSkewSymmetryIndices}),
it is sufficient to only  consider $\tilde{\pi}_L$ with $L$ ordered increasingly and no repeated indices.
To any isuch ncreasingly ordered $k$-tuple $L$ there is a uniquely associated partition $\lambda$
whose parts are given by
\be
\lambda_i :=L_{k-i+1}+i, \quad i=1, \dots k,
\label{eq:FiniteLThroughLambda}
\ee
and whose Young diagram fits into the $k \times (n-k)$ rectangular diagram.

\begin{definition} Denote by
\be
\mathbf\Lambda_{k,n}=\left\{\lambda\;\big|\; \lambda\subset (n-k)^k\right\}
\ee
 the  set of partitions whose Young diagrams fit into the $k \times (n-k)$ rectangle.
\end{definition}
 Relative to the basis $\{|\lambda\rangle\}_{ \lambda \in \mathbf\Lambda_{k,n}}$ for $\Lambda^k\HH_{k,n}$ defined by
\be
|\lambda\rangle := e_{L_k} \wedge \cdots \wedge e_{L_1},\quad\lambda\in\mathbf\Lambda_{k,n},
\ee
these may alternatively be labelled by the corresponding partitions $\{\lambda\}$ defined in (\ref{eq:FiniteLThroughLambda}):
\be
\pi_\lambda(w) := \tilde{\pi}_{L_1,\dots,L_k}(w).
\label{eq:piTildepiLambda}
\ee
The image of $w$ under the Pl\"ucker map is then given by
\be
\grP^n_k(w) = \big[ {\hskip -0.3em}\sum_{\lambda\in\mathbf\Lambda_{k,n}} \pi_{\lambda}(w) |\lambda\rangle\,\big].
\ee
and the Pl\"ucker relations are equivalent to the decomposability of the element $\grP^n_k(w)$ defined
by (\ref{eq:plucker_map_kn}).
In terms of $\{\pi_\lambda\}$, the bi-grading (\ref{eq:BihomogeneousDegreePiIndices}) is simply
\begin{equation}
\deg\, \pi_\lambda=(1,|\lambda|) \in\Nbb\times\Nbb,
\label{eq:BidegreeOfPiLambda}
\end{equation}
where
\be
|\lambda| = \sum_{i=1}^k \lambda_i
\label{weight_lambda}
\ee
is the weight of $\lambda $.

We thus have an equivalent presentation of the coordinate ring
 $\SS_{k,n}(\Gr_{V_k}(\HH_{k,n}))$ in terms of Pl\"ucker coordinates
 labelled by partitions.
Denote by
\begin{equation}
\SS_{k,n}:=\Cbb[\pi_\lambda\;|\;\lambda\in\mathbf\Lambda_{k,n}]
\end{equation}
the free polynomial ring generated by indeterminates $\{\pi_\lambda\}$ labelled by partitions $\lambda \in  \mathbf\Lambda_{k,n}$,
and let $\II_{k,n}\subset\SS_{k,n}$ be the ideal generated by the Pl\"ucker relations (\ref{eq:PluckerRelationIndices}),
written in terms of  such $\{\pi_\lambda\}$'s. The isomorphism
\begin{equation}
 \SS_{k,n}(\Gr_{V_k}(\HH_{k,n})):=\SS_{k,n}\big/\II_{k,n} \;\cong\;\widetilde{\SS}_{k,n}\big/\widetilde{\II}_{k,n} =:\widetilde{\SS}_{k,n}(\Gr_{V_k}(\HH_{k,n})),
\end{equation}
of these $\Nbb\times\Nbb$-graded rings then follows from (\ref{eq:PluckerSkewSymmetryIndices})
and  (\ref{eq:piTildepiLambda}).

\subsection{Embeddings and partial projections of Grassmannians}
\label{embeddings_partial_proj}

Define two sequences of embeddings 
\bea
i_k: \Gr_{V_k}(\HH_{k,n}) &\&  \hra \Gr_{V_{k+1}}(\HH_{k+1,n+1})\cr
i_k: w&\&\mapsto e_{-k-1}\oplus w\;\subset\HH_{k+1,n+1},\\
i^n: \Gr_{V_k}(\HH_{k,n}) &\&  \hra \Gr_{V_k}(\HH_{k,n+1})\cr
i^n(w)&\&\mapsto w\;\subset \HH_{k,n+1}.
\eea
The diagrams
\be
\begin{tikzcd}
\Gr_{V_k}(\HH_{k,n}) \ar[rr,hookrightarrow,"i^n"] \ar[dd,hookrightarrow,"i_k"] && \Gr_{V_k}(\HH_{k,n+1})  \ar[dd,hookrightarrow,"i_k"] \\\\
\Gr_{V_{k+1}}(\HH_{k+1,n+1}) \ar[rr,hookrightarrow,"i^{n+1}"]   &&  \Gr_{V_{k+1}}(\HH_{k+1,n+2})
\end{tikzcd}
\label{eq:Embeddings-ik-in}
\ee
 are  tautologically commutative.

At the level of coordinate rings, the pull-backs of these embeddings
define commuting diagrams of surjective homomorphisms of coordinate rings

\begin{equation}
\begin{tikzcd}
\SS_{k,n}(\Gr_{V_k}(\HH_{k,n}))&& \SS_{k,n+1}(Gr_{V_k}(\HH_{k,n+1}))\ar[ll,twoheadrightarrow,"(i^n)^*"]\\\\
\SS_{k+1,n+1}(\Gr_{V_{k+1}}(\HH_{k+1,n+1}))\ar[uu,twoheadrightarrow,"(i_k)^*"]&& \SS_{k+1,n+2}(\Gr_{V_{k+1}}(\HH_{k+1,n+2})) \ar[ll,twoheadrightarrow,"(i^{n+1})^*"] \ar[uu,twoheadrightarrow,"(i_k)^*"]
\end{tikzcd}
\label{S_kn_commut_diag}
\end{equation}
where  $(i_k)^*$ and $(i^n)^*$ act on the Pl\"ucker coordinates as follows.
\begin{subequations}
\begin{eqnarray}
(i_k)^*(\pi_\lambda)&=&\left\{\begin{array}{cl}
\pi_\lambda,&\lambda\in\mathbf\Lambda_{k,n},\cr
0,&\textrm{otherwise},
\end{array}\right.
\qquad\textrm{for all}\;\lambda\in\mathbf\Lambda_{k+1,n+1},
\label{eq:ikstarAction}\\[0.5em]
(i^n)^*(\pi_\lambda)&=&\left\{\begin{array}{cl}
\pi_\lambda,&\lambda\in\mathbf\Lambda_{k,n},\cr
0,&\textrm{otherwise},
\end{array}\right.
\qquad\textrm{for all}\;\lambda\in\mathbf\Lambda_{k,n+1}.
\label{eq:instarAction}
\end{eqnarray}
\label{eq:istarAction}
\end{subequations}

The generators $\{\pi_\lambda\}$ of the ring $\SS_{k,n}$ are
a subset of those of the two larger rings $\SS_{k,n+1}$ and $\SS_{k+1,n+1}$.
Denoting  the tautological injection maps as
\bea
\iota^n:\SS_{k,n}&\& \ra \SS_{k, n+1}, \qquad \iota_k:\SS_{k,n} \ra \SS_{k+1, n+1}\cr
\iota^n(\pi_\lambda) &\& \mapsto \pi_\lambda  {\hskip 56 pt}
\iota_k(\pi_\lambda)  \mapsto \pi_\lambda,
\eea
this determines a commutative diagram of injective ring homomorphisms
\begin{equation}
\begin{tikzcd}
\SS_{k,n}\ar[r,hookrightarrow,"\iota^n"]\ar[d,hookrightarrow,"\iota_k"]&\SS_{k,n+1} \ar[d,hookrightarrow,"\iota_k"]\\
\SS_{k+1,n+1}\ar[r,hookrightarrow,"\iota^{n+1}"] &  \SS_{k+1,n+2}
\end{tikzcd}
\end{equation}

Since the Pl\"ucker relations that generate $\II_{k,n}$ are also satisfied in  the larger Grassmannians, we have two inclusions:
\begin{lemma}
\begin{equation}
\iota_k(\II_{k,n})\subset\II_{k+1,n+1},\qquad \iota^n(\II_{k,n})\subset\II_{k,n+1}.
\label{eq:InclusionOfIkn}
\end{equation}
\label{lemm:InclusionOfIkn}
\end{lemma}
The homomorphisms $\iota_k$ and $\iota^n$ project to the corresponding quotient rings,
providing homomorphisms between the coordinate rings of the Grassmannians. It follows from the definitions that
\begin{equation}
(i_k)^*\circ\iota_k= \mathrm{Id}_{\SS_{k,n}(\Gr_{V_k}(\HH_{kn}))} =(i^n)^*\circ\iota^n,
\end{equation}
and therefore:
\begin{corollary}
For every pair of integers $(k,n)$ with $0\leq k\leq n$, the following is a commutative diagram of injective ring homomorphisms
\begin{equation}
\begin{tikzcd}
\SS_{k,n}(\Gr_{V_k}(\HH_{k,n})) \ar[dd,hookrightarrow,"\iota_k"] \ar[rr,hookrightarrow,"\iota^n"]&& \SS_{k,n+1}(Gr_{V_k}(\HH_{k,n+1})) \ar[dd,hookrightarrow,"\iota_k"]\\\\
\SS_{k+1,n+1}(\Gr_{V_{k+1}}(\HH_{k+1,n+1})) \ar[rr,hookrightarrow,"\iota^{n+1}"]&& \SS_{k+1,n+2}(\Gr_{V_{k+1}}(\HH_{k+1,n+2}))
\end{tikzcd}
\label{eq:Homomorphisms_iota_CommutativeDiagram}
\end{equation}
\end{corollary}

\subsection{Big cell}
\label{affine_big_cell}

Relative to the standard complete flag
\be
\span\{e_{-k}\}   \ss  \span\{e_{-k}, e_{-k+1}\} \ss  \cdots  \ss \span\{e_{-k}, \dots, e_{n-1}\},
\ee
the {\em big cell} $\Gr^\emptyset_{V_k}(\HH_{k,n})\ss \Gr_{V_k}(\HH_{k,n})$
is  defined by the condition that the Pl\"ucker coordinate $\pi_\emptyset$
corresponding to the trivial partition is nonzero:
\begin{equation}
\Gr^\emptyset_{V_k}(\HH_{k,n}):=\left\{w^\emptyset\in \Gr_{V_k}(\HH_{k,n})\;\big|\;\pi_{\emptyset}(w^\emptyset)\neq 0\right\}.
\end{equation}
Equivalently, it consists of all elements $w^\emptyset\in \Gr_{V_k}(\HH_{k,n})$ that can be represented as
 the graph of a map $M(w^\emptyset): V_k \ra V^*_{n-k}$.
\begin{lemma}
Every element $w^\emptyset\in \Gr^\emptyset_V(\HH_{k,n})$ can be uniquely represented
as the image of the linear  injection  $\epsilon_{w^\emptyset}: V_k \hra \HH_{k,n}$ defined by
\be
\epsilon_{w^\emptyset}: e_{-j} \mapsto e_{-j}+\sum_{i=1}^{n-k}M_{ij}(w^\emptyset)(-1)^{i-1}e_{i-1} \ \textrm{for all}\, 1\leq j\leq k,
\label{eq:affine_basis_graph}
\ee
where  $\{M_{ij}(w^\emptyset)\}_{i=1, \cdots n-k, j=1, \dots, k}$ are the matrix elements
of the map $M(w^\emptyset)$ relative to the basis $(e_{-k}, \cdots e_{-1})$ for $V_k$  and the standard
dual basis $(e^*_{-1}, \dots , e^*_{-n+k})$ for $V^*_{n-k}$.
\label{lemm:CoordinatesOnTheBigCellFinite}
\end{lemma}
\begin{proof}
It follows from formulae (\ref{eq:piLDefDet}), (\ref{eq:FiniteLThroughLambda}) and (\ref{eq:piTildepiLambda})
that $\pi_\emptyset$ corresponds to the determinant of the first $k\times k$ block of the homogeneous coordinate matrix
corresponding to any basis.
Multiplying the homogeneous coordinate matrix on the right by the inverse of this block
places the identity matrix in the top  $k\times k$ block, and the matrix appearing as
the coefficients in the sum in (\ref{eq:affine_basis_graph}), which is the affine coordinate matrix,
as the remaining $n-k \times k$ block. (The sign factors $(-1)^{j-1}$  follow from the
dual basis labelling convention (\ref{eq:dualizing_basis}).)
\end{proof}

The matrix elements $\{M_{ij}(w^\emptyset)\}_{i=1, \cdots n-k,\,  j=1, \dots, k}$
are uniquely expressed in terms of the Pl\"ucker coordinates corresponding to {\em hook partitions} $[i,1^{j-1}]=(i-1|j-1)$
as follows
\begin{equation}
M_{ij}(w^\emptyset)=\frac{\pi_{(i-1|j-1)}}{\pi_\emptyset ^\emptyset} (w^\emptyset)\qquad\textrm{for}\qquad 1\leq i\leq n-k,\quad 1\leq j\leq k.
\label{eq:AffineCoordinatesGrkn}
\end{equation}
By the generalized Giambelli identity (\cite{HB},  Sec.~4.10 and App.~C.8), the remaining Pl\"ucker coordinates
$\pi_{\lambda}(w^\emptyset)$ on the big cell,
corresponding to partitions $\lambda=(\ab |\bb)$ of Frobenius rank $r>1$,  can be expressed as determinants
of  $r \times r$ submatrices of the affine coordinate matrix via the formula  (\cite{HB},  App.C.8)
\begin{equation}
\frac{\pi_{\lambda}}{\pi_\emptyset}(w^\emptyset) =\det\left(\frac{\pi_{(a_i|b_j)}}{\pi_\emptyset}(w^\emptyset)\right)_{1\leq i,j\leq r}.
\label{eq:DeterminantFormulaPluckerCoordinatesFinite}
\end{equation}

The  big cell $\Gr^\emptyset_{V_k}(\HH_{k,n})$ may be viewed
as an affine variety with  polynomial coordinate ring generated by the affine coordinate matrix elements
\begin{equation}
\OO(\Gr^\emptyset_{V_k}(\HH_{k,n}))=\Cbb\big[M_{ij}\;\big|\; 1\leq i\leq n-k,\;1\leq j\leq k\big].
\end{equation}
\begin{definition}
For  $k, n \in \Nbb^+$, with $0\leq k\leq n$,  define the homomorphism
\be
\xi_{k,n}:\SS_{k,n}\rightarrow \Cbb[\pi_\emptyset,\pi_\emptyset^{-1}]\big[\pi_{(i-1|j-1)}\,\big|\,1\leq i\leq n-k,\;1\leq j\leq k\big],
\ee
by its values on the generators of the polynomial ring  $\SS_{k,n}$
\be
\xi_{k,n}:\pi_\lambda\mapsto\pi_\emptyset\,\det\left(\frac{\pi_{(a_i|b_j)}}{\pi_\emptyset}\right)_{1\leq i,j\leq r},
\label{eq:XiknMapDefinition}
\ee
where $r$ is the Frobenius rank of the partition $\lambda=(\ab|\bb)$.
\label{def:XiknMap}
\end{definition}

\begin{lemma}
$\xi_{k,n}$ annihilates all Pl\"ucker quadratic forms $\{p_{I,J}\}$ in the coordinate ring $\SS_{k,n}$, so
\begin{equation}
\II_{k,n}\subset\ker \xi_{k,n}.
\end{equation}
\label{lemm:XiknAnnihilatesPlucker}
\end{lemma}
\begin{proof}
Viewing the big cell as an open submanifold of the full Grassmannian
\begin{equation}
\Gr_{V_k}^\emptyset(\HH_{k,n}) \hookrightarrow \Gr_{V_k}(\HH_{k,n})
\end{equation}
gives rise to a surjective homomorphism of their coordinate rings
\begin{equation}
\begin{aligned}
\zeta_{k,n}:\SS_{k,n}(\Gr_{V_k}(\HH_{k,n}))& \ra \; \OO(\Gr_{V_k}^\emptyset(\HH_{k,n}))\\
\pi_\lambda&\mapsto\; \det M_{(\ab |\bb)},
\label{eq:MapToCoordinateRingOfTheAffineBigCellFinite}
\end{aligned}
\end{equation}
where $M_{(\ab|\bb)}$ denotes the  submatrix of the affine coordinate matrix (\ref{eq:AffineCoordinatesGrkn})
with row set $(a_1,\dots,a_r)$ and column set $(b_1,\dots,b_r)$ given by Frobenius indices of
the partition $\lambda=(\ab|\bb)$. By construction, the homomorphism $\zeta_{k,n}$
 annihilates all Pl\"ucker relations generating $\II_{k,n}$.

Formula (\ref{eq:AffineCoordinatesGrkn}) also determines an injective ring homomorphism
\begin{equation}
\chi_{k,n}:\quad\OO(\Gr_{V_k}^\emptyset(\HH_{k,n}))\hookrightarrow \Cbb[\pi_\emptyset,\pi_\emptyset^{-1}][\pi_{(i-1|j-1)}\,|\,1\leq i\leq n-k,\;1\leq j\leq k]
\label{eq:CoordinateRingOfBigCellToLocalizationFinite}
\end{equation}
generated by
\be
\chi_{k,n}: M_{ij} \mapsto \frac{\pi_{(i-1|j-1)}} {\pi_\emptyset}.
\ee
Up to the common $\pi_\emptyset$ factor,  the composite map $\chi_{k,n}\circ \zeta_{k, n}$ acts on the generators $\pi_\lambda$ 
in exactly the same way as $\xi_{k,n}$ in (\ref{eq:XiknMapDefinition}). 
Therefore all Pl\"ucker quadratic forms are annihilated by the map $\xi_{k, n}$.
\end{proof}

Similarly to the inclusions (\ref{eq:Homomorphisms_iota_CommutativeDiagram}) of coordinate rings,
we have a system of inclusions of coordinate rings of the corresponding affine big cells
\begin{subequations}
\bea
\tilde\iota_k:&\&\ \OO(\Gr_{V_k}^\emptyset(\HH_{k,n}))\rightarrow\OO(\Gr_{V_{k+1}}^{\emptyset}(\HH_{k+1,n+1})),\\
\tilde\iota^n:&\&\ \OO(\Gr_{V_k}^\emptyset(\HH_{k,n}))\rightarrow\OO(\Gr_{V_k}^\emptyset(\HH_{k,n+1})),
\eea
\label{eq:TildeIota_kn_definition}
\end{subequations}
where
\begin{equation}
\tilde\iota_k(M_{ij})=M_{ij}=\tilde\iota^n(M_{ij})\quad\forall \ (i,j) \mid \ 1\leq i\leq n-k, \ 1\leq j\leq k
\end{equation}
defines a pair of inclusions taking the affine coordinate matrix into  the larger ones.
\begin{lemma}
For every pair of integers $(k,n)$ with  $0\leq k\leq n$, we have the following commutative 
cubic diagram of ring homomorphisms.
\begin{equation}
\adjustbox{scale=0.95,center}{
\begin{tikzcd}[nodes=overlay,row sep=4em, column sep=10em]
&\OO(\Gr_{V_k}^\emptyset(\HH_{k,n}))\ar[rr,hookrightarrow,"\tilde\iota^n"] \ar[dd,hookrightarrow,"\tilde\iota_k" near start]
&&\OO(\Gr_{V_k}^\emptyset(\HH_{k,n+1})) \ar[dd,hookrightarrow,"\tilde\iota_k"]\\
\SS_{k,n}(\Gr_{V_k}(\HH_{k,n})) \ar[dd,hookrightarrow,"\iota_k"] \ar[rr,hookrightarrow,pos=0.35,"\iota^n",crossing over] \ar[ru,twoheadrightarrow,"\zeta_{k,n}"]&& \SS_{k,n+1}(\Gr_{V_k}(\HH_{k,n+1})) \ar[ru,twoheadrightarrow,"\zeta_{k,n+1}"']\\
&\OO(\Gr_{V_{k+1}}^\emptyset(\HH_{k+1,n+1}))\ar[rr,hookrightarrow,"\iota^{n+1}" pos=0.3]&&\OO(\Gr_{V_{k+1}}^\emptyset(\HH_{k+1,n+2}))\\
\SS_{k+1,n+1}(\Gr_{V_{k+1}}(\HH_{k+1,n+1})) \ar[rr,hookrightarrow,"\iota^{n+1}"] \ar[ru,twoheadrightarrow,"\zeta_{k+1,n+1}" pos=0.6]&& \SS_{k+1,n+2}(\Gr_{V_{k+1}}(\HH_{k+1,n+2})) \ar[ru,twoheadrightarrow,pos=0.6,"\zeta_{k+1,n+2}"'] \ar[from=uu,hookrightarrow,near start,"\iota_k",crossing over]
\end{tikzcd}
}
\label{eq:ThreeDimensionalDiagramInclusions}
\end{equation}
\label{lemm:BigCellInclusionComaptibility}
\end{lemma}
\begin{proof}
Take any partition $\lambda=(\ab|\bb)\in\mathbf\Lambda_{k,n}$ and follow the path of the corresponding generator $\pi_\lambda\in\SS_{k,n}$ along any face of the above cubic diagram.
\end{proof}

\subsection{Nested finite dimensional Lagrangian Grassmannians}
\label{nested_finite_lagrange}

Let $k\in\Nbb^+$ be a positive integer. Define the {\em Lagrangian Grassmannian}
\begin{equation}
\Gr_{V_k}^{\LL}(\HH_{k,2k},\omega_k)\subset \Gr_{V_k}(\HH_{k,2k})
\end{equation}
as the subvariety consisting of maximal isotropic subspaces of the
$2k$-dimensional space $\HH_{k,2k}$ with respect to the symplectic form
\be
\omega_k:=\omega|_{\HH_{k,2k}} = \sum_{i=0}^{k-1}(-1)^i  e^*_{i}\wedge e^*_{-i-1}
\ee

Equivalently, denote by
\be
\hat{\omega}_{k} = \sum_{i=0}^{k-1}(-1)^i\psi_{i}\psi_{-i-1}\ \in \End(\Lambda(\HH_{k,2k})),
\label{eq:OperatorOmega_{2k}}
\ee
where
\be
 \psi_i := e_i\wedge \in \End(\Lambda(\HH_{k,2k})), \quad \psi^\dag_i := i_{e^*_i} \in \End(\Lambda(\HH_{k, 2k})), \quad j \in \Zbb,
\ee
the quadratic Clifford algebra element obtained by taking the exterior product with $\omega_k$:
\bea
   \hat{\omega}_{k}:\Lambda(\HH_{k,2k})&\& \ra \Lambda(\HH_{k,2k}) \cr
   \hat{\omega}_{k} :  \mu &\& \mapsto \omega_{k}\swedge \mu.
   \label{eq:OperatorOmega_2k}
\eea
Then $w^0$ is in the Lagrangian Grassmannian $ \Gr_{V_k}^{\LL}(\HH_{k,2k}) $ if and only if its image $\grP^{2k}_k(w^0)$
under the Pl\"ucker map is in the kernel of the restriction of  $\hat{\omega}_{2k}$ to $\Lambda^k(V_k)$
\begin{equation}
\grP_{k,2k}(w^0)=\big[\sum_{\lambda\in\mathbf\Lambda_{k,2k}}\pi_\lambda|\lambda\rangle\big] \in \ker\hat\omega_k\big|_{\Lambda^k(\HH_{k,2k})}.
\label{eq:LagrangianConditionKernelOmega}
\end{equation}
Condition (\ref{eq:LagrangianConditionKernelOmega}) is equivalent
to imposing the system of linear relations
\be
  \hat\omega_k\big(\sum_{\lambda\in\mathbf\Lambda_{k,2k}}\pi_\lambda|\lambda\rangle\big)=0,
  \label{eq: lagrange_cond_finite}
\ee
 on the Pl\"ucker coordinates, one for each basis element of $\hat\omega(\Lambda^{k-2}(\HH_{k,2k}))$.
(See \cite{AHH} , Proposition 2.3 for more details.)

Inclusion of the Lagrangian Grassmannian $\Gr_{V_k}^\LL(\HH_{k,2k},\omega_k)$ in $ \Gr_{V_{k+1}}^\LL(\HH_{k+1,2k+2},\omega_{k+1})$
 is compatible with the corresponding inclusion of the full Grassmannians.
 For every $k\in\Nbb^+$ we thus have a commutative diagram of inclusions
\begin{equation}
\begin{tikzcd}
\Gr_{V_k}^\LL(\HH_{k,2k},\omega_k)\ar[rr,hookrightarrow,"i_k^\LL"]\ar[dd,hookrightarrow,"\ell_k"]&& \Gr_{V_{k+1}}^\LL(\HH_{k+1,2k+2},\omega_{k+1})\ar[dd,hookrightarrow,"\ell_{k+1}"]\\\\
\Gr_{V_k}(\HH_{k,2k})\ar[rr,hookrightarrow,"i_k^{2k}"]&&\Gr_{V_{k+1}}(\HH_{k+1,{2k+2}})
\end{tikzcd}
\end{equation}
where
 \be
 i_k^{2k}:=i_k\circ i^{2k}
 \ee
 is defined by composition,
$i_k^\LL$ is its restriction to $\Gr_{V_k}^\LL(\HH_{k,2k}, \omega_k) \ss \Gr_{V_k}(\HH_{k, 2k})$, 
\be
i_k^\LL := i_k^{2k}\big|_{\Gr_{V_k}^\LL(\HH_{k,2k}, \omega_k)},
\ee
and
\be
\ell_k:\Gr_{V_k}^\LL(\HH_{k,2k},\omega_k) \hookrightarrow \Gr_{V_k}(\HH_{k,2k},\omega_k)
\ee
 is inclusion. 
At the level of homogeneous coordinate rings, by (\ref{eq:istarAction}), we get
\bea
\left(i_k^\LL\right)^*:\mathcal S_{k+1,2k+2}(\Gr_{V_{k+1}}^\LL(\HH_{k+1,2k+2},\omega_{k+1}))&\&\twoheadrightarrow \mathcal S_{k,2k}(\Gr_{V_k}^\LL(\HH_{k,2k},\omega_k))\cr
\left(i_k^\LL\right)^*:(\pi_\lambda)&\& \mapsto
\begin{cases}
\pi_\lambda \quad \text{ if } \lambda\in \mathbf{\Lambda}_{k,2k},\\
0 \quad  \ \,  \text{ if } \lambda \notin \mathbf{\Lambda}_{k,2k},
\end{cases}\\\
\forall \;\lambda &\&\in\mathbf{\Lambda}_{k+1,2k+2}.
\nonumber
\eea

Denote the ideal of $\mathcal S_{k,2k}$ generated by the Pl\"ucker relations and
Lagrange relations (\ref{eq: lagrange_cond_finite})  as $\II_k^\LL$.
The Lagrangian Grassmannian is a projective variety with homogeneous coordinate ring
\begin{equation}
\SS_{k,2k}(\Gr_{V_k}^{\LL}(\HH_{k,2k},\omega))=\mathcal S_{k,2k}/\II_k^\LL.
\label{eq:IkLDefinition}
\end{equation}
For $k \in \Nbb^+$, define the injective maps
$\iota_k^{2k}:\SS_{k,2k}\rightarrow\SS_{k+1,2k+2}$ by composition:
\be
\iota_k^{2k}:= \iota_k \circ \iota^{2k}.
\ee
\begin{lemma}
The defining ideal of (\ref{eq:IkLDefinition}) is preserved under the inclusions of the polynomial rings
\begin{equation}
\iota_k^{2k}:\SS_{k,2k}\rightarrow\SS_{k+1,2k+2},\quad \iota_k^{2k}(\II_k^\LL)\subset\II_{k+1}^\LL,
\label{iota_k_def}
\end{equation}
which induces a homomorphism of the quotient rings
\begin{equation}
\iota_{k}^\LL:\mathcal S_{k,2k}(\Gr_{V_k}^\LL(\HH_{k,2k},\omega_k)) \hookrightarrow\mathcal S_{k+1,2k+2}(\Gr_{V_{k+1}}^\LL(\HH_{k+1,2k+2},\omega_{k+1}))
\label{eq:CoordinateRingInclusionLagrangianFinite}
\end{equation}
satisfying
\begin{equation}
\left(i_k^\LL\right)^*\circ\iota_k^\LL=\mathrm{Id}_{\mathcal S_{k,2k}(\Gr_{V_k}^\LL(\HH_{k,2k}),  \omega_k)}.
\end{equation}
\end{lemma}

The big cell of the Lagrangian Grassmannian is defined as the intersection
\begin{equation}
\Gr_{V_k}^{\LL,\emptyset}(\HH_{k,2k},\omega_k):= \Gr_{V_k}^{\LL}(\HH_{k,2k},\omega_k)\cap\Gr_{V_k}^{\emptyset}(\HH_{k,2k}).
\end{equation}
On $\Gr_{V_k}^{\LL,\emptyset}(\HH_{k,2k},\omega_k)$, the Lagrangian condition (\ref{eq: lagrange_cond_finite}) 
 is equivalent to the the affine coordinate matrix being symmetric:
 \begin{equation}
M_{ij}=M_{ji}\quad 1\leq i, j\leq k.
\label{eq:LagrangianConditionBigCellFinite}
\end{equation}
The big cell  $\Gr_{V_k}^{\LL,\emptyset}(\HH_{k,2k},\omega_k)$ is 
thus the same as the variety of symmetric $k\times k$ matrices, with polynomial coordinate ring
\begin{equation}
\mathcal O(\Gr_{V_k}^{\LL,\emptyset}(\HH_{k,2k},\omega_k)=\Cbb[M_{ij},\,|\,1\leq i\leq j\leq k].
\end{equation}
Similarly to (\ref{eq:MapToCoordinateRingOfTheAffineBigCellFinite}), for all $k\in\Nbb^+$ we have a surjective ring homomorphism
\be
\zeta_k^\LL:\SS_{k,2k}(\Gr_{V_k}^\LL(\HH_{k,2k},\omega_k))\twoheadrightarrow\mathcal O(\Gr_{V_k}^{\LL,\emptyset}(\HH_{k,2k},\omega_k)),
\ee
generated by
\be
\zeta_k^\LL(\pi_\lambda)=\det\Big(M_{\min(a_i,b_j), \max(a_i,b_j)}\Big)_{1\leq i,j\leq r},
\ee
where $r=r(\lambda)$ is the Frobenius rank of partition $\lambda\in\mathbf\Lambda_{k,2k}$.

For $n=2k$, the diagonal slice of diagram (\ref{eq:ThreeDimensionalDiagramInclusions}) bounded by
the upper left edge and  lower right edge is mapped surjectively to its Lagrangian counterpart. 
We therefore have
\begin{lemma}
For all $k\in\Nbb^+$ the following diagram is commutative
\begin{equation}
\adjustbox{scale=0.95,center}{
\begin{tikzcd}[nodes=overlay,row sep=4.2em, column sep=10em]
&\mathcal O(\Gr_{V_k}^\emptyset(\HH_{k,2k})) \arrow[rr,hookrightarrow,
"\tilde\iota_k^{2k}"] \arrow[dd,twoheadrightarrow,pos=0.22,"\tilde\ell_k^*"] &&\mathcal O(Gr_{V_{k+1}}^\emptyset(\HH_{k+1,2k+1})) \arrow[dd,twoheadrightarrow,pos=0.22,"\tilde\ell_{k+1}^*"]\\
\mathcal S_{k,2k}(\Gr_{V_k}(\HH_{k,2k})) \arrow[rr,crossing over,hookrightarrow,pos=0.3,"\iota_k^{2k}"] \arrow[ru,twoheadrightarrow,"\zeta_{k,2k}"] \arrow[dd,twoheadrightarrow,pos=0.22,"\ell_k^*"] &&\mathcal S_{k+1,2k+1}(\Gr_{V_{k+1}}(\HH_{k+1,2k+2})) \arrow[ru,twoheadrightarrow,"\zeta_{k+1,2k+2}"]\\
&\mathcal O(\Gr_{V_k}^{\mathcal L,\emptyset}(\mathcal H_{k,2k},\omega_k)) \arrow[rr,hookrightarrow,pos=0.3,"\tilde\iota_k^\LL"]&& \mathcal O(\Gr_{V_{k+1}}^{\mathcal L,\emptyset}(\mathcal H_{k+1,2k+2},\omega_{k+1}))\\
\mathcal S_{k,2k}(\Gr_{V_k}^\LL(\HH_{k,2k},\omega_k)) \arrow[ru,twoheadrightarrow,"\zeta_k^\LL"] \arrow[rr,hookrightarrow,
"\iota_k^\LL"]&& \mathcal S_{k+1,2k+2}(\Gr_{V_{k+1}}^{\LL}(\HH_{k+1,2k+2},\omega_{k+1})) \arrow[ru,twoheadrightarrow,pos=0.45,"\zeta_{k+1}^\LL"] \arrow[uu,twoheadleftarrow,crossing over,pos=0.78,"\ell_{k+1}^*"']
\end{tikzcd}
}
\label{eq:ThreeDimensionalDiagramLagrangian}
\end{equation}
\label{lemm:ThreeDimensionalDiagramLagrangian}
\end{lemma}

\section{Infinite Grassmannians and direct limits}
\label{sec:Infinite_Grassmannians_KP_tau}

 We start by defining the homogeneous coordinate ring $\SS(\Gr_{\HH_+}(\HH))$ of the infinite
 dimensional Grassmannian $\Gr_{\HH_+}(\HH)$ as a direct limit of finite dimensional ones, and  view
 $\Gr_{\HH_+}(\HH)$ as the homogeneous spectrum of $\SS(\Gr_{\HH_+}(\HH))$.
 
\subsection{Direct limits of coordinate rings}

Recall first some standard definitions regarding {\em directed sets}, {\em direct systems} and {\em direct limits} \cite{Maclane}.

\begin{definition}
Let $\grI$ be a set and $\preceq\,\subset\,\grI\times\grI$  a binary relation on $\grI$. A pair $(\grI,\preceq)$
is called a \textit{directed set} if it has the following properties
\begin{enumerate}
\item  $a\preceq a$ for all $a\in\grI$ (reflexivity)
\item For all triples $a,b,c\in\grI$,  if $a\preceq b$ and $b\preceq c$ then $a\preceq c$ (transitivity)
\item For every pair $a,b\in\grI$ there exists an element $c\in\grI$ such that $a\preceq c$ and $b\preceq c$ (common successor).
\end{enumerate}
\label{def:DirectedSet}
\end{definition}

The set of finite-dimensional Grassmannians
\begin{equation}
\grG=\left\{\Gr_{V_k}(\HH_{k,n})\;|\;0\leq k\leq n\right\}_{n\in \Nbb^+}
\end{equation}
forms a directed set $(\grG,\preceq)$ with respect to the following partial ordering
\begin{equation}
\Gr_{V_k}(\HH_{k,n})\preceq \Gr_{V_k'}(\HH_{k', n'})\quad \textrm{iff}\ k\leq k' \ \textrm{and}\  n-k\leq n'-k'.
\end{equation}

\begin{definition}
Let $(\grI,\preceq)$ be a directed set and $\RR_{\grI}=\{\RR_a,\;a\in\grI\}$  a collection of rings labelled by elements of $\grI$
 such  that for every pair $a\preceq b$ we have a ring homomorphism
\begin{equation}
\varphi_{a,b}:\;\RR_a\rightarrow\RR_b.
\end{equation}
Then $(\grI,\preceq,\RR_\grI,\varphi)$ is called a \textit{direct system} of ring homomorphisms over $\JJ$ if
\begin{enumerate}
\item $\varphi_{a,a}=Id_{\RR_a}$ for all $a\in\grI$,
\item $\varphi_{b,c}\circ\varphi_{a,b}=\varphi_{a,c}$ whenever $a\preceq b\preceq c$.
\end{enumerate}
\end{definition}

For the case $\grG$ of  Grassmannians, where the indexing set consists of pairs of integers $(k,n) $, $0\le k\le n$,
it follows from the commutative diagram (\ref{eq:Homomorphisms_iota_CommutativeDiagram}) that
the homomorphisms $\iota_k$ and $\iota^n$ in (\ref{eq:Homomorphisms_iota_CommutativeDiagram}) 
generate a direct system of injective ring homomorphisms
\begin{equation}
\varphi_{(k,n),(k',n')}:\quad\SS_{k,n}(\Gr_{V_k}(\HH_{k,n}))\hookrightarrow\SS_{k',n'}(\Gr_{V_{k'}}(\HH_{k',n'}))
\label{eq:DirectedSystemOfRingHomomorphismsChi}
\end{equation}
for all pairs $\Gr(k,n)\preceq\Gr(k',n')$.
\label{lemm:DirectedSystemOfCoordinateRings}
\begin{definition}
Let $(\grI,\preceq,\RR_{\grI},\varphi)$ be a direct system of ring homomorphisms.
We say that a pair $(\RR,\rho:=\{\rho_a\}_{a\in \JJ})$ consisting of a ring $\RR$ and a family of homomorphisms
\begin{equation}
\rho_a:\quad\RR_a\rightarrow\RR\qquad\textrm{for all}\quad a\in\grI
\end{equation}
is a \textit{homomorphism from the direct system $(\grI,\preceq,\RR_{\grI},\varphi)$ to the ring $R$} if
\begin{equation}
\rho_b=\varphi_{a,b}\circ\rho_a\qquad\textrm{for all}\quad a\preceq b.
\end{equation}
\end{definition}

Let $\mathbf\Lambda$  denote the  set of all partitions $\lambda$ (of any weight $|\lambda|$
and length $\ell(\lambda)$) and
\begin{equation}
\SS:=\Cbb[\pi_{\lambda}\;|\;\lambda\in\mathbf\Lambda]
\end{equation}
 the polynomial ring in infinitely many indeterminates $\{\pi_\lambda\}$ labelled by partitions.
Consider the ideal $\II\ss \SS$ generated by all Pl\"ucker relations
\begin{equation}
\II:=\Big\langle\bigcup_{0\leq k\leq n}\II_{k,n}\Big\rangle\subset\SS.
\label{eq:Definition_II}
\end{equation}

\begin{lemma}
We have a homomorphism $(\SS/\II,\sigma)$ from the direct system (\ref{eq:DirectedSystemOfRingHomomorphismsChi})
to the ring $\SS/\II$, generated by
\bea
\sigma_{k,n}:\SS_{k,n}(\Gr_{V_{k}}(\HH_{k,n})) &\&\ra\SS/\II,\cr
\sigma_{k,n}:\pi_\lambda&\& \mapsto\pi_\lambda.
\eea
\end{lemma}

\begin{definition}
Let $(\grI,\preceq,\RR_\grI,\varphi)$ be a direct system of homomorphisms.
We say a homomorphism $(\RR,\rho)$ from this system to a ring $\RR$
is the \textit{direct limit} of $\RR_{\grI}$
if it satisfies the following {\em universal property}:

For every homomorphism $(\RR',\rho')$ from this system to a ring $\RR$,
there exists a unique ring homomorphism $\upsilon: \RR\rightarrow\RR'$
which makes the following diagram commutative
\begin{equation}
\begin{tikzcd}
\RR_a \ar[dd,"\rho_a"] \ar[rrdd,"\rho'_a"]\\\\
\RR \ar[rr,dashed,"\upsilon"]&& \RR'
\end{tikzcd}
\end{equation}
for all $a\in\grI$.
\end{definition}
It is immediate from the definition that, when it exists, the direct limit is unique up to an isomorphism. When there is no ambiguity, we omit everything but the rings themselves and simply write
\begin{equation}
\RR=\lim_{\longrightarrow}\RR_a .
\end{equation}

\begin{proposition}
\begin{equation}
\SS/\II=\lim_{\longrightarrow}\,\SS_{k,n}(\Gr_{V_k}(\HH_{k,n})) .
\label{eq:DirectLimitRing}
\end{equation}
\label{prop:DirectLimitCoordinateRing}
\end{proposition}
\begin{proof}
We need to show that $(\SS/\II,\sigma)$ satisfies the universal property. Suppose there is another
homomorphism $(\RR',\sigma')$. If there is a ring homomorphism $\upsilon:\SS/\II\rightarrow\RR'$ such that
\begin{equation}
\sigma'_{k,n}=\upsilon\circ\sigma_{k,n}\qquad\textrm{for all}\quad 0\leq k\leq n,
\label{eq:UpsilonSigma_kn_Commutativity}
\end{equation}
then for all $0\leq k\leq n$ we have
\begin{equation}
\upsilon(\pi_\lambda)=\rho'_{k,n}(\pi_\lambda)\qquad\textrm{for all}\quad \lambda\in\mathbf\Lambda_{k,n}.
\label{eq:UpsilonActionOnGeneratorsPiLambda}
\end{equation}

The action on the generators (\ref{eq:UpsilonActionOnGeneratorsPiLambda}) defines a homomorphism $\upsilon:\SS\rightarrow\RR'$ of the polynomial ring which annihilates $\II_{k,n}$ for all $0\leq k\leq n$. Therefore $\upsilon$ annihilates $\II$
and hence induces a homomorphism of the quotient ring $\upsilon:\SS/\II\rightarrow R'$ satisfying (\ref{eq:UpsilonSigma_kn_Commutativity}).
\end{proof}

\begin{definition}
The \textit{infinite dimensional Grassmannian} is the homogeneous spectrum of the direct limit ring (\ref{eq:DirectLimitRing})
\begin{equation}
\Gr_{\HH_+}(\HH):=\mathrm{Proj}(\SS/ \II).
\end{equation}
\end{definition}
By construction, the coordinate ring of the infinite dimensional Grassmannian is
\begin{equation}
\SS(\Gr_{\HH_+}(\HH)):=\SS/\II=\lim_{\longrightarrow} S_{k,n}(\Gr_{V_k}(\HH_{k,n})).
\end{equation}

\begin{proposition}
The ideal $\II\subset\SS$ is prime, so $\SS / \II$ is an integral domain.
\label{prop:SOverI_IntegralDomain}
\end{proposition}
\begin{proof}
Let $a,b\in\SS$ be any pair of elements such that $ab\in\II$. From (\ref{eq:Definition_II}) and (\ref{eq:InclusionOfIkn}) it
follows  that there exists a pair of integers $0\leq k\leq n$ large enough such that $ab\in\II_{k,n}$.
But since $\II_{k,n}$ is prime we must have either $a\in\II_{k,n}\subset\II$ or $b\in\II_{k,n}\subset\II$.
\end{proof}

\begin{remark}
Relevant homogeneous maximal ideals of $\SS(\Gr_{\HH_+}(\HH))$ are in  $1-1$ correspondence
with solutions to the full set of Pl\"ucker relations (i.e.~closed points).
Here,  ``relevant'' means all ideals except the one generated by all elements of positive degree.
The reason is that in affine charts the {\em Nullstellensatz} still holds.
\end{remark}

 \begin{remark}
The above construction is different from the so-called Ind-variety (\cite{Kum}) where one takes the inverse limit of the coordinate rings or, equivalently, the direct limit of the varieties. Whereas our coordinate ring $\SS/\II$ consists of equivalence classes modulo $\II$ of finite polynomials in the Pl\"ucker coordinates, the coordinate ring $\widehat{\SS/\II}$ of an Ind-variety is obtained by completion of the ring with respect to the filtration induced by an ascending chain of Grassmannians $\dots\preceq\Gr(n,2n)\preceq \Gr(n+1,2n+2)\preceq\dots$. For example, the following formal sum belongs only to the coordinate ring of an Ind-variety
\begin{equation}
\sum_{i=0}^{\infty}\pi_{(i)}\in ( \widehat{\SS/\II}) \big\backslash(\SS/ \II).
\end{equation}

On the other hand, points in the Ind-variety always belong to arbitrarily large but finite Grassmannians, while points of $\Gr_{\HH_+}(\HH)$ are solutions of all the infinitely many Pl\"ucker relations with possibly infinitely many nonzero Pl\"ucker coordinates. 
For example
\be
\pi_{\lambda}=
\begin{cases}
n^n \ \text{ if } \lambda=(n)\\
0 \quad  \text{ otherwise}
\end{cases}
\ee
defines a point in $\Gr_{\HH_+}(\HH)$ which does not belong to any finite Grassmannian.
\end{remark}

\subsection{Big cell of the infinite Grassmannian $\Gr_{\HH_+}(\HH)$}
\label{sec: big_cell_inf}

The homomorphisms $\tilde\iota_k$ and $\tilde\iota^n$ in (\ref{eq:TildeIota_kn_definition}) generate a direct system of ring homomorphisms.
\begin{definition}
The \textit{big cell of the infinite dimensional Grassmannian} is the spectrum of the direct limit of coordinate rings of
the finite dimensional big cells
\begin{equation}
\Gr_{\HH_+}^\emptyset(\HH):= \mathrm{Spec}\,\lim_{\longrightarrow}\,\OO(\Gr_{V_k}(\HH_{k,n})).
\end{equation}
\end{definition}
It follows that the coordinate ring of the big cell
\begin{equation}
\mathcal O(\Gr_{\HH_+}^\emptyset(\HH))=\lim_{\longrightarrow}\,\OO(\Gr_{V_k}(\HH_{k,n})) =\Cbb[M_{i,j}\;|\;i,j\in\Nbb^+]
\label{eq:AffineCoodinateRing}
\end{equation}
is a polynomial ring in countably many variables labelled by pairs of nonnegative integers. 
\begin{remark}
Since the cardinality of the base field $|\Cbb|>|\Nbb^+\times\Nbb^+|$ is greater than the cardinality of the generating set for (\ref{eq:AffineCoodinateRing}), the set of closed points in $\Gr_{\HH_+}^\emptyset(\HH)$ is in one-to-one correspondence with 
$\Nbb^+\times\Nbb^+$ matrices of complex numbers. . Note, however, that this does not require any finiteness or 
convergence conditions on the elements, so they are not necessarily matrix representations
of bounded or continuous linear operators in the functional analytic sense.
\end{remark}

The  following are the infinite dimensional counterparts of Definition \ref{def:XiknMap} and Lemma \ref{lemm:XiknAnnihilatesPlucker}.
\begin{definition}
Let
\begin{equation}
\begin{aligned}
\xi:\quad\SS& \ra \Cbb[\pi_\emptyset,\pi_\emptyset^{-1}][\pi_{(i-1|j-1)}\,|\,i,j\in\Nbb^+], \\[0.3em]
\pi_\lambda& \mapsto \pi_\emptyset\,\det\left(\frac{\pi_{(a_i|b_j)}}{\pi_\emptyset}\right)_{1\leq i,j\leq r}.
\label{eq:XiMapDefinition}
\end{aligned}
\end{equation}
 be the homomorphism defined by its action on the generators of the polynomial ring,
where $r$ is the Frobenius rank of the partition $\lambda=(\ab|\bb)$.
\label{def:XiMap}
\end{definition}
\begin{lemma}
All Pl\"ucker quadratic forms in $\OO(Gr_{\HH_+}(\HH))$ are annihilated by $\xi$, so
\begin{equation}
\II\subset\ker \xi.
\label{eq:ISubsetKerXi}
\end{equation}
\label{lemm:XiAnnihilatesPlucker}
\end{lemma}
\begin{proof}
Note that restriction $\xi\big|_{\SS_{k,n}}=\xi_{k,n}$ to the subring $\SS_{k,n}$ coincides with the
homomorphism (\ref{eq:XiknMapDefinition}). Hence, by Lemma \ref{lemm:XiknAnnihilatesPlucker} we have
\begin{equation}
\II_{k,n}\subset\ker\xi\qquad\textrm{for all}\quad 0\leq k\leq n.
\end{equation}
Since  $\II$ is generated by $\II_{k,n}$, it follows that (\ref{eq:ISubsetKerXi}) holds true.
\end{proof}

\subsection{The infinite Lagrangian Grassmannian and its big cell}

\begin{definition}
Let $\II_{\LL}\subset\SS$ denote the homogeneous ideal generated by all
the Pl\"ucker  (\ref{eq:PluckerRelationIndices}),  and Lagrange  (\ref{eq: lagrange_cond_finite}) relations (for all $ k$).
\end{definition}
\begin{proposition}
Consider the direct system of ring homomorphisms $\{ \iota_k^\LL\}_{k\in\Nbb^+}$ defined in (\ref{eq:CoordinateRingInclusionLagrangianFinite}).
 We have
\begin{equation}
\SS/ \II^\LL=\lim_{\longrightarrow}\SS_{k,2k}(\Gr_{V_k}^\LL(\HH_{k,2k})).
\label{eq:DirectLimitRingLagrangian}
\end{equation}
\end{proposition}
\begin{proof}
Similar to Proposition \ref{prop:DirectLimitCoordinateRing}.
\end{proof}
\begin{definition}
The infinite dimensional Lagrangian Grassmannian is the homogeneous
spectrum of the  direct limit ring (\ref{eq:DirectLimitRingLagrangian})
\begin{equation}
\Gr_{\HH_+}^\LL(\HH,\omega)):=\mathrm{Proj}(\SS/\II^\LL).
\end{equation}
\end{definition}
By construction, the coordinate ring of the infinite Lagrangian Grassmannian is
\be
\SS(\Gr_{\HH_+}(\HH,\omega))= \SS/\II^\LL.
\ee

\begin{definition}
The \textit{big cell of the infinite Lagrangian Grassmannian} is the spectrum of the direct limit of coordinate rings
of the big cells of finite dimensional Lagrangian Grassmannians.
\begin{equation}
\Gr_{\HH_+}^{\LL,\emptyset}(\HH,\omega):=\Spec\, (\lim_{\longrightarrow}\,\mathcal O(\Gr_{V_k}^{\LL,\emptyset}(\HH_{k,2k},\omega_k))).
\end{equation}
\end{definition}
By construction, the coordinate ring of the big cell of  the infinite Lagrangian Grassmannian
\begin{equation}
\mathcal O(\Gr_{\HH_+}^{\LL,\emptyset}(\mathcal H,\omega))=\Cbb[M_{ij}\,|\,1\leq i\leq j]
\end{equation}
is a polynomial ring in countably many variables. We have a surjective homomorphism from the coordinate ring
of the full big cell to the Lagrangian big cell
\bea
\tilde\ell^*:\OO[\Gr_{\HH_+}^\emptyset(\HH,\omega)]&\& \twoheadrightarrow \OO[\Gr_{\HH_+}^{\LL,\emptyset}(\HH,\omega)]\cr
M_{i,j}&\& \mapsto M_{\min(i,j),\max(i,j)}.
\eea

The Lagrangian analogue of Definition \ref{def:XiMap} and Lemma \ref{lemm:XiAnnihilatesPlucker} is:
\begin{definition}
Define the following homomorphism from the polynomial ring $\SS$
to the ring $\Cbb\big[\pi_\emptyset,\pi_\emptyset^{-1}\big]\big[\pi_{(i-1|j-1)}\;\big|\;1\le i\leq j\big]$ by its action on
the generators
\bea
\xi^\LL: \SS&\& \rightarrow \Cbb\big[\pi_\emptyset,\pi_\emptyset^{-1}\big]\big[\pi_{(i-1|j-1)}\;\big|\;1\le i\leq j\big]\cr
&\& \cr
\xi^\LL: \pi_\lambda &\& \mapsto \pi_\emptyset\; \det\left(\frac{\pi_{(\min(a_i,b_j)| \max(a_i,b_j))}}{\pi_\emptyset}\right)_{1\leq i,j\leq r},
\label{eq:XiLDefinition}
\eea
where $r$ is the Frobenius rank of the partition $\lambda=(\ab|\bb)$.
\end{definition}

\begin{lemma}
The homomorphism $\xi^{\LL}$ annihilates all Pl\"ucker and Lagrange relations in the
 infinite dimensional Lagrangian Grassmannian $\Gr_{\HH_+}^{\LL}(\HH,\omega)$. Hence we have
\begin{equation}
\II_{\LL}\subset\ker\xi^{\LL}.
\end{equation}
\label{lemm:XiLAnnihilatesIL}
\end{lemma}
\begin{proof}
From $\xi^\LL=\tilde\ell^*\circ\xi$, and  Lemma \ref{lemm:XiAnnihilatesPlucker}, it follows that $\xi$ annihilates all the Pl\"ucker 
quadratic forms, and so does $\xi^L$.
To prove that $\xi^\LL$ annihilates all Lagrange relations, consider the restriction to the  finitely generated subring
\begin{equation}
\xi^\LL\Big|_{\SS_{k,2k}}=\tilde\ell_k^*\circ\zeta_{k,2k} \;\stackrel{(\ref{eq:ThreeDimensionalDiagramLagrangian})}{=}\; \zeta_l^\LL\circ\ell_k^*,
\label{eq:xiLagrangianFiniteRestriction}
\end{equation}
where we have used the commutativity of the leftmost face in diagram (\ref{eq:ThreeDimensionalDiagramLagrangian}). By construction, $\ell_k^*$ annihilates all Lagrange relations in $\II_k^\LL$, and so does the right hand side of (\ref{eq:xiLagrangianFiniteRestriction}).

Since this argument is valid for all $k\in\Nbb^+$,  it follows that $\xi^\LL$ annihilates all Lagrange
relations in the defining ideal $\II^\LL$ of the infinite dimensional Lagrangian Grassmannian.
\end{proof}


\section{Evaluation of $\tau$-functions and lattices  in $\Gr_{\HH_+}(\HH)$}
\label{sec:lattice_embeddings}


\subsection{Lattice evaluations of KP $\tau$-functions }
\label{sec:lattice_eval_KP_tau}

Now consider two families of indeterminants, $\{x_j\}_{i\in\Zbb}$ and $\{t_i\}_{i\in\Nbb^+}$ labelled by integers and positive integers respectively.
Let $\Cbb[[\xb,\tb]]$ be the ring of formal power series in  $\{x_j, t_i\}_{ j\in\Zbb,   \,  i \in\Nbb^+}$.  Introduce a grading on $\Cbb[[\xb,\tb]]$
by assigning
\begin{equation}
\deg (x_i)=1,\qquad \deg (t_j)=j, \quad i\in\Zbb{>0},\ j\in\Zbb.
\label{eq:CxtGradingGenerators}
\end{equation}

\begin{lemma}
For every  $i\in\Zbb$, the following maps generate  graded ring automorphisms
\begin{equation}
\begin{aligned}
\hat \delta_i:\quad \Cbb[[\xb,\tb]]\rightarrow\Cbb[[\xb,\tb]],\qquad
\left\{
\begin{array}{rcl}
x_i&\mapsto& x_i,\\
t_j&\mapsto& t_j+\dfrac{x_i^j}{j},
\end{array}
\right.
\end{aligned}
\end{equation}
that mutually commute
\begin{equation}
\hat \delta_i\hat \delta_j=\hat \delta_j\hat \delta_i, \quad \forall \ i,j\in\Zbb.
\end{equation}
\label{lemm:ShiftAutomorphismFormalPowerSeries}
\end{lemma}

In the following, the indeterminates $\tb=(t_1,t_2,t_3,\dots)$ will be viewed as   KP flow parameters and $\xb$ as a set of additional parameters
determining lattice spacings as defined below.
When there is no ambiguity, we omit explicitly indicating the dependence on $\xb$ and simply write
\begin{equation}
f(\tb+\sum_{i\in \Zbb}n_i[x_i]):=\hat \delta_i^\nb f(\tb),
\end{equation}
where
\be
\nb =(\dots , n_{-1}, n_0, n_1, \dots n_i, \dots ), \quad n_i \in \Zbb
\ee
and
\be
[x_i]:= \Big(x_i, \frac{x_i^2}{2}, \frac{x_i^3}{3}, \dots\Big)
\ee
for any $f(\tb)\in\Cbb[[\xb,\tb]]$.

Consider the multiplicative subset generated by the elements $\{x_i-x_j\}_{i\neq j}$
\begin{equation}
S_{\xb}^\times:=\langle x_i-x_j\;|\;i,j\in\Zbb, i\neq j\rangle\ \subset\ \Cbb[[\xb,\tb]];
\label{eq:MultiplicativeSubsetX}
\end{equation}
that is,  the set of all finite products of positive powers of $(x_i-x_j)$. Since $\Cbb[[\xb,\tb]]$ is an integral domain,
we can introduce its localization $\left(S_{\xb}^\times\right)^{-1}\Cbb[[\xb,\tb]]$ with respect to $S_{\xb}^\times$;
 i.e., the ring of equivalence classes consisting of quotients of elements of $\Cbb[[\xb,\tb]]$  by those
in $S_{\xb}^\times$.

\begin{remark}
Note that the generators $\{x_i-x_j\}_{i\neq j}$ of $S_{\xb}^\times$ may be identified with the roots of $\mathfrak{gl}(\infty)$.
\end{remark}
\begin{corollary}
For every integer $i\in\Zbb$, the graded automorphism $\hat \delta_i$  extends uniquely to an automorphism of the localized ring
\begin{equation}
\hat \delta_i:\quad \left(S_{\xb}^\times\right)^{-1}\Cbb[[\xb,\tb]]\rightarrow \left(S_{\xb}^\times\right)^{-1}\Cbb[[\xb,\tb]].
\end{equation}
\label{cor:ShiftAutomorphismFieldOfFractions}
\end{corollary}

Let $\Ab\subset\Zbb^\infty$ be the infinitely generated free abelian subgroup consisting of elements with only a finite number of nonzero entries.
This will also be referred to as the infinite-dimensional integer lattice.

Fix an element $w\in \Gr_{\HH_+}(\HH)$, and let $\tau_w^{KP}(\tb)\in\Cbb[[\tb]]\subset \Cbb[[\xb,\tb]]$ denote
the corresponding KP $\tau$-function. To this we can associate a function $H_w$ on the infinite dimensional lattice $\Ab$,
where for each element $\nb\in\Ab$ we have
\be
    H_w^{\nb}(\tb):=\prod_{i<j}(x_i-x_j)^{n_in_j}\;\tau_w^{KP}\Big(\tb+\sum_{i=-\infty}^{\infty}n_i [x_i]\Big)\in \left(S_{\xb}^\times\right)^{-1}\Cbb[[\xb,\tb]].
\label{eq:NormalizedEvaluationOfTauFunction}
\ee
\begin{remark}
For most purposes in what follows, we view $\tau$-functions as  formal power series. Throughout this section we will be working in the
 localized ring $(S_{\xb}^\times)^{-1}\Cbb[[\xb,\tb]]$, leaving all convergence issues aside. It is worth noting, however, that in the simplest case, when $w\in \Gr_{\HH_+}(\HH)$ is chosen so that $\tau_w^{KP}(\tb)$ is a polynomial in the $\tb$ variables
  \cite{HL, KRV, HO},  $H^{\nb}_w(\tb)$ is also a polynomial in these variables, with coefficients that are rational functions in $x_i$. Moreover, because the only type of denominators that appear in (\ref{eq:NormalizedEvaluationOfTauFunction}) are powers of $x_i-x_j$, we can  evaluate the external parameters  in this case at any set of pairwise distinct complex numbers. 
\end{remark}

\begin{definition}
For $i\in \Zbb$, let $\alphab_i \in \Ab$ be the basis element with all components $0$ except for a single nonzero 
entry $1$ in the $i$th position.
\end{definition}
The following Lemma describes the change in $H_w^{\nb}$ with respect to shifts of the origin by the elements $\alphab_i$.
\begin{lemma}
For all $i\in\Zbb$, we have
\be
H_w^{\nb+\alphab_i}(\tb)=(-1)^{\sum_{j<i}n_j}\prod_{j\neq i}(x_i-x_j)^{n_j}H_w^{\nb}(\tb+[x_i]).
\label{H_transl_alpha_i}
\ee
\label{lemm:HbasePointShift}
\end{lemma}
\begin{proof}
This follows from the definition (\ref{eq:NormalizedEvaluationOfTauFunction}) .
\end{proof}
\begin{remark}
\label{H_w_tau_functions}
Note that, since $H^{\nb}_w(\tb)$ only differs from $\tau_w(\tb)$  by a constant translation in the argument $\tb$ and a 
   $\tb$-independent multiplicative factor that is rational in the parameters, $H^\nb_w(\tb)$ is also a KP $\tau$-function.
   Since $\tau$-functions are only defined up to multiplication by an arbitrary nonzero constant  factor we have, 
   up to projective  equivalence,
   \be
   H^\nb_{\gamma_+(\sb)w}(\tb)\sim  H^\nb_{w}(\tb+ \sb)
   \ee 
   and hence
   \be
    H^{\nb+\alphab_i}_{w}(\tb)\sim H^{\nb}_{w}(\tb +[x_i])) =  H^\nb_{\gamma_+([x_i])w}(\tb).
    \label{H_trans_alpha_i}
    \ee
   Eq.~(\ref{H_transl_alpha_i}) may therefore be viewed as generating a projective action on the space of such $\tau$-functions $H^\nb_w(\tb)$,
    of the infinite discrete abelian subgroup $\Gamma^\Ab_+ \ss\Gamma_+$ consisting of evaluations
    \be
    \gamma_+^\nb := \gamma_+\left(\sum_{j\in \Zbb} n_j [x_j]\right) , \quad \nb = (\dots, n_i, n_{i+1}, \dots) \in \Ab
    \ee
    at lattice points $\nb \in \Ab$.  By eq.~(\ref{H_trans_alpha_i}),
   this may be viewed as evaluation of the $\tau$-function at the lattice points 
   \be
   \gamma_+\left(\sum_{j\in \Zbb} n_j [x_j]\right)w\in \Gamma_+(w)
   \ee
    embedded within the $\Gamma_+$ orbit $\{w(\tb)\}$.
\end{remark}

The polynomial ring in the  $H_w^{(\nb)}(\tb)$'s for fixed  $w$ and $\tb$ and varying $\nb\in\Ab$ has a
natural grading by elements of $\Ab$. From Lemma \ref{lemm:HbasePointShift} it follows that:
\begin{corollary}
For every monomial $H_w^{\nb^{(1)}}(\tb)\dots H_w^{\nb^{(m)}}(\tb)$ of total degree
\be
\nb^{(tot)}:=\nb^{(1)}+\dots+\nb^{(m)},
\ee
 we have
\be
\begin{aligned}
    &H_w^{\nb^{(1)}+\mathbf\alphab_i}(\tb)\dots H_w^{\nb^{(m)}+\mathbf\alphab_i}(\tb)\\
    &\qquad=(-1)^{\sum_{j<i}n^{(tot)}_j}\prod_{j\neq i}(x_i-x_j)^{n^{(tot)}_j}  H_w^{\nb^{(1)}}(\tb+[x_i])\dots H_w^{\nb^{(m)}}(\tb+[x_i])
\end{aligned}
\label{eq:ShiftTransformationHomogeneousH}
\ee
\label{cor:ShiftTransformationHomogeneousH}
\end{corollary}

Thus, all monomials of given homogeneity degree $\nb^{(tot)}$ scale by the same factor when the origin of the lattice $\Ab$ is shifted.

\subsection{Pl\"ucker relations satisfied by $H_w^{\nb}(\tb)$}
\label{sec:lattice_points_in_orbit}
\begin{proposition}
\label{prop:pi_H_alpha_i_hom}
Fix  an element $w\in \Gr_{\HH_+}(\HH)$ of the infinite Grassmannian and let $k,n\in\Nbb^+$,  $k < n$,
be an ordered pair of positive integers.
The following generates a homomorphism of commutative algebras
\bea
    \phi_{w,k,n}^{\0b}: \SS(\Gr_{V_k}(\HH_{k,n})) &\&\rightarrow \left(S_{\xb}^\times\right)^{-1}\Cbb[[\xb,\tb]],\\
    \widetilde{\pi}_{L_1,\dots,L_k}&\&\mapsto\sgn(L) H_w^{\alphab_{L_1}+\dots+\alphab_{L_k}}(\tb).
      \label{eq:pi_lambda_H_alpha_i_hom}
\eea
\label{prop:FiniteGrassmannianOrigin}
\end{proposition}
\begin{proof}
Let $s$ be the cardinality $s = |I\cap J|$ of the intersection of the multi-indices $I$ and $J$ appearing in the Pl\"ucker relation (\ref{eq:PluckerRelationIndices}).
Because of  skew-symmetry, there is no loss of generality in assuming that $I_i=J_i$  for $1\le i \le s$,
and all elements $I_{s+1},\dots, I_{k-1},J_{s+1},\dots J_{k+1}$ are distinct, and ordered,
\be
    I_{s+1}<I_{s+2}<\dots< I_{k-1}\quad \text{and}\quad J_{s+1}<J_{s+2}<\dots<J_{k+1}.
\label{eq:TailOrderingIJ}
\ee
and showing that the Pl\"ucker relations  (\ref{eq:PluckerRelationIndices}) are satisfied by all 
elements $H_w^{\alphab_{L_1}+\dots+\alphab_{L_k}}(\tb)$ in this case.
Under these assumptions, define
\be
   \tb^{I\cap J} :=\tb+\sum_{i=1}^s[x_{I_i}]  =\tb+\sum_{i=1}^s[x_{J_i}].
\ee
Applying the KP addition formulae (\ref{tau_k_addition_formula}),  using
\be
H_w^{\sum_{i=1}^k \alphab_{L_i}}(\tb) = \sgn(L_1, \dots, L_k)\zeta_N(\tb,  x_{L_1}, \dots, x_{L_k}),
\ee
we get
\bea
 \sum_{j=s+1}^{k+1}(-1)^{j}  &\& \sgn(I_{s+1},\dots,I_{k-1},J_j) H_w^{\alphab_{I_{s+1}}+\dots+\alphab_{I_{k-1}}+\alphab_{J_j}}( \tb^{I\cap J}) \cr
 &\& \times H_w^{\alphab_{J_{s+1}}+\dots+\hat{\alphab}_{J_{j}}+\dots+\alphab_{J_{k+1}}}( \tb^{I\cap J})  =0, \cr
 &\&
\label{eq:PluckerRelationHOriginS}
\eea
where $\hat{\alphab}_{J_{j}}$ denotes the absence of the term $\alphab_{J_{j}}$ from the sum.
The following therefore holds for all $0\leq s'\leq s$
\bea
  \sum_{j=s'+1}^{k+1}(-1)^{j}  &\& \sgn(I_{s'+1},\dots,I_{k-1},J_j) H_w^{\alphab_{I_{s'+1}}+\dots+\alphab_{I_{k-1}}+\alphab_{J_j}}(\tb^{I\cap J}) \cr
  &\& \times
   \sgn(J_{s'+1},\dots,\hat{J}_{j},\dots,J_{k+1}) H_w^{\alphab_{J_{s'+1}}+\dots+\hat{\alphab}_{J_{j}}+\dots+\alphab_{J_{k+1}}}(\tb^{I\cap J}) =0.\cr
   &\&
\label{eq:PluckerRelationHOriginInductive}
\eea
This follows by finite induction in $s'$, starting at $s'=s$, and descending.
The case $s'=s$ is just (\ref{eq:PluckerRelationHOriginS}), where, according to (\ref{eq:TailOrderingIJ}) we have $\sgn(J_{s+1},\dots,J_{k+1})=1$.
Assuming now that (\ref{eq:PluckerRelationHOriginInductive}) holds for some $s'$, we show that it also holds for $s'-1$.

Note that the left-hand side of eq.~(\ref{eq:PluckerRelationHOriginInductive}) is homogeneous with respect to the grading given by elements of $\Ab$. Using Corollary (\ref{cor:ShiftTransformationHomogeneousH}) for $i=I_{s'}=J_{s'}$ and cancelling the common factor
\be
\prod_{j=s'+1}^{k-1}(x_{I_{s'}}-x_{I_j})\prod_{j=s'+1}^{k+1}(x_{J_{s'}}-x_{J_j}) ,
\ee
we get
\bea
   &\&  \sum_{j=s'}^{k+1}(-1)^{j}\sgn(I_{s'},\dots,I_{k-1},J_j) H_w^{\alphab_{I_{s'}}+\dots+\alphab_{I_{k-1}}+\alphab_{J_j}}\Big(\tb+\sum_{i=1}^{s'-1}[x_{I_i}]\Big)\times\cr
    &\& \times \sgn(J_{s'},\dots,J_{j-1},J_{j+1},\dots,J_{k+1}) H_w^{\alphab_{J_{s'}}+\dots+\hat{\alphab}_{J_{j}}+\dots+\alphab_{J_{k+1}}} \Big(\tb+\sum_{i=1}^{s'-1}[x_{J_i}]\Big)=0, \cr
    &\&
\label{eq:PluckerRelationHInductiveStepConclusion}
\eea
where we have used the fact that for $s'\leq s$ we have $I_{s'}=J_{s'}$, and hence the following ratio of signs is independent of $j$:
\begin{equation}
\begin{aligned}
&\frac{\sgn(I_{s'},\dots,I_{k-1},J_j)\sgn(J_{s'},\dots,\widehat{J_{j}},\dots,J_{k+1})} {\sgn(I_{s'+1},\dots,I_{k-1},J_j)\sgn(J_{s'+1},\dots,\widehat{J_{j}},\dots,J_{k+1})}\\[0.5em]
&\qquad=\left(\prod_{i=s'+1}^{k-1}\sgn(I_{s'},I_{i})\right)\sgn(I_{s'},J_j) \left(\prod_{i=s'+1}^{k+1}\sgn(J_{s'},J_i)\right)\sgn(J_{s'},J_j)\\
&\qquad=\prod_{i=s'+1}^{k-1}\sgn(I_{s'},I_{i}) \prod_{i=s'+1}^{k+1}\sgn(J_{s'},J_i).
\end{aligned}
\end{equation}

Eq.~(\ref{eq:PluckerRelationHInductiveStepConclusion}) shows that (\ref{eq:PluckerRelationHOriginInductive}) is true for $s'-1$,
and hence, by induction, it holds for all $0\leq s'\leq s$.
In particular, for $s'=0$ we have
\bea
    &\&\sum_{j=1}^{k+1}(-1)^{j}\sgn(I_1,\dots,I_{k-1},J_j) H_w^{\alphab_{I_1}+\dots+\alphab_{I_{k-1}}+\alphab_{J_j}}(\tb)\cr
    &\&\quad\times\ \sgn( J_1,\dots,J_{j-1},J_{j+1},\dots,J_{k+1}) H_w^{\alphab_{J_1}+\dots+\hat{\alphab}_{J_{j}}+\dots+\alphab_{J_{k+1}}}(\tb)\cr
    &\&=\sum_{j=1}^{k+1}(-1)^j\phi_{w,k,n}^{\0b}\big(\widetilde{\pi}_{I_1,\dots,I_{k-1},J_j}\big) \phi_{w,k,n}^{\0b}\big(\widetilde{\pi}_{J_1,\dots,\widehat{J_j},\dots,J_{k+1}}\big)=0.
\eea
\end{proof}

\begin{proposition}
\label{thm:homomorph_grassmannian_lattice_eval}
Fix an element $w\in \Gr_{\HH_+}(\HH)$ and let $k,n\in\Nbb$,  $k< n$ be an ordered pair of positive integers.   For every element
  $\nb\in\Ab$, we have a homomorphism of commutative algebras
\be
    \phi_{w,k,n}^{\nb}:\SS(\Gr_{V_k}(\HH_{k,n}))  \rightarrow \left(S_{\xb}^\times\right)^{-1}\Cbb[[\xb,\tb]]
    \ee
    such that
    \be
     \phi_{w,k,n}^{\nb}: \widetilde{\pi}_{L_1,\dots,L_k}\mapsto \sgn(L) H_w^{\nb+\alphab_{L_1}+\dots+\alphab_{L_k}}(\tb).
    \label{eq:FiniteGrassmannianGeneralLatticePoint}
    \ee
\label{prop:LatticeFiniteGrassmannianPoints}
\end{proposition}
\begin{proof}
We proceed by induction on the lattice $\Ab$. The starting point  is ${\nb}=\0b$,
which is given by Proposition \ref{prop:FiniteGrassmannianOrigin}. Now, assume that $\phi_{w,k,n}^{\nb}$ is a homomorphism for some $\nb\in\Ab$.
 We will show that this implies $\phi_{w,k,n}^{\nb+\alphab_i}$ is also a commutative ring homomorphism for the translation by any $\alphab_i$.

To do  this, let $I$ be a $k-1$-tuple of indices and $J$ a $k+1$ tuple. By the inductive assumption
\bea
    0=&\&\sum_{j=1}^{k+1}(-1)^j\phi_{w,k,n}^{\nb}(\widetilde{\pi}_{I_1,\dots,I_{k-1},J_j}) \phi_{w,k,n}^{\nb}(\widetilde{\pi}_{J_1,\dots,\widehat{J_j},\dots,J_{n+1}})\cr
    =&\&\sum_{j=1}^{k+1}(-1)^{j}\sgn(I_1,\dots,I_{k-1},J_j) H_w^{\nb+\alphab_{I_1}+\dots+\alphab_{I_{k-1}}+\alphab_{J_j}}(\tb)\cr
    &\&\qquad \times \sgn(J_1,\dots,\widehat{J_j},\dots,J_{k+1})
    H_w^{\nb+\alphab_{J_1}+\dots+\hat{\alphab}_{J_j}+\dots+\alphab_{J_{k+1}}}(\tb).
\label{eq:FiniteGrassmannianLatticeShiftInductiveAssumption}
\eea
Note that the right hand side of (\ref{eq:FiniteGrassmannianLatticeShiftInductiveAssumption}) is homogeneous with respect to the grading given by elements of $\Ab$. Corollary \ref{cor:ShiftTransformationHomogeneousH} implies that
\bea
    0=&\& \sum_{j=1}^{k+1}(-1)^{j}\sgn(I_1,\dots,I_{k-1},J_j) H_w^{\nb+\alphab_i+\alphab_{I_1}+\dots+\alphab_{I_{k-1}}+\alphab_{J_j}}(\tb-[x_i])\cr
    &\& \quad\times \sgn(J_1,\dots,\widehat{J_j},\dots,J_{k+1}) H_w^{\nb+\alphab_i+\alphab_{J_1}+\dots+\widehat{\alphab}_{J_j}+\dots+\alphab_{J_{k+1}}}(\tb-[x_i]).
\eea
The right hand side of the latter is the result of applying $\phi_{w,k,n}^{\nb+\alphab_i}\circ \hat \delta_i$ to the same Pl\"ucker relation, where $\hat \delta_i$ is the shift operator introduced in Lemma \ref{lemm:ShiftAutomorphismFormalPowerSeries} and Corollary \ref{cor:ShiftAutomorphismFieldOfFractions}. The same result holds for all Pl\"ucker relations,
so we conclude that $\phi_{w,k,n}^{\nb+\alphab_i}\circ \hat \delta_i$ is a homomorphism for all $i\in \Zbb$. Since
$\hat \delta_i$ is an automorphism of the ring $\Cbb[[\tb]]$, it follows that $\phi_{w,k,n}^{\nb+\alphab_i}$ must also be a homomorphism.
\end{proof}

For  consistency with the infinite dimensional case, it is convenient
 to label Pl\"ucker coordinates by partitions. To every $k$-element increasingly ordered subset
\begin{equation}
L =(L_1 \dots, L_k)\subset\{-k,\dots,n-k-1\}
\end{equation}
having $k-r$ negative elements, where  $0\le r \le k$,
\be
-k \le L_i\le -1, \quad {i \in \{1, \dots, k-r\}},
\ee
and  $r$ non-negative elements,
\be
0 \le L_{k-r+i}\le n-k-1, \quad {i \in \{1, \dots, r\}},
\ee
we can associate two increasingly ordered sets of $r$ positive integers
\begin{subequations}
\bea
I&\&=(I_1, \dots I_r),, \quad  1\le I_i  \le I_{i+1}\le n-k ,  \\
 J&\&=(J_1, \dots J_r), \quad  1\le J_i < J_{i+1} \le k, \quad i =1, \dots, r,
\eea
\end{subequations}
such that
\be
I_i=1+L_{k-r+i}, \quad J=\{1,\dots, k\}\backslash (-L).
\label{eq:FiniteLDecompositionIJ}
\ee

Let
\be
(\ab | \bb)= (a_1 \dots a_r | b_1 \dots b_r)
\ee
be the Frobenius notation for a partition $\lambda(I,J)$  of Frobenius rank $r$,
whose Young diagram fits into the rectangular one for $(n-k)^k$,  with Frobenius indices defined as
\be
a_i=I_{r-i+1}-1,\quad b_i=J_{r-i+1}-1,\quad 1\leq i\leq r.
\ee
Introduce the following notation
 \be
   H_{w}^{\nb,\lambda}(\tb):=H_w^{\nb+\sum_{i=1}^r \alphab_{a_i}-\sum_{i=1}^r\alphab_{-b_i-1} }(\tb),
    \label{eq:H_n0_w_lambda}
\ee
where $\lambda:=\lambda(I,J)=(\ab|\bb)$, and define
\begin{equation}
\Phi_{w,k,n}^{\nb}:=\phi_{w,k,n}^{\nb+\sum_{j=1}^k\alphab_{-j}}.
\label{eq:ShiftedHomomorphismFinite}
\end{equation}
We can reformulate Proposition \ref{prop:LatticeFiniteGrassmannianPoints} in terms of Pl\"ucker coordinates labelled by partitions.
\begin{corollary}
\label{cor:plucker_lattice_homomorph}
Fix an element $w\in \Gr_{\HH_+}(\HH)$  and
a lattice point $\nb\in\Ab$.  The homomorphism (\ref{eq:ShiftedHomomorphismFinite})
applied to $\pi_\lambda$ gives
\be
\Phi_{w,k,n}^{\nb}(\pi_{\lambda})= H_w^{\nb,\lambda}.
\label{Phi_n0_lambda_H_n0_lambda}
\ee
\label{cor:FiniteGrassmannianLatticeLambda}
\end{corollary}
\begin{proof}
We reduce this to the statement of Proposition \ref{prop:LatticeFiniteGrassmannianPoints}.
Let $J\subset\{1,\dots,k\}$ and $I\subset\{1,\dots,n-k\}$ be a pair of ordered subsets of the same cardinality $|I|=|J|=r$ and $\lambda(I,J)$
 the corresponding partition of Frobenius rank $r$. From (\ref{eq:FiniteLDecompositionIJ}) we get the
expression for the indexing set
\be
L= (-J^c)\sqcup (I_1-1, \dots, I_r-1) \subset (-k,\dots,n-k-1)
\ee
where
\be
J^c=\{1,\dots,k\}\backslash J.
\ee
in reversed; i.e., decreasing order)
It follows that
\begin{equation}
\begin{aligned}
\mb=&\nb+\sum_{i=1}^r\alphab_{I_i-1}-\sum_{i=1}^r\alphab_{-J_i}\\[0.3em]
=&\nb+\sum_{i=1}^r\alphab_{I_i-1}+\sum_{i=1}^{k-r}\alphab_{-J^c_i}-\sum_{j=1}^k\alphab_{-j}\\[0.3em]
=&\nb-\sum_{j=1}^k\alphab_{-j}+\sum_{i=1}^{k}L_i,
\end{aligned}
\end{equation}
and hence
\begin{equation}
\Phi_{w,k,n}^{\nb}(\pi_\lambda)=H_w^{\mb}=\phi_{w,k,n}^{\nb-\sum_{j=1}^k\alphab_{-j}}(\widetilde{\pi}_{L_1,\dots,L_k}).
\end{equation}
Since this holds for all partitions $\lambda$ that fit into the $k\times (n-k)$ rectangle,
 we arrive at (\ref{Phi_n0_lambda_H_n0_lambda}), which concludes the proof.
\end{proof}

\subsection{Lattice mapping to $\Gr_{\HH_+}(\HH)$}
\label{sec: lattice_inf_grassmann}

In this section we show that to every KP $\tau$-function $\tau_w^{KP}(\tb)$ and every choice of evaluations
of the indeterminants $\{x_i\}_{i\in\in \Zbb}$, we can associate an infinite dimensional lattice of points in $\Gr_{\HH_+}(\HH)$.
For every base point $\nb\in\Ab$, evaluations of the corresponding $\tau$-function
$\tau^{KP}_w$ as in (\ref{eq:NormalizedEvaluationOfTauFunction}) at a lattice of associated points in the $\Gamma_+$
orbit of $w$ satisfy the Pl\"ucker relations of an infinite dimensional Grassmannian and thus can be interpreted as Pl\"ucker coordinates
of some other Grassmannian element $\tilde{w} \in \Gr_{\HH_+}(\HH)$.

\begin{theorem}
\label{thm:plucker_lattice_homomorph}
Fix an element $w\in \Gr_{\HH_+}(\HH)$ and let $\nb\in\Ab$ be a lattice point.
The following map, evaluated on the Pl\"ucker coordinates, viewed as generators of the ring $\SS(\Gr_{\HH_+}(\HH))$,
extends to a homomorphism of commutative algebras
\begin{equation}
\begin{aligned}
\Phi_w^{\nb}:\SS(\Gr_{\HH_+}(\HH))\rightarrow&\left(S_{\xb}^\times\right)^{-1}\Cbb[[\xb,\tb]]\\
\pi_\lambda\mapsto&H_w^{\nb,\lambda}.
\end{aligned}
\label{eq:LatticeOfPointsInInfiniteGrassmannian}
\end{equation}
\label{th:LatticeOfPointsInInfiniteGrassmannian}
\end{theorem}

\begin{proof}
Every Pl\"ucker relation involves only finitely many generators $\{\pi_\lambda\;|\;\lambda\in\Lambda\}$ indexed by some 
finite set $\Lambda$ of partitions.
Let $n\in\Nbb^+$ be  large enough that all Young diagrams of partitions $\lambda\in \Lambda$  
fit into an $n\times n$ square. Under these assumptions we have
\begin{equation}
\Phi_w^{\nb}(\pi_\lambda)=\Phi_{w,n,2n}^{\nb}(\pi_\lambda),\qquad\textrm{for all}\quad\lambda\in\Lambda.
\end{equation}
From Corollary \ref{cor:FiniteGrassmannianLatticeLambda} it follows that all the Pl\"ucker quadratic forms
 are annihilated by $\Phi_w^{\nb}$.
This is valid for any Pl\"ucker relation, and any $w \in \Gr_{\HH_+}(\HH)$.
\end{proof}

\begin{example}{\bf The octahedron relations (discrete Hirota equations)}

Consider the short Pl\"ucker relation (\ref{eq:PluckerRelationIndices})
 \be
\pi_\emptyset \,\pi_{(2,2)} - \pi_{(1)}\, \pi_{(2,1)} + \pi_{(2)}\,  \pi_{(1,1)} = 0.
\ee
 for
 \be
 (k, n)=(2,2), \ I_1=-2, \ (J_1,J_2, J_3) = (-1, 0, 1).
 \ee
Choose parameter values
\be
x_{-2}=0, \quad x_{-1}=a, \quad x_0 = b,\quad x_1 = c,
\ee
and lattice values
\bea
n_{-2} &\&=1, \quad n_{-1} =l+1, \quad  n_{0} =m, \quad  n_{1} =n, \\
&\& \cr
n_i &\& = 0 \quad \text{if} \quad i \neq -2, -1, 0, 1.
\eea
and define the quantities   $H^{\nb, \emptyset}_w$,  $H^{\nb, (1)}_w$, $H^{\nb,(2)}_w$,
$H^{\nb, (1,1)}_w$, $H^{\nb,(2,1)}_w$, $H^{\nb,(2,2)}_w$
as in  eq.~(\ref{eq:H_n0_w_lambda}).  Theorem \ref{thm:plucker_lattice_homomorph} then implies that these also satisfy
\be
H_w^{\nb,\emptyset} \,H_w^{\nb,(2,2)} -H_w^{\nb,(1)}\, H_w^{\nb,(2,1)} + H_w^{\nb,(2)}\, H_w^{\nb,(1,1)} = 0
\ee
for all $\nb$. Defining $\tau_{l,m,n}$ as in eq.~(\ref{eq:tau_lmn_def}), and
making the substitutions (\ref{eq:H_n0_w_lambda}), we get
\be
a(b-c) \,\tau_{l+1, m,n} \tau_{l, m+1, n+1} + b(a-c)\, \tau_{l, m+1, n} \tau_{l,m+1, n+1}
+ c(a-b)\, \tau_{l, m, n+1} \tau_{l+1, m+1, n} =0,
\label{eq:octahedron_rels4}
\ee
which are the octahedron (or discrete Hirota) relations \cite{Hir, Mi}.
\end{example}
\subsection{The fraction field $\mathbb F_{\xb,\tb}$ and the effective affine big cell}
\label{sec:fraction_field}

For every element $w\in \Gr_{\HH_+}(\HH)$, the corresponding KP $\tau$-function $\tau^{KP}_{w}(\tb)\in\Cbb[[\xb,\tb]]$ is a nonzero element of the ring of formal power series and so are its normalized shifts
\bea
 H_w^{\nb}(\tb)&\&=\prod_{i<j}(x_i-x_j)^{n_in_j}\;\tau_w^{KP}\Big(\tb+\sum_{i=-\infty}^{\infty}n_i [x_i]\Big)  \neq0,
\label{eq:NonvanishingTauFunction}
\eea
for all $\nb\in\Ab$.  Therefore, for every $\nb\in\Ab$, the image of the Plucker coordinate corresponding to the trivial partition
\begin{equation}
\Phi_w^{\nb}(\pi_\emptyset) =H_w^{\nb,\emptyset} =H_w^{\nb}\neq 0.
\label{eq:NonvanishingFormalEvaluation}
\end{equation}
is a nonzero element.

Recalling that the ring of formal power series $\Cbb[[\xb,\tb]]$ is an integral domain, denote its field of fractions by $\mathbb F_{\xb,\tb}$. Lattice evaluations (\ref{eq:LatticeOfPointsInInfiniteGrassmannian}) are elements of the localized ring
\begin{equation}
\Phi_w^{\nb}(\pi_\lambda)=H_{w}^{\nb,\lambda} \in \left( S_{\xb}^\times\right)^{-1}\Cbb[[\xb,\tb]] \subset \mathbb F_{\xb,\tb}.
\end{equation}
The localized ring, in turn, is a subring of $\mathbb F_{\xb,\tb}$. By (\ref{eq:NonvanishingFormalEvaluation}), we can thus view
 $\Phi_w^{\nb}(\pi_\emptyset)$ as invertible elements of the field $\mathbb F_{\xb,\tb}$.

The analogue of the determinant formula (\ref{eq:DeterminantFormulaPluckerCoordinatesFinite}) for evaluations of the $\tau$-function
is given by the following.
\begin{proposition}
\label{prop:DeterminantFormulaAffineBigCell}
Let $w\in \Gr_{\HH_+}(\HH)$ be a fixed point in the infinite Grassmannian and $\nb\in\Ab$  a lattice point.
For every partition $\lambda=(\ab|\bb)$ of Frobenius rank $r$, the identity
\begin{equation}
\frac{H^{\nb,\lambda}_w}{H^{\nb}_w}=\det\left(\frac{H^{\nb+\alphab_{a_i}-\alphab_{-b_j-1}}_w}{H^{\nb}_w}\right)_{1\leq i,j\leq r}
\label{eq:AffineDeterminantRelationLattice}
\end{equation}
holds in the field of fractions $\mathbb F_{\xb,\tb}$.
\end{proposition}

\begin{proof}
By Proposition 2 in \cite{Shig} (or Proposition 3.10.4 in \cite{HB}), we have
the following consequence  of the addition formulae (Theorem \ref{th:addition_formula},
eq.~\ref{tau_k_addition_formula}).
\begin{equation}
\frac{\tau^{KP}_w\big(\tb+\sum_{i=1}^r[ x_i]-\sum_{i=1}^r[ y_i]\big)}{\tau^{KP}_w\big(\tb\big)}\frac{\prod_{i<j}( x_i- x_j)(y_j- y_i)}{\prod_{i,j=1}^r(x_i-y_j)} =\det\left(\frac{\tau^{KP}_w\big(\tb+[x_i]-[ y_i]\big)}{\tau^{KP}\big(\tb\big)}\right)_{1\leq i,j\leq r}.
\label{eq:AdditionDeterminantFormula}
\end{equation}
Substituting
\begin{equation}
\widetilde{\tb}\ra\tb+\sum_{i=-\infty}^{\infty}n_i[x_i],\qquad \widetilde x_i\ra x_{a_i},\qquad \widetilde y_i \ra x_{-b_i-1}
\end{equation}
into (\ref{eq:AdditionDeterminantFormula}) we get precisely (\ref{eq:AffineDeterminantRelationLattice}).
\end{proof}

\begin{remark}
\label{big_cell_ring_extension}
Viewing lattice evaluations $H_w^{\nb}\in\mathbb F_{\xb,\tb}$ as elements of the fraction field of formal power series in the parameters, 
eq.~(\ref{eq:NonvanishingFormalEvaluation}) implies that, effectively, we are always in the big cell.
It is worth noting, however, that when the element $w\in \Gr_{\HH_+}(\HH)$ is chosen in such a way that $\tau_w^{KP}(\tb)$ is a polynomial, \cite{HL, KRV, HO} its lattice evaluations for particular values of parameters and points of the lattice are allowed to be zero, because 
$\tau_w^{KP}(\tb)$ is required to be nontrivial only as a function of $\tb$. This corresponds to the fact that although identity (\ref{eq:AffineDeterminantRelationLattice}) holds for rational functions of $\xb$ and $\tb$, it does not necessarily translate to evaluations of parameters, even in this simple case.
\label{rem:EffectiveBigCell}
\end{remark}


\subsection{Pl\"ucker coordinates for $w \in \Gr_{\HH_+}(\HH)$ vs. evaluation of $\tau^{KP}_w(\tb)$ on a lattice.}
\label{sec:lattice_eval_plucker_rels}

 Theorem \ref{thm:plucker_lattice_homomorph}  tells us that the ``shifted'' $\tau$-functions $H_w^{\nb,\lambda}$ are homomorphic images of the Plucker coordinates $\pi_\lambda$   on the Grassmannian $\Gr_{\HH_+}(\HH) $, and so satisfy Pl\"ucker relations. Considering just the  $H_w^{\nb}$; these are, up to some (important) normalizations, the $\tau$-function evaluated on a lattice in the infinite dimensional  Grassmannian $\Gr_{\HH_+}(\HH)$ parametrized by $\nb$. The Pl\"ucker relations then translate into recursions on the lattice.
 
Pl\"ucker relations are normally associated to several different coordinates evaluated for a fixed subspace $w$ rather than, as is the case here, a single Pl\"ucker coordinate (the $\tau$-function) evaluated on some subset of a lattice of points. It is natural to ask if there is a relation. The answer is in a sense given  by Theorem \ref{thm:plucker_lattice_homomorph}, but one can see this more directly, at least in a family of simple cases. Our specific example will be based on finite dimensional reductions of the infinite Grassmannian through a suitable subquotient procedure (\cite{HB}, Chapt. 6).

We  fix a finite family of points in the plane $\{x_i\}_{i=1,..,k}$ with $|x_i|>1$, and establish a correspondence between
\begin{itemize}
\item Pl\"ucker coordinates, parametrized by  a set of indices $\{i_1, ..., i_l\}\subset \{1,...,k\}$ for $ l\leq k$,  for a single subspace $w\in \Gr_{\HH_+}(\HH)$.
\item The single Pl\"ucker coordinate $\pi_\emptyset( w(i_1, ..., i_l))$ (i.e. an evaluation of the $\tau$-function) for a corresponding  family of elements
$w(i_1, ..., i_l)\in \Gr_{\HH_+}(\HH)$.
\end{itemize}
  We do this using the procedure  for reducing the infinite dimensional Grassmannian to  finite dimensional ones 
  given in \cite{HB}, Secs.~6.3, 6.4. 
 Consider the family $Gr(w_1, w_2) $ of subspaces $w \ss \HH$ satisfying 
  \be
   w_1  \ss   w  \ss  w_2 . 
   \ee
   where the subspaces $w_1 \ss w_2 \ss \HH$ are  defined by requiring a specified set of zeroes in the 
   the elements of $w_1\ss \HH_+$ and allowing a set of permissible poles in $w_2$:
   \be
    w_1 = r(z) \HH_+,  \quad w_2 := \frac{r(z)}{p(z)} \HH_+.
   \ee
Here $r(z)$, $p(z)$ are specified polynomials with roots inside the unit circle $S^1\ss \Cb$ centred at the origin, and
the virtual dimensions of $w_1$ and $w_2$ are  $-\deg(r)$ and $\deg(p)- \deg(r)$, respectively. Our space $Gr(w_1, w_2) $ is a union over virtual dimensions lying between $-\deg(r)$ and $\deg(p)- \deg(r)$ of finite dimensional Grassmannians  of spaces  in $\HH$.

For our example, we choose 
  \be
  r(z) := \prod_{i=1}^{k} (1-x_iz), \quad p(z) :=  \prod_{i=1}^{k} (1-x_iz)^2,
  \ee
   with $|x_i|>1$ and all $x_i$'s distinct.The quotient $w_2/w_1$ has a basis $\{d_{i,m}\}_{ i\in \{1,  \dots, k\}, \atop m\in \{1, 2\}}$ given by 
\be
d_{i,m}(z):= r(z) (1-x_i z)^{-m},  \quad i\in \{1,  \dots, k\}, \ m\in \{1, 2\}.
\label{eq:basis_W_12_quotient}
\ee

 Recall  that the lattice in the preceding sections  was given by the action of a discrete subgroup   of $\Gamma_+$
 defined by evaluating the $\tb$ variables at translates by sums of integer multiples of the elements
\be
 [x_i]:= \Big(x_i, \frac{x_i^2}{2}, \frac{x_i^3}{3}, \dots\Big).
\ee
 In terms of  $\HH =L^2(S^1)$, our discrete actions are generated by multiplication by
\be
S_i := \gamma_+([x_i])=e^{\sum_{i=1}^\infty \frac{1}{j}( x_i z )^j} = e^{-\ln (1-x_iz)} = \frac{1}{1-x_iz}, \quad i =0, \dots ,k.
\ee

Consider the multiplicative action of $S_i$  on  $\Gr_{\HH_+}(\HH) $ mapping  a space  $w$ to $(1-x_iz)^{-1} w$. This maps the spaces in $Gr(w_1, w_2) $ to $\Gr_{\HH_+}(\HH) $  and does not preserve $Gr(w_1, w_2) $.
It does however act on a multi-filtration on $\HH/w_1$ defined by the increasing order of poles of elements  at the  $x_\ell^{-1}, \ell=1,...,k$, augmenting  
 the order by one  for $(x_i)^{-1}$ and leaving the order unchanged at the other $(x_\ell)^{-1} $. 

Now consider  the action of  $S_{i_1}S_{i_2}...S_{i_l}$, with the $i_j$'s distinct, on an element $w \in Gr(w_1, w_2)$ 
of virtual dimension $-l$. The resulting plane $w(i_1, ..., i_l) = S_{i_1}S_{i_1}...S_{i_l}w$ is of virtual dimension $0$.
The null partition Pl\"ucker coordinate $\pi_\emptyset (S_{i_1}S_{i_1}...S_{i_l}w)$ is the coordinate   corresponding to the exterior product 
of all the basis elements $d_{i,1}$ (\ref{eq:basis_W_12_quotient}) of $w_2/w_1$. This in turn is  the Pl\"ucker coordinate  
of  $w$ corresponding  to  the  exterior product of the basis elements $d_{i,1}, i\notin \{i_1,i_2,..,i_l\}$
of $w_2/w_1$.

For any plane $\hat w$ of virtual dimension zero, the $\tau$-function $\tau_{\hat w}^{KP}$, evaluated at $0$, is the Pl\"ucker coordinate (\ref{tau_null_plucker}) of the subspace $\hat w\ss\HH$
  corresponding to the null partition; that is, up to a projective scaling factor,  the determinant (\ref{tau_fn_det_def})
 of the projection from $\hat w \in \Gr_{\HH_+}(\HH)$ to  the subspace $\HH_+$. From the definitions, the expressions 
 $H_w^{\alphab_{i_1}+\dots+\alphab_{i_l}}(0)$ of Proposition \ref{prop:FiniteGrassmannianOrigin}
 are also null partition Pl\"ucker coordinates,  but of an element $w$ shifted by $S_{i_1}...S_{i_l}$, 
up to a scaling factor.
 
 We have seen that these  are Pl\"ucker coordinates of the fixed plane $w$, 
and with the identifications made, it is easy to see to see that the relations 
 implied by the homomorphism (\ref{eq:pi_lambda_H_alpha_i_hom}) of Proposition (\ref{prop:pi_H_alpha_i_hom}) 
 become Pl\"ucker relations for the fixed plane $w$, with suitable rescalings. 
 These  are essential; one cannot adjust each Pl\"ucker coordinate individually. For the relations to correspond, one needs the  appropriate set of scalings  by defining  $H^{\nb}_w(\tb)$
 as in (\ref{eq:NormalizedEvaluationOfTauFunction}).

\section{Lattice mappings to $\Gr^\LL_{\HH_+}(\HH, \omega)$}
\label{sec:lattice_embeddings_lagrangian}

In this section we show that, restricting to KP $\tau$-functions corresponding
to the CKP hierarchy  and evaluating these on suitably defined lattices within their KP orbits, 
we can associate to every element $w^0\in \Gr_{\HH_+}^{\LL}(\HH,\omega)$
of $\Gr^\LL_{\HH_+}(\HH, \omega)$ a lattice of homomorphisms from the coordinate ring of
$\Gr^\LL_{\HH_+}(\HH, \omega)$ to a localization of the ring $\Cbb[[\yb, \tb']]$ of formal power series in a set of
 indeterminates $(\yb, \tb')$, where $\tb'=(t_1, 0, t_3, 0 \dots)$ is identified with the odd KP flow parameters.
We also prove the analogue of Theorem \ref{th:LatticeOfPointsInInfiniteGrassmannian} for the Lagrangian case.
For evaluations of the formal power series $\tau^{KP}_{w^0}$ that  are convergent (e.g., 
when $\tau^{KP}_{w^0}$ is a polynomial), our result implies that $\tau^{KP}_{w^0}$ gives rise to a lattice of points 
in $\Gr_{\HH_+}^{\LL}(\HH,\omega)$ whose Pl\"ucker coordinates are given by normalized evaluation of the $\tau$-function.

\subsection{Lattice evaluations of Lagrangian KP $\tau$-functions}
\label{sub:lattice_CKP_red}

Consider a family of formal parameters $\{y_i\}_{i\in\Nbb^+}$ labelled by positive integers.
 Denote by $\Cbb[[\yb,\tb']]$ the ring of formal power series in  the variables $y_i$ and
 the odd KP flow variables $\{t_{2j-1}\}_{j\in \Nbb^+}$. Similarly to (\ref{eq:NormalizedEvaluationOfTauFunction}),
 in order to define  the properly normalized evaluation of the $\tau$-function we have to include certain
denominators involving the variables $y_i$. To do this, consider the multiplicative subset of the ring $\Cbb[[\yb,\tb']]$
of formal power series generated by elements $\{y_i\pm y_j\}_{i,j\in \Nbb^+}$
\begin{equation}
S_{\yb}^\times=\langle y_i\pm y_j\;|\;i,j\in\Nbb^+\rangle \ \ss\  \Cbb[[\yb]] \ss  \Cbb[[\yb, \tb']].
\label{eq:MultiplicativeSubsetY}
\end{equation}
\begin{remark}
The generators of $S_{\yb}^\times$ correspond to roots of $\mathfrak{sp}(\infty)$.
\end{remark}
Since $\Cbb[[\yb,\tb']]$ is an integral domain, we can introduce its localization
$\left(S_{\yb}^\times\right)^{-1}\Cbb[[\yb,\tb']]$ with respect to $S_{\yb}^\times$.
Note that $\left(S_{\yb}^\times\right)^{-1}\Cbb[[\yb,\tb']]$ is itself an integral domain.

The following map of generators defines a surjective homomorphism of graded algebras
\begin{equation}
\begin{aligned}
\TT_{\grs\grp}:\Cbb[[\xb,\tb]]\twoheadrightarrow&\;\Cbb[[\yb,\tb']],\\[0.5em]
t_j\mapsto&\left\{\begin{array}{ll}
t_j,&j\equiv 1\bmod 2,\\[0.3em]
0,&j\equiv 0\bmod 2,
\end{array}\right.\\[0.5em]
x_i\mapsto&\left\{\begin{array}{ll}
y_{i+1},&i\geq0,\\[0.3em]
-y_{-i},&i<0.
\end{array}\right.
\end{aligned}
\label{eq:LHomomorphismDefinition}
\end{equation}
\begin{lemma}
The homomorphism $\TT_{\mathfrak{sp}}$ extends uniquely to a homomorphism of localized rings
\begin{equation}
\TT_{\mathfrak{sp}}: \left(S_{\xb}^\times\right)^{-1}\Cbb[[\xb,\tb]] \rightarrow \left(S_{\yb}^\times\right)^{-1}\Cbb[[\yb,\tb']].
\label{eq:LambdaExtended}
\end{equation}
\end{lemma}
\begin{proof}
Consider the multiplicative subset $S_{\xb}^\times\subset\Cbb[[\xb,\tb]]$ introduced in (\ref{eq:MultiplicativeSubsetX}) and note that its image is exactly the multiplicative set (\ref{eq:MultiplicativeSubsetY})
\begin{equation}
\TT_{\mathfrak{sp}}(S_{\xb}^\times) = S_{\yb}^\times.
\end{equation}
In other words, all elements of $\TT_{\mathfrak{sp}}(S_{\xb}^\times)$ are units of the localized ring $\left(S_{\yb}^\times\right)^{-1}\Cbb[[\yb,\tb']]$. It follows from the Universal Property of localization
that there exists a unique homomorphism between localized rings that makes the following diagram commutative
\begin{equation}
\begin{tikzcd}
\Cbb[[\xb,\tb]]\arrow[r,"\TT_{\mathfrak{sp}}"]\arrow[d,hook]&\left(S_{\yb}^\times\right)^{-1}\Cbb[[\yb,\tb']]\\
\left(S_{\xb}^\times\right)^{-1}\Cbb[[\xb,\tb]]\arrow[ru,dashed,"\exists!"']
\end{tikzcd}
\end{equation}
\end{proof}

Define the Lagrangian evaluation of the $\tau$-function as the $\TT_{\mathfrak{sp}}$-image of (\ref{eq:NormalizedEvaluationOfTauFunction}) for all $\nb\in\Ab$
\begin{equation}
\begin{aligned}
h^{\nb}_{w^0}:=&\TT_{\mathfrak{sp}}(H^{\nb}_{w^0})\\
=&\prod_{1\leq i< j}(y_i-y_j)^{n_{i-1}n_{j-1}+n_{-i}n_{-j}} \prod_{i,j=1}^{\infty}(-1)^{n_{-i}n_{j-1}}(y_i+y_j)^{n_{-i}n_{j-1}}
\\ &\qquad\times
\tau_{w^0}^{KP}\left(\tb'+\sum_{i=1}^{\infty}\big(n_{i-1}[y_i]+n_{-i}[-y_i]\big)\right).
\end{aligned}
\label{eq:LagrangianEvaluationTauFunction}
\end{equation}
\begin{remark}
We view $h_{w^0}^{\nb}\in(S_{\yb}^\times)^{-1}\Cbb[[\yb,\tb']]$ here
as an element of the localized ring of formal power series. Note, however, that in the simplest special case
when the element $w^0\in \Gr_{\HH_+}^{\LL}(\HH,\omega)$ of the infinite dimensional
Lagrangian Grassmannian is  such that $\tau^{KP}_{w^0}(\tb')$ is a polynomial
in the KP flow parameters, we can evaluate all parameters in (\ref{eq:LagrangianEvaluationTauFunction}).
If we then set $\{t_1,t_3,t_5,\dots\}$ to be arbitrary complex numbers and  the $y_i$'s  to be pairwise distinct complex
numbers satisfying $y_i+y_j\neq 0$, formula (\ref{eq:LagrangianEvaluationTauFunction}) defines a lattice of complex numbers.
\end{remark}

Consider the infinite dimensional sublattice
\begin{equation}
\Bb=\langle \betab_i:=\alphab_{i-1}-\alphab_{-i}\;|\;i\in\Nbb^+\rangle \ \subset\ \Ab,
\end{equation}
with generators labeled by the positive integers; i.e.,  $\Bb\ss \Ab$ consists of elements 
\be
\sum_{i\in \Zbb} n_i \alphab_i = \sum_{i\in \Nbb^+} n_{i-1} \betab_i \in \Ab
\ee
satisfying
\be
n_{i-1}   = -   n_{-i }, \quad  \forall \ i\in \Nbb^+.
\label{n_prime_i_sym}
\ee
We  denote such elements $\nb'$.
Similarly to the previous section, for every element $\nb'\in\Bb$ of the sublattice and every
partition $\lambda=(\ab|\bb)$
of Frobenius rank $r$ we introduce the  notation
\be
h_{w^0}^{\nb',\lambda}:=h^{\mb} =h_{w^0}^{\nb'+\sum_{i=1}^r \alphab_{a_i}-\sum_{i=1}^r\alphab_{-b_i-1}},
\label{eq:hn0lambdadefinition}
\ee
where
 \be
\mb := \nb'+\sum_{i=1}^r \alphab_{a_i}-\sum_{i=1}^r\alphab_{-b_i-1}
\label{eq:mb_def}
\ee
and $\nb'\in \Bb$  is an arbitrary element of the sublattice, with components $\{n'_i\}_{i\in\Zbb}$
 satisfying (\ref{n_prime_i_sym}).
Note that for a symmetric partition $\lambda=\lambda^T$, we have  $\{a_i=b_i\}_{i-1, \dots, r}$, so
$\mb$ is also in the sublattice $\Bb$. But for a general partition $\lambda$, $\mb\in \Ab$ is not in $\Bb$.

For each $i\in \Nbb^+$, define the map $\hat\delta_i^{\LL}$
which, evaluated on the generators, gives
\bea
\hat\delta_i^{\LL} : y_j&\& \mapsto y_j,\cr
\hat\delta_i^{\LL} : t_{2j+1}&\&\mapsto t_{2j+1}+\dfrac{2y_i^{2j+1}}{2j+1}.
\eea
\begin{lemma} The maps  $\hat\delta_i^{\LL}$ extend to homomorphisms
 \be
  \hat\delta_i^{\LL}:\Cbb[[\yb,\tb']]\rightarrow \Cbb[[\yb,\tb']],
  \ee
 that mutually commute
\begin{equation}
\hat\delta_i^{\LL}\hat\delta_j^{\LL} =\hat\delta_j^{\LL}\hat\delta_i^{\LL}, \quad \forall\  i,j\in\Nbb^+.
\end{equation}

\end{lemma}
They act on the evaluations of the $\tau$-function as shifts along the generators $\betab_i\in\Bb$ of the sublattice $\Bb\subset\Ab$. For a given element $w^0\in\Gr_{\HH_+}^{\LL}(\HH,\omega)$ of infinite dimensional Lagrangian Grassmannian, we have
\begin{equation}
\hat\delta_i^{\LL} h_{w^0}^{\nb}=h_{w^0}^{\nb+\beta_i}=\TT_{\mathfrak{sp}}\big(H_{w^0}^{\nb+\betab_i}\big)=\TT_{\mathfrak{sp}}\big(\hat\delta_{i-1}\hat\delta_{-i}^{-1}H_{w^0}^{\nb}\big)
\label{eq:DeltaLhAction}
\end{equation}
for all $\nb\in\Ab$ and $i\in\Nbb^+$


\subsection{Lagrangian condition on the big cell}
\label{sub:lagrangian_big_cell}

The analogue of Lagrangian condition on the big cell (\ref{eq:LagrangianConditionBigCellFinite})
for Lagrangian evaluations (\ref{eq:hn0lambdadefinition}) of the $\tau$-function is given by the following.
\begin{lemma}
\label{symmetric_h_lagrange_cond}
Let $w^0\in \Gr_{\HH_+}^{\LL}(\HH,\omega)$ be an element of the infinite Lagrangian Grassmannianian. We 
then have
\begin{equation}
h_{w^0}^{\0b,(i|j)}=h_{w^0}^{\0b,(j|i)}, \quad \forall \ i,j\in\Nbb^+.
\end{equation}
\label{lemm:SymmetryForEvaluationsOnTheBigCell}
\end{lemma}
\begin{proof}
\begin{equation}
\begin{aligned}
h_{w^0}^{\0b,(i|j)} \;\stackrel{(\ref{eq:hn0lambdadefinition})}{=}\;& h_{w^0}^{\alphab_{i-1}-\alphab_{-j}}\\[0.5em] \;\stackrel{(\ref{eq:LagrangianEvaluationTauFunction})}{=}\;& \frac{-\tau_{w^0}^{KP}\left(\tb'+[y_i]-[-y_j]\right)}{y_i+y_j} \\[0.5em] \;\stackrel{(\ref{eq:LagrangianIsotropyConditionTau})}{=}\;& \frac{-\tau_{w^0}^{KP}\left(\tb'+[y_j]-[-y_i]\right)}{y_j+y_i}=h_{w^0}^{\alphab_{j-1}-\alphab_{-i}}=h_{w^0}^{\0b,(j|i)}.
\label{eq:SymmetryProofEvaluationsAffineCell}
\end{aligned}
\end{equation}
\end{proof}

\begin{corollary}
Let $w^0\in\Gr_{\HH_+}^{\LL}(\HH,\omega)$ be a fixed point in the infinite dimensional Lagrangian Grassmannian such that $h_{w^0}^{\0b}\neq0$ is a nontrivial formal power series. The following identity holds in the field of fractions $\Fbb_{\yb,\tb'}$
\begin{equation}
\frac{h^{\0b,\lambda}_{w^0}}{h^{\0b}_{w^0}}=\det\left(\frac{h_{w^0}^{\mathbf\alphab_{\min(a_i,b_j)} -\mathbf\alphab_{-\max(a_i,b_j)-1}}}{h^{\0b}_{w^0}}\right)_{1\leq i,j\leq r}
\label{eq:h0lambdaeffectivebigcell}
\end{equation}
\end{corollary}
\begin{proof}
Apply homomorphism $\TT_{\mathfrak{sp}}$ to both sides of (\ref{eq:AffineDeterminantRelationLattice}) and use Lemma \ref{lemm:SymmetryForEvaluationsOnTheBigCell}.
\end{proof}


\subsection{Lattice points in Lagrangian Grassmannian}
\label{sub:lattice_lagrangian}

Lemma \ref{symmetric_h_lagrange_cond} gives a condition, valid on the big cell,  for a subspace to be Lagrangian;
In what follows  it is explained why the result on the big cell is sufficient to imply it in general (cf.~remark \ref{big_cell_ring_extension}).
Our first step is to show that the evaluation (\ref{eq:LagrangianEvaluationTauFunction}) gives rise to a family of homomorphisms of the coordinate rin
identifying the homomorphism that corresponds to the origin $\0b\in\Bb$ of the sublattice.
\begin{proposition}
Fix an element $w^0\in \Gr^{\LL}_{\HH_+}(\HH,\omega)$ of the infinite dimensional Lagrangian Grassmannian such that $\tau^{KP}_{w^0}(\tb')\neq 0$ is a nontrivial formal power series. The following map, defined on the linear basis elements $\pi_\lambda\in \SS(\Gr^{\LL}_{\HH_+}(\HH,\omega))$,   generates  a homomorphism of commutative algebras
\bea
\Psi_{w^0}^{\0b}:\SS(\Gr^{\LL}_{\HH_+}(\HH,\omega))&\&\rightarrow \left(S_{\yb}^\times\right)^{-1}\Cbb[[\yb,\tb']],\cr
&\& \cr
\Psi_{w^0}^{\0b}: \pi_\lambda&\&\mapsto h_{w^0}^{\0b,\lambda}.
\eea
\label{prop:LagrngianGrassmannianianHomomorphismOrigin}
\end{proposition}
\begin{proof}
Consider the homomorphism $\widetilde\Psi_{w^0}^{\0b}$ defined by the same action on generators of the polynomial ring $\SS$
\bea
\widetilde{\Psi}_{w^0}^{\0b}:\SS&\&\rightarrow \left(S_{\yb}^\times\right)^{-1}\Cbb[[\yb,\tb']],\cr
&\& \cr
\widetilde{\Psi}_{w^0}^{\0b}:\pi_\lambda&\&\mapsto h_{w^0}^{\0b,\lambda}.
\label{eq:PsiTildeZeroDefinition}
\eea
We need to show that $\Psi_{w^0}^{\0b}$ annihilates all Pl\"ucker and Lagrange relations. For the Pl\"ucker relations, simply note that
\begin{equation}
\widetilde{\Psi}_{w^0}^{\0b}=\TT_{\mathfrak{sp}}\circ\widetilde{\Phi}_{w^0}^{\0b},
\end{equation}
is a composition of two ring homomorphisms. By Theorem \ref{th:LatticeOfPointsInInfiniteGrassmannian},  $\widetilde{\Phi}_{w^0}^{\0b}$ annihilates all Pl\"ucker relations and, hence, so does the $\widetilde{\Psi}_{w^0}^{\0b}$.

It remains to prove that $\widetilde{\Psi}_{w^0}^{\0b}$ annihilates all Lagrange relations. Recall that the ring $(S_{\yb}^\times)^{-1}\Cbb[[\yb,\tb']]$ is a localization of an integral domain, so it must itself be an integral domain. Denote its field of fractions by $\mathbb F_{\yb,\tb}$. By the assumption of the Proposition, the $\tb'$-restriction $\tau_{w^0}^{KP}(\tb')\in\Cbb[[\tb']]\subset\mathbb F_{\yb,\tb'}$ of the KP hierarchy $\tau$-function is a nontrivial formal power series, and hence, an invertible element of the field of fractions.
\color{black}
By the universal property of localization, we have a unique homomorphism
\begin{equation}
\overline \Psi_{w^0}^{\0b}: \Cbb\big[\pi_\emptyset,\pi_\emptyset^{-1}\big]\big[\pi_{(i-1|j-1)}\,\big|\,1\leq i\leq j\big]\rightarrow \mathbb F_{\yb,\tb'}
\end{equation}
obtained as a localization of the homomorphism $\widetilde{\Psi}_{w^0}^{\0b}$ restricted to the subring
\begin{equation}
\Cbb\big[\pi_\emptyset\big]\big[\pi_{(i-1|j-1)}\,\big|\, 1\leq i\leq j\big]\subset\SS
\end{equation}
generated by Pl\"ucker ccoordinates corresponding to hook  partitions; i.e., those of Frobenius rank 1. Consider the following diagram
\begin{equation}
\begin{tikzcd}
\SS=\Cbb[\pi_\lambda\,|\,\lambda\in\mathbf \Lambda] \arrow[rr,"\widetilde\Psi_{w^0}^{\0b}"] \arrow[dd,"\xi^{\LL}"] \arrow[ddrr,dashed,"\Delta"] &&(S_{\yb}^\times)^{-1}\Cbb[[\yb,\tb']] \arrow[dd,hook]\\\\
\Cbb\big[\pi_\emptyset,\pi_\emptyset^{-1}\big]\big[\pi_{(i-1|j-1)}\,\big|\,0<i\leq j\big]\arrow[rr,"\overline\Psi_{w^0}^{\0b}"]&&\mathbb F_{\yb,\tb'}
\end{tikzcd}
\label{eq:CommutativeDiagramLagrangianF}
\end{equation}
Its commutatative follows from comparing the actions on the generators $\pi_{\lambda}$ 
for all partitions $\lambda=(\ab|\bb)$, of any Frobenius rank $r$. This gives
\begin{equation}
\begin{aligned}
\overline\Psi_{w^0}^{\0b}(\xi^{\LL}(\pi_\lambda)) \;\;\stackrel{(\ref{eq:XiLDefinition})}{=}&\;\; \overline\Psi_{w^0}^{\0b}\left(\pi_\emptyset\, \det\Big(\frac{\pi_{(\min(a_i,b_j)|\max(a_i,b_j))}}{\pi_\emptyset}\Big)_{1\leq i,j\leq r}\right)\\[0.5em]
\;\;\stackrel{(\ref{eq:PsiTildeZeroDefinition})}{=}&\;\; h_{w^0}^{\0b,\emptyset}\, \det\left(\frac{h_{w^0}^{\0b,(\min(a_i,b_j)|\max(a_i,b_j))}}{h_{w^0}^{\0b,\emptyset}}\right)_{1\leq i,j\leq r}\\[0.5em]
\;\;\stackrel{(\ref{eq:hn0lambdadefinition})}{=}&\;\; h_{w^0}^{\0b}\, \det\left(\frac{h_{w^0}^{\alphab_{\min(a_i,b_j)}-\alphab_{-\max(a_i,b_j)-1}}}{h_{w^0}^{\0b}}\right)_{1\leq i,j\leq r}\\[0.5em]
\;\;\stackrel{(\ref{eq:h0lambdaeffectivebigcell})}{=}&\;\; h_{w^0}^{\0b,\lambda}=\widetilde\Psi_{w^0}^{\0b}(\pi_\lambda).
\end{aligned}
\end{equation}

By Lemma \ref{lemm:XiLAnnihilatesIL},  $\xi^{\LL}$ annihilates the  $\II^\LL$ defining Lagrangian Grassmannians.
Hence
\begin{equation}
\II^\LL\subset\ker\Delta
\label{eq:ILAnnihilatedByDelta}
\end{equation}
is in the kernel of the diagonal map $\Delta$ of diagram (\ref{eq:CommutativeDiagramLagrangianF}). But $\Delta$ is the composition 
of an injective map with $\widetilde\Psi_{w^0}^{\0b}$, so
\begin{equation}
\ker\Delta=\ker\widetilde\Psi_{w^0}^{\0b}.
\label{eq:KerDeltaEqualsKerPsi}
\end{equation}
Combining (\ref{eq:ILAnnihilatedByDelta}) with (\ref{eq:KerDeltaEqualsKerPsi}) completes the proof.
\end{proof}

This brings us to the main result of this section.
\begin{theorem}
Fix an element $w^0\in\Gr^{\LL}_{\HH_+}(\HH,\omega)$ of the infinite Lagrangian Grassmannian such that $\tau_{w^0}^{KP}(\tb')\neq0$.\footnote{We expect that the condition  $\tau_{w^0}^{KP}(\tb')\neq0$  is automatically satisfied for all elements of the infinite Lagrangian Grassmannian,
whatever the value of $\tb'$ is. However a proof of this remains to be found. We therefore impose this  condition as an additional requirement in the formulation of the Theorem.} To every element $\nb_0'\in\Bb$ of the sublattice there is a homomorphism of commutative algebras
generated by the following evaluation on linear basis elements
\bea
\Psi_{w^0}^{\nb_0'}:\SS(\Gr^{\LL}_{\HH_+}(\HH,\omega)) &\&\rightarrow (S_{\yb}^\times)^{-1}\Cbb[[\yb,\tb']],\cr
&\& \cr
 \Psi_{w^0}^{\nb_0'}:\pi_\lambda &\& \mapsto h_{w^0}^{\nb_0',\lambda}.
\eea
\label{th:LatticeOfPointsInInfiniteLagrangianGrassmannian}
\end{theorem}
\begin{proof}
We prove this by induction on the sublattice $\Bb \ss \Ab$. The starting point is  $\nb_0'=\0b\in\Bb$, which is given by Proposition \ref{prop:LagrngianGrassmannianianHomomorphismOrigin}. Now suppose that the statement of the Theorem holds for 
$\nb_0'\in\Bb$ and let $i\in\Nbb^+$ be a positive integer. Note that for all partitions $\lambda$ we have
\begin{equation}
\Psi_{w^0}^{\nb_0'+\beta_i}(\pi_\lambda)=h^{\nb_0'+\beta_i,\lambda}_{w^0} \;\;\stackrel{(\ref{eq:DeltaLhAction})}{=}\;\; \hat\delta^{\LL}_i(h^{\nb_0',\lambda}_{w^0}) =\hat\delta^{\LL}_i(\Psi_{w^0}^{\nb_0'}(\pi_\lambda)).
\end{equation}
Hence
\begin{equation}
\Psi_{w^0}^{\nb_0'+\beta_i}=\hat\delta_i^{\LL}\circ \Psi_{w^0}^{\nb_0'}
\end{equation}
is a composition of two ring homomorphisms and therefore must itself be a ring homomorphism.
By induction on the sublattice $\Bb$,  this implies that $\mathbf\Psi_{w^0}^{\nb_0'}$ is a homomorphism for all $\nb_0'\in\mathbb B$,
concluding the proof.
\end{proof}


\section{Applications to Lattice Integrable Systems}
\label{sec:applications_lattice}

\subsection{Kenyon-Pemantle recurrences for $\Gr_{\HH_+}(\HH)$}
\label{hexahedron_ordinary_grassmannian}

\begin{lemma}
For every ordered  triple of nonnegative integers $i,j,k\in\Nbb$, \ $i<j<k$, the cubic 
\begin{equation}
\kappa_1:=\pi_{(\emptyset|\emptyset)}(\pi_{(i|i)}\pi_{(k,j|k,j)}+\pi_{(j|k)}\pi_{(k,i|j,i)}) -\pi_{(i|i)}\pi_{(j|j)}\pi_{(k|k)}+\pi_{(i|j)}\pi_{(j|k)}\pi_{(k|i)}.
\label{eq:Kappa1RelationDef}
\end{equation}
\label{lemm:Kappa1Relation}
 belongs to the defining ideal  $\II$ of the coordinate ring $\SS(\Gr_{\HH_+}(\HH))=\SS/\II$ of $\Gr_{\HH_+}(\HH)$.
\end{lemma}
\begin{proof}
(\ref{eq:Kappa1RelationDef}) can be written as
\begin{equation}
\kappa_1=\pi_{(j|k)}\rho_1+\pi_{(\emptyset|\emptyset)}\rho_2,
\end{equation}
where
\begin{subequations}
\begin{eqnarray}
&&\rho_1:=\pi_{(i|j)}\pi_{(k|i)}-\pi_{(k|j)}\pi_{(i|i)}+\pi_{(\emptyset|\emptyset)}\pi_{(k,i|j,i)},
\label{eq:rho1Plucker} \\[0.3em] &&\rho_2:=\pi_{(j|k)}\pi_{(k|j)}-\pi_{(j|j)}\pi_{(k|k)}+\pi_{(\emptyset|\emptyset)}\pi_{(k,j|k,j)}
 \label{eq:rho2Plucker}
\end{eqnarray}
\end{subequations}
are (short) Pl\"ucker quadratic forms  belonging to $\II$. Therefore, so does $\kappa_1$.
\end{proof}

Since the Pl\"ucker relations  are invariant under transposition of all partitions, (\ref{eq:Kappa1RelationDef}) also implies
the dual relation:
\begin{lemma}
For every ordered  triple of nonnegative integers $i,j,k\in\Nbb^+$, \ $i<j<k$, the cubic
\begin{equation}
\kappa_1^T:=\pi_{(\emptyset|\emptyset)}(\pi_{(i|i)}\pi_{(k,j|k,j)}+\pi_{(k|j)}\pi_{(j,i|k,i)}) -\pi_{(i|i)}\pi_{(j|j)}\pi_{(k|k)}+\pi_{(j|i)}\pi_{(k|j)}\pi_{(i|k)}
\label{eq:Kappa1RelationDualDef}
\end{equation}
 also belongs to $\II$.
\end{lemma}

Since the Pl\"ucker ideal is invariant under permutations of basis elements of $\HH$,
four more relations of the form (\ref{eq:Kappa1RelationDef}), (\ref{eq:Kappa1RelationDualDef}) follow.
 Let $i,j,k\in\Nbb^+$ be an ordered triple of nonnegative integers $i<j<k$. Consider the 
 permutations $\sigma_1^{(i,j,k)}$, $\sigma_2^{(i,j,k)}$ of the basis elements $\{e_l\}_{l\in\Zbb}$ of  $\HH$ given by
\begin{equation}
\sigma_1^{(i,j,k)}(e_l)=\left\{\begin{array}{ll}
e_j,& \text{if }\  l=i,\\
e_k,& \text{if }\  l=j,\\
e_i,& \text{if }\  l=k,\\
e_{-j-1},& \text{if }\  l=-i-1,\\
e_{-k-1},& \text{if }\  l=-j-1,\\
e_{-i-1},& \text{if }\  l=-k-1,\\
e_l,&\textrm{otherwise},
\end{array}\right.
\qquad
\sigma_2^{(i,j,k)}(e_l)=\left\{\begin{array}{ll}
e_{-k-1},& \text{if }\  l=i,\\
e_{-j-1},&  \text{if }\  l=j,\\
e_{-i-1},& \text{if }\  l=k,\\
e_k,& \text{if }\  l=-i-1,\\
e_j,& \text{if }\  l=-j-1,\\
e_i,& \text{if }\  l=-k-1,\\
e_l,&\textrm{otherwise}.
\end{array}\right.
\label{sigma_perm_automorph}
\end{equation}
These induce invertible endomorphisms $\hat\sigma_1^{(i,j,k)},\hat\sigma_2^{(i,j,k)}\in\mathrm{End}(\FF)$ of the Fock space,
which are just the Cliffford representations of these as $\GL(\HH)$ elements.
 The dual action on the Pl\"ucker coordinates gives  a pair of automorphisms
\bea
\hat\sigma_1^{(i,j,k)*}&\&:\SS\rightarrow\SS,  \cr
\hat\sigma_2^{(i,j,k)*}&\&:\SS\rightarrow\SS 
\eea 
 of the polynomial ring $\SS$  that preserve the Pl\"ucker ideal $\II$ and thus descend to automorphisms of the coordinate ring 
 $\SS(\Gr_{\HH_+}(\HH))=\SS/\II$:
\bea
\hat\sigma_1^{(i,j,k)*}&\&:\SS(\Gr_{\HH_+}(\HH))\rightarrow\SS(\Gr_{\HH_+}(\HH)), \cr
\hat\sigma_2^{(i,j,k)*}&\&:\SS(\Gr_{\HH_+}(\HH))\rightarrow\SS(\Gr_{\HH_+}(\HH)).
\label{eq:AutomorphismSimga_ijk}
\eea
\begin{proposition}
The following seven elements belong to the defining ideal $\II$ of  $\SS(\Gr_{\HH_+}(\HH))$.
\begin{subequations}
\bea
\kappa_1&\&:=\pi_{(\emptyset|\emptyset)}(\pi_{(i|i)}\pi_{(k,j|k,j)}+\pi_{(j|k)}\pi_{(k,i|j,i)}) -\pi_{(i|i)}\pi_{(j|j)}\pi_{(k|k)}+\pi_{(i|j)}\pi_{(j|k)}\pi_{(k|i)}
\label{eq:Kappa1Relation}\\
\kappa_2&\&:=\pi_{(\emptyset|\emptyset)}(\pi_{(j|j)}\pi_{(k,i|k,i)}-\pi_{(k|i)}\pi_{(j,i|k,j)}) -\pi_{(i|i)}\pi_{(j|j)}\pi_{(k|k)}+\pi_{(i|j)}\pi_{(j|k)}\pi_{(k|i)}
\label{eq:Kappa2Relation}\\
\kappa_3&\&::=\pi_{(\emptyset|\emptyset)}(\pi_{(k|k)}\pi_{(j,i|j,i)}+\pi_{(i|j)}\pi_{(k,j|k,i)}) -\pi_{(i|i)}\pi_{(j|j)}\pi_{(k|k)}+\pi_{(i|j)}\pi_{(j|k)}\pi_{(k|i)}
\label{eq:Kappa3Relation}\\
\kappa_1^T&\&:=\pi_{(\emptyset|\emptyset)}(\pi_{(i|i)}\pi_{(k,j|k,j)}+\pi_{(k|j)}\pi_{(j,i|k,i)}) -\pi_{(i|i)}\pi_{(j|j)}\pi_{(k|k)}+\pi_{(j|i)}\pi_{(k|j)}\pi_{(i|k)}
\label{eq:Kappa1TRelation}\\
\kappa_2^T&\&:=\pi_{(\emptyset|\emptyset)}(\pi_{(j|j)}\pi_{(k,i|k,i)}-\pi_{(i|k)}\pi_{(k,j|j,i)}) -\pi_{(i|i)}\pi_{(j|j)}\pi_{(k|k)}+\pi_{(j|i)}\pi_{(k|j)}\pi_{(i|k)}
\label{eq:Kappa2TRelation}\\
\kappa_3^T&\&:=\pi_{(\emptyset|\emptyset)}(\pi_{(k|k)}\pi_{(j,i|j,i)}+\pi_{(j|i)}\pi_{(k,i|k,j)}) -\pi_{(i|i)}\pi_{(j|j)}\pi_{(k|k)}+\pi_{(j|i)}\pi_{(k|j)}\pi_{(i|k)}
\label{eq:Kappa3TRelation}\\
\kappa_0&\&:=\pi_{(\emptyset|\emptyset)}^2\pi_{(k,j,i|k,j,i)}\pi_{(i|j)}\pi_{(j|k)}\pi_{(k|i)} -(\pi_{(i|j)}\pi_{(j|k)}\pi_{(k|i)})^2\\
&\&+(\pi_{(i|j)}\pi_{(j|k)}\pi_{(k|i)}) \left(2\pi_{(i|i)}\pi_{(j|j)}\pi_{(k|k)} 
-\pi_{(\emptyset|\emptyset)}(\pi_{(i|i)}\pi_{(k,j|k,j)} +\pi_{(j|j)}\pi_{(k,i|k,i)} +\pi_{(k|k)}\pi_{(j,i|j,i)})\right)\cr
&\&-(\pi_{(i|i)}\pi_{(j|j)} -\pi_{(\emptyset|\emptyset)}\pi_{(j,i|j,i)})(\pi_{(i|i)}\pi_{(k|k)} 
-\pi_{(\emptyset|\emptyset)}\pi_{(k,i|k,i)})(\pi_{(j|j)}\pi_{(k|k)} -\pi_{(\emptyset|\emptyset)}\pi_{(k,j|k,j)})
\cr
&\& 
\label{eq:Kappa0Relation}
\eea
\end{subequations}
for every triple $i,j,k\in\Nbb^+$ of nonnegative integers satisfying $i<j<k$.
\label{prop:hexahedronRelations}
\end{proposition}
\begin{proof}
Relations (\ref{eq:Kappa1Relation})--(\ref{eq:Kappa3Relation}) are  the orbit of (\ref{eq:Kappa1RelationDef}) under
the three-element cyclic group of automorphisms generated by $\hat\sigma_1^{(i,j,k)*}$ and  relations (\ref{eq:Kappa1TRelation})--(\ref{eq:Kappa3TRelation}) are the orbit of (\ref{eq:Kappa1RelationDualDef}) under the same group.
To prove (\ref{eq:Kappa0Relation}),first  note that applying the automorphism $\hat\sigma_2^{(i,j,k)*}$ to relation (\ref{eq:Kappa1Relation}) gives
\bea
\widetilde\kappa_1&\&:=\hat\sigma_2^{(i,j,k)*}(\kappa_1)=\pi_{(k,j,i|k,j,i)} (\pi_{(k|k)}\pi_{(j,i|j,i)}+\pi_{(i|j)}\pi_{(k,j|k,i)})\cr
&\&{\hskip 90 pt}  +\pi_{(j,i|k,j)}\pi_{(k,i|j,i)}\pi_{(k,j|k,i)} -\pi_{(j,i|j,i)}\pi_{(k,i|k,i)}\pi_{(k,j|k,j)}.
\eea

It follows from Proposition \ref{prop:SOverI_IntegralDomain} that $\SS(\Gr_{\HH_+}(\HH))=\SS/\II$ is an integral domain. In the corresponding field of fractions we can solve the relations $\kappa_1=0$, $\kappa_2=0$,
$\kappa_3=0$ for $\pi_{(k,i|j,i)}$, $\pi_{(j,i|k,j)}$, $\pi_{(k,j|k,i)}$, respectively,
 and substitute the solutions into relation $\widetilde\kappa_1=0$ to get the following identity in the fraction field
\begin{equation}
(\pi_{(i|j)}\pi_{(j|k)}\pi_{(k|i)}-\pi_{(i|i)}\pi_{(j|j)}\pi_{(k|k)})\kappa_0=0.
\end{equation}
The first factor is nonzero, as can be seen by computing it in terms of determinants of homogeneous
coordinates on a sufficiently large finite Grassmannian $\Gr_{V_n}(\HH_{n,2n})$ (i.e. with $n\ge k+1$).

 Finally, using the fact that $\SS(\Gr_{\HH_+}(\HH))$ is an integral domain, 
it follows that the relation $\kappa_0=0$ holds in this ring and hence $\kappa_0\in\II$.
\end{proof}

Now fix  an element $w\in\Gr_{\HH_+}(\HH)$ and define the infinite family of functions on the sublattice $\Bb\subset\Ab$,  
consisting of the function 
\bea
\Hb:\Bb &\&\rightarrow\left(S_x^\times\right)^{-1}\Cbb[[\xb,\tb]], \cr
\Hb:\nb' &\& \mapsto H_w^{\nb',\emptyset}
\eea
corresponding to the null partition, together with infinitely many functions labelled by 
nonsymmetric hook partitions $(p|q)$,  $p\neq q$ with values
\bea
\Hb^{(p|q)}:\Bb&\&\rightarrow\left(S_x^\times\right)^{-1}\Cbb[[\xb,\tb]],\cr
\Hb^{(p|q)}:\nb'&\& \mapsto H^{\nb',(p|q)}_w,
\eea
Also define, for  $\vb'\in \Bb$ the lattice functions $\Hb _{\vb'}, \Hb^{(p|q)} _{\vb'}$  with shifted arguments 
having values:
\bea
\Hb _{\vb'}:\Bb &\&\ra \left(S_x^\times\right)^{-1}\Cbb[[\xb,\tb]],\cr
&\& \cr
\Hb_{\vb'}: \nb'  &\&\mapsto H^{\nb'+\vb', \emptyset}_w,  \cr
&\& \cr
\Hb^{(p|q)} _{\vb'}:\Bb &\&\ra \left(S_x^\times\right)^{-1}\Cbb[[\xb,\tb]],\cr
&\& \cr
\Hb^{(p|q)}_{\vb'}: \nb'  &\&\mapsto H^{\nb'+\vb', (p|q)}_w,  \quad
\forall \  \vb' \in \Bb,  \quad p, q \in \Nbb.
\eea
\begin{theorem}
\label{th:hexahedron_relations_signs}
For every ordered triple $(i,j,k)$ of nonnegative integers,  $i<j<k$, these lattice functions satisfy the system of recurrence relations
\begin{subequations}
\bea
\Hb\Hb^{(j|k)}\Hb_{\betab_i}^{(k|j)}&\&=-\Hb\Hb_{\betab_i}\Hb_{\betab_j+\betab_k} +\Hb_{\betab_i}\Hb_{\betab_j}\Hb_{\betab_k}-\Hb^{(i|j)}\Hb^{(j|k)}\Hb^{(k|i)},
\\
\Hb\Hb^{(k|i)}\Hb_{\betab_j}^{(i|k)}&\&=\Hb\Hb_{\betab_j}\Hb_{\betab_i+\betab_k} -\Hb_{\betab_i}\Hb_{\betab_j}\Hb_{\betab_k}+\Hb^{(i|j)}\Hb^{(j|k)}\Hb^{(k|i)},
\\
\Hb\Hb^{(i|j)}\Hb^{(j|i)}_{\betab_k} &\&=-\Hb\Hb_{\betab_k}\Hb_{\betab_i+\betab_j} +\Hb_{\betab_i}\Hb_{\betab_j}\Hb_{\betab_k}-\Hb^{(i|j)}\Hb^{(j|k)}\Hb^{(k|i)},
\\
\Hb\Hb^{(k|j)}\Hb^{(j|k)}_{\betab_i}&\&=-\Hb\Hb_{\betab_i}\Hb_{\betab_j+\betab_k} +\Hb_{\betab_i}\Hb_{\betab_j}\Hb_{\betab_k}-\Hb^{(j|i)}\Hb^{(k|j)}\Hb^{(i|k)},
\\
\Hb\Hb^{(i|k)}\Hb^{(k|i)}_{\betab_j}&\&=\Hb\Hb_{\betab_j}\Hb_{\betab_i+\betab_k} -\Hb_{\betab_i}\Hb_{\betab_j}\Hb_{\betab_k}+\Hb^{(j|i)}\Hb^{(k|j)}\Hb^{(i|k)},
\\
\Hb\Hb^{(j|i)}\Hb^{(i|j)}_{\betab_k}&\&=-\Hb\Hb_{\betab_k}\Hb_{\betab_i+\betab_j} +\Hb_{\betab_i}\Hb_{\betab_j}\Hb_{\betab_k}-\Hb^{(j|i)}\Hb^{(k|j)}\Hb^{(i|k)},\\
\Hb^2\Hb^{(i|j)}\Hb^{(j|k)}\Hb^{(k|i)}\Hb_{\betab_i+\betab_j+\betab_k}&\& =
 (\Hb^{(i|j)}\Hb^{(j|k)}\Hb^{(k|i)})^2-(\Hb^{(i|j)}\Hb^{(j|k)}\Hb^{(k|i)})
\cr
&\&\times\Big(2\Hb_{\betab_i}\Hb_{\betab_j}\Hb_{\betab_k}-\Hb(\Hb_{\betab_i} \Hb_{\betab_j+\betab_k}+\Hb_{\betab_j}\Hb_{\betab_i+\betab_k} +\Hb_{\betab_k}\Hb_{\betab_i+\betab_j})\Big)
\cr
&\&+(\Hb_{\betab_i}\Hb_{\betab_j} -\Hb\Hb_{\betab_i+\betab_j})(\Hb_{\betab_i} \Hb_{\betab_k} -\Hb\Hb_{\betab_i+\betab_k})(\Hb_{\betab_j}\Hb_{\betab_k} -\Hb\Hb_{\betab_j+\betab_k}).
\cr
&\&
\eea
\label{eq:hexahedronRecurrenceRawSigns}
\end{subequations}
\end{theorem}
\begin{proof}
Combine Proposition \ref{prop:hexahedronRelations} with Theorem \ref{thm:plucker_lattice_homomorph}.
\end{proof}

Denote the sum of nonvanishing components of $\nb'$ as
\begin{equation}
\mathrm{height}(\nb'):=\sum_{j=-\infty}^{\infty}n'_j
\end{equation}
and suppose that
\begin{enumerate}
\item The values   of  $\Hb$ are known for all $\nb'\in\Bb$ with
$l\leq \mathrm{height}(\nb')\leq l+2$.
\item The values of the functions $\Hb^{(p,q)}$ are known for all $\nb'\in\Bb$ with 
$\mathrm{height}(\nb')=l$.
\end{enumerate}
If all known values are nonzero,  eqs.~(\ref{eq:hexahedronRecurrenceRawSigns}) allow us to determine 
the values  $H^{\nb'}_w$ for $\mathrm{height}(\nb')=l+3$ and the values  $H^{\nb', (p|q)}_w$ for $\mathrm{height}(\nb')=l+1$ as Laurent polynomials in the known values. However, as written, the recurrence relations 
(\ref{eq:hexahedronRecurrenceRawSigns}) contain negative coefficients.

However, if we restrict to elements $\nb'\in \Bb_3$ of the $3$-dimensional sublattice
\begin{equation}
\Bb_3:=\mathrm{Span}\{\betab_1,\betab_2,\betab_3\}\subset\Bb,
\end{equation}
we can eliminate the $(-)$ signs by defining the following values for all $\nb'\in\Bb_3$
\begin{equation}
\begin{aligned}
\Hc^{\nb'}:=&(-1)^{n'_1+n'_2+n'_3+n'_1n'_2+n'_1n'_3+n'_2n'_3}H^{\nb'}_w,\\
\Hc^{\nb', (x)}:=&(-1)^{(n'_1)^2+n'_2n'_3+(n'_3)^2}H^{\nb', (1|2)}_w,\\
\Hc^{\nb', (-x)}:=&(-1)^{(n'_1)^2+n'_2n'_3+(n'_3)^2}H^{\nb', (2|1)}_w,\\
\Hc^{\nb', (y)}:=&(-1)^{n'_1n'_3+(n'_3)^2}H^{\nb', (2|0)}_w,\\
\Hc^{\nb', (-y)}:=&(-1)^{n'_1n'_3+(n'_3)^2}H^{\nb', (0|2)}_w,\\
\Hc^{\nb', (z)}:=&(-1)^{n'_1n'_2+(n'_2)^2+(n'_3)^2}H^{\nb', (0|1)}_w,\\
\Hc^{\nb',(-z)}:=&(-1)^{n'_1n'_2+(n'_2)^2+(n'_3)^2}H^{\nb', (1|0)}_w.
\end{aligned}
\end{equation}

Denote the lattice translated values of the corresponding seven functions by
\begin{equation}
\Hc^{\nb'}_{\vb'}:=\Hc^{\nb'+\vb'},\quad \Hc^{\nb', (\cdot)}_{\vb}:=\Hc^{\nb'+\vb', (\cdot)},\quad \nb',\vb' \in\Bb_3.
\end{equation}
\begin{corollary}
\label{cor:hexahedron_evaluations_general}
For every $w\in\Gr_{\HH_+}(\HH)$ these lattice functions give  solutions of the following system of recurrence relations
\begin{subequations}
\bea
\Hbc\Hbc^{(x)}\Hbc_{\betab_1}^{(-x)}&\&=\Hbc\Hbc_{\betab_1}\Hbc_{\betab_2+\betab_3} +\Hbc_{\betab_1}\Hbc_{\betab_2}\Hbc_{\betab_3}+\Hbc^{(x)}\Hbc^{(y)}\Hbc^{(z)},
\\
\Hbc\Hbc^{(y)}\Hbc_{\betab_2}^{(-y)}&\&=\Hbc\Hbc_{\betab_2}\Hbc_{\betab_1+\betab_3} +\Hbc_{\betab_1}\Hbc_{\betab_2}\Hbc_{\betab_3}+\Hbc^{(x)}\Hbc^{(y)}\Hbc^{(z)},
\\
\Hbc\Hbc^{(z)}\Hbc^{(-z)}_{\betab_3}&\&=\Hbc\Hbc_{\betab_3}\Hbc_{\betab_1+\betab_2} +\Hbc_{\betab_1}\Hbc_{\betab_2}\Hbc_{\betab_3}+\Hbc^{(x)}\Hbc^{(y)}\Hbc^{(z)},
\\
\Hbc\Hbc^{(-x)}\Hbc_{\betab_1}^{(x)}&\&=\Hbc\Hbc_{\betab_1}\Hbc_{\betab_2+\betab_3} +\Hbc_{\betab_1}\Hbc_{\betab_2}\Hbc_{\betab_3}+\Hbc^{(-x)}\Hbc^{(-y)}\Hbc^{(-z)},
\\
\Hbc\Hbc^{(-y)}\Hbc_{\betab_2}^{(y)}&\&=\Hbc\Hbc_{\betab_2}\Hbc_{\betab_1+\betab_3} +\Hbc_{\betab_1}\Hbc_{\betab_2}\Hbc_{\betab_3}+\Hbc^{(-x)}\Hbc^{(-y)}\Hbc^{(-z)},
\\
\Hbc\Hbc^{(-z)}\Hbc^{(z)}_{\betab_3}&\&=\Hbc\Hbc_{\betab_3}\Hbc_{\betab_1+\betab_2} +\Hbc_{\betab_1}\Hbc_{\betab_2}\Hbc_{\betab_3}+\Hbc^{(-x)}\Hbc^{(-y)}\Hbc^{(-z)},\\
\Hbc^2\Hbc^{(x)}\Hbc^{(y)}\Hbc^{(z)}\Hbc_{\betab_1+\betab_2+\betab_3}&\&= (\Hbc^{( x)}\Hbc^{(y)}\Hbc^{(z)})^2
+(\Hbc^{(x)}\Hbc^{(y)}\Hbc^{(z)})  \cr
&\&\times \left(2\Hbc_{\betab_1}\Hbc_{\betab_2}\Hbc_{\betab_3} +\Hbc(\Hbc_{\betab_1} \Hbc_{\betab_2+\betab_3}+\Hbc_{\betab_2}\Hbc_{\betab_1+\betab_3} +\Hbc_{\betab_3}\Hbc_{\betab_1+\betab_2})\right)\cr
&\&+(\Hbc_{\betab_1}\Hbc_{\betab_2} +\Hbc\Hbc_{\betab_1+\betab_2})(\Hbc_{\betab_1} \Hbc_{\betab_3} +\Hbc\Hbc_{\betab_1+\betab_3})(\Hbc_{\betab_2}\Hbc_{\betab_3} +\Hbc\Hbc_{\betab_2+\betab_3}).\cr
&\&
\eea
\label{eq:hexahedronRecurrenceGr36}
\end{subequations}
\label{cor:hexahedronRecurrenceGr36}
\end{corollary}

Note that, when evaluated at lattice points $\nb'$,  eq.~(\ref{eq:hexahedronRecurrenceGr36})
can be understood  as containing two subcases: the first is when
\begin{equation}
\Hc^{\nb'},\;\left\{\begin{array}{ll}
\Hc^{\nb', (x)},\; \Hc^{\nb', (y)},\; \Hc^{\nb', (z)},&\text{and height}(\nb')\;\textrm{is even},\\[0.3em]
\Hc^{\nb', (-x)},\; \Hc^{\nb', (-y)},\; \Hc^{\nb', (-z)},&\text{and height}(\nb')\;\textrm{is odd},
\end{array}\right.
\end{equation}
the second, when the words {\em even} and {\em odd} are interchanged.
In the next section we show how these results may be reduced to the case of
Lagrangian Grassmannians.

\subsection{Kenyon-Pemantle recurrences for $\Gr_{\HH_+}^{\LL}(\HH,\omega)$}
\label{hexahedron_rels_lagrangian_grassmannian}

We now consider the Lagrangian counterpart of relations (\ref{eq:hexahedronRecurrenceGr36}), which were
 named ``hexahedron recurrence'' relations by R.~Kenyon and R.~Pemantle \cite{KePe2}. We will show that, for every element $w^0\in\Gr_{\HH_+}^\LL(\HH)$, suitably normalized evaluations of the KP $\tau$-function $\tau^{KP}_{w_0}\in\Cbb[[\tb]]$ give rise to solutions of the hexahedron relations. In particular, when the element $w_0\in\Gr_{\HH_+}^\LL(\HH)$ is such that the
 lattice evaluations of the formal power series $\tau^{KP}_{w_0}$ converge (e.g., for polynomial $\tau$-functions), we obtain a solution to the hexahedron recurrence in complex numbers.

Fix an element of the Lagrangian Grassmannian $w^0\in\Gr_{\HH_+}^\LL(\HH)$,  and define  four functions
 on the three-dimensional sublattice $\Bb_3\subset\Bb$
\begin{subequations}
\be
\hbc,\hbc^{(x)},\hbc^{(y)},\hbc^{(z)}:\Bb_3\rightarrow\left(S_y^\times\right)^{-1}\Cbb[[\yb,\tb']]
\ee
where, for $\nb'\in\Bb_3$, the values are
\bea
\hc^{\nb}:=&\&\mathcal T_{\mathfrak{sp}}(\Hc^{\nb'})=(-1)^{n'_1+n'_2+n'_3+n'_1n'_2+n'_1n'_3+n'_2n'_3}h^{\nb',\emptyset}_{w^0},\\
\hc^{\nb', (x)}:=&\&\mathcal T_{\mathfrak{sp}}(\Hc^{\nb', (\pm x)})=(-1)^{(n'_1)^2+n'_2n'_3+(n'_3)^2}h^{\nb',(1|2)}_{w^0},\\
\hc^{\nb', (y)}:=&\&\mathcal T_{\mathfrak{sp}}(\Hc^{\nb', (\pm y)})=(-1)^{n'_1n'_3+(n'_3)^2}h^{\nb',(0|2)}_{w^0},\\
\hc^{\nb', (z)}:=&\&\mathcal T_{\mathfrak{sp}}(\Hc^{\nb', (\pm z)})=(-1)^{n'_1n'_2+(n'_2)^2+(n'_3)^2}h^{\nb',(0|1)}_{w^0}.
\eea
\label{eq:mathfrakhDefinition}
\end{subequations}
Because of the Lagrangian conditions (\ref{eq: lagrange_cond_finite}) (for all $k$)   we have:
\be
h_{w^0}^{\nb',(1|2)}=h_{w^0}^{\nb',(2|1)}, \quad h_{w^0}^{\nb',(0|2)}= h_{w^0}^{\nb',(2|0)}, \quad h_{w^0}^{\nb',(1|0)}= h_{w^0}^{\nb',(0|1)} 
\ee
so we need not introduce $h_{w^0}^{\nb',(-x)}$, $h_{w^0}^{\nb',(-y)}$, $h_{w^0}^{\nb',(-z)}$ independently.
Denote the lattice translates of these  functions as
\be
\hbc_{\vb'},\hbc^{(x)}_{\vb'},\hbc^{(y)}_{\vb'},\hbc^{(z)}_{\vb'}:\Bb_3\rightarrow\left(S_y^\times\right)^{-1}\Cbb[[\yb,\tb']], \quad \vb'\in\Bb_3, 
\ee
with values
\begin{equation}
\hc^{\nb'}_{\vb'}:=\hc^{\nb'+\vb'},\quad \hc^{\nb', (\cdot)}_{\vb'}:=\hc^{\nb'+\vb', (\cdot)},\quad \nb',\vb'\in\Bb_3, \quad \cdot = x, y, z.
\end{equation}
Then as a specialization of Corollary \ref{cor:hexahedron_evaluations_general}, we obtain:
\begin{corollary}
\label{cor:hexahedron_tau_eval}
For any given element $w^0\in\Gr_{\HH_+}^\LL(\HH)$, the functions (\ref{eq:mathfrakhDefinition}) satisfy 
the system of recurrence equations, with positive coefficients:
\begin{subequations}
\bea
\hbc\hbc^{(x)}\hbc_{\betab_1}^{(x)}&\&=\hbc\hbc_{\betab_1}\hbc_{\betab_2+\betab_3} +\hbc_{\betab_1}\hbc_{\betab_2}
\hbc_{\betab_3}+\hbc^{(x)}\hbc^{(y)}\hbc^{(z)},
\\
\hbc\hbc^{(y)}\hbc_{\betab_2}^{(y)}&\&=\hbc\hbc_{\betab_2}\hbc_{\betab_1+\betab_3} +\hbc_{\betab_1}\hbc_{\betab_2}\hbc_{\betab_3}+\hbc^{(x)}\hbc^{(y)}\hbc^{(z)},
\\
\hbc\hbc^{(z)}\hbc^{(z)}_{\betab_3}&\&=\hbc
_{\betab_3}\hbc_{\betab_1+\betab_2} +\hbc_{\betab_1}\hbc_{\betab_2}\hbc_{\betab_3}+\hbc^{(x)}\hbc^{(y)}\hbc^{(z)},\\
 \hbc^2\hbc^{(x)}\hbc^{(y)}\hbc^{(z)}\hbc_{\betab_1+\betab_2+\betab_3} &\&= (\hbc^{(x)}\hbc^{(y)}\hbc^{(z)})^2+
 (\hbc^{(x)}\hbc^{(y)}\hbc^{(z)}) \cr
&\&\times \left(2\hbc_{\betab_1}\hbc_{\betab_2}\hbc_{\betab_3}+\hbc(\hbc_{\betab_1} \hbc_{\betab_2+\betab_3}+\hbc_{\betab_2}\hbc_{\betab_1+\betab_3} +\hbc_{\betab_3}\hbc_{\betab_1+\betab_2})\right)\cr
&\& +(\hbc_{\betab_1}\hbc_{\betab_2} +\hbc\hbc_{\betab_1+\betab_2})(\hbc_{\betab_1} \hbc_{\betab_3} +\hbc\hbc_{\betab_1+\betab_3})(\hbc_{\betab_2}\hbc_{\betab_3} +\hbc\hbc_{\betab_2+\betab_3}).\cr
&\&
\eea
\label{eq:hexahedronRecurrenceLagrangianGr36}
\end{subequations}
\end{corollary}

\subsection{Hyperdeterminantal relations}
\label{hyperdet_rec}

As shown in \cite{KePe2}, the hexahedron recurrence implies the hyperdeterminantal relations. Below we show that, for every element $w^0\in\Gr_{\HH_+}^\LL(\HH)$ of the Lagrangian Grassmanian, properly normalized evaluations of the KP $\tau$-function $\tau^{KP}_{w_0}\in\Cbb[[\tb]]$ give rise to a solution of the hyperdeterminantal relations in any dimensions.

The infinite-dimensional counterpart of the core hyperdeterminantal relations \cite{HoSt, Oed, AHH} is given by the following.
\begin{lemma}
For all triples $i,j,k\in\Nbb^+$ satisfying $i<j<k$, the following element belongs to the defining ideal of $\SS(\Gr_{\HH_+}^\LL(\HH,\omega))$
\begin{equation}
\begin{aligned}
\varkappa:=&\pi_{(\emptyset|\emptyset)}^2\pi_{(k,j,i|k,j,i)}^2 +\pi_{(i|i)}^2\pi_{(k,j|k,j)}^2 +\pi_{(j|j)}^2\pi_{(k,i|k,i)}^2 +\pi_{(k|k)}^2\pi_{(j,i|j,i)}^2\\
&-2\pi_{(\emptyset|\emptyset)}\pi_{(k,j,i|k,j,i)}\left(\pi_{(i|i)}\pi_{(k,j|k,j)} +\pi_{(j|j)}\pi_{(k,i|k,i)} +\pi_{(k|k)}\pi_{(j,i|j,i)}\right)\\
&-2\left(\pi_{(i|i)}\pi_{(j|j)}\pi_{(k,i|k,i)}\pi_{(k,j|k,j)} +\pi_{(j|j)}\pi_{(k|k)}\pi_{(j,i|j,i)}\pi_{(k,i|k,i)} +\pi_{(i|i)}\pi_{(k|k)}\pi_{(j,i|j,i)}\pi_{(k,j|k,j)}\right)\\
&+4\pi_{(\emptyset|\emptyset)}\pi_{(j,i|j,i)}\pi_{(k,i|k,i)}\pi_{(k,j|k,j)} +4\pi_{(i|i)}\pi_{(j|j)}\pi_{(k|k)}\pi_{(k,j,i|k,j,i)}.
\end{aligned}
\end{equation}
\label{lemm:HyperdeterminantalRelationFrobenius}
\end{lemma}
\begin{proof}
The element $\varkappa$ can be written as
\begin{equation}
\varkappa=\varrho_1^2 +4\pi_{(k,i|k,j)}^2\varrho_2 +4(\pi_{(i|i)\pi_{(j|j)}}-4\pi_{(\emptyset|\emptyset)}\pi_{(j,i|j,i)})\varrho_3,
\end{equation}
where
\begin{subequations}
\begin{eqnarray}
\varrho_1&:=&2\pi_{(i|j)}\pi_{(k,i|k,j)} +\pi_{(\emptyset|\emptyset)}\pi_{(k,j,i|k,j,i)} \\ &&-\pi_{(i|i)}\pi_{(k,j|k,j)}-\pi_{(j|j)}\pi_{(k,i|k,i)}+\pi_{(k|k)}\pi_{(j,i|j,i)}, \nonumber\\[0.3em]
\varrho_2&:=&\pi_{(i|j)}^2-\pi_{(i|i)}\pi_{(j|j)}+\pi_{(\emptyset|\emptyset)}\pi_{(j,i|j,i)}, \\[0.3em]
\varrho_3&:=&\pi_{(k,i|k,j)}^2-\pi_{(k,i|k,i)}\pi_{(k,j|k,j)}+\pi_{(k|k)}\pi_{(k,j,i|k,j,i)}.
\end{eqnarray}
\end{subequations}
The quadratic forms $\varrho_2$ and $\varrho_3$ coincide with (short) Pl\"ucker forms up to linear relations
and
\begin{equation}
\pi_{(i|j)}-\pi_{(j|i)},\;\pi_{(k,i|k,j)}-\pi_{(k,j|k,i)}\;\in\;\II^\LL.
\end{equation}
follow from the Lagrangian conditions (\ref{eq:LagrangianConditionKernelOmega}).
To show that $\varrho_1\in\II^\LL$, note that
\begin{equation}
\begin{aligned}
\varrho_1=&P_{i,j}+P_{i,k}+P_{j,k}\\
&+\pi_{(k,i|k,j)}(\pi_{(i|j)}-\pi_{(j|i)}) +\pi_{(j,i|k,j)}(\pi_{(i|k)}-\pi_{(k|i)})\\ &+\pi_{(j,i|k,i)}(\pi_{(k|j)}-\pi_{(j|k)})+\pi_{(j|k)}(\pi_{(j,i|k,i)}-\pi_{(k,i|j,i)})\\ &+\pi_{(i|k)}(\pi_{(k,j|j,i)}-\pi_{(j,i|k,j)}) +\pi_{(i|j)}(\pi_{(k,i|k,j)}-\pi_{(k,j|k,i)}),
\end{aligned}
\end{equation}
where
\begin{equation}
\begin{aligned}
P_{i,j}:=&\pi_{(k|k)}\pi_{(j,i|j,i)}-\pi_{(k|j)}\pi_{(j,i|k,i)}+\pi_{(k|i)}\pi_{(j,i|k,j)} -\pi_{(\emptyset|\emptyset)}\pi_{(k,j,i|k,j,i)},\\[0.3em]
P_{i,k}:=&-\pi_{(j|j)}\pi_{(k,i|k,i)}+\pi_{(j|k)}\pi_{(k,i|j,i)}+\pi_{(j|i)}\pi_{(k,i|k,j)} +\pi_{(\emptyset|\emptyset)}\pi_{(k,j,i|k,j,i)},\\[0.3em]
P_{j,k}:=&-\pi_{(i|i)}\pi_{(k,j|k,j)}-\pi_{(i|k)}\pi_{(k,j|j,i)}+\pi_{(i|j)}\pi_{(k,j|k,i)} +\pi_{(\emptyset|\emptyset)}\pi_{(k,j,i|k,j,i)}
\end{aligned}
\end{equation}
are four-term Pl\"ucker quadratic forms.
\end{proof}

Fix an element $w^0\in\Gr_{\HH_+}^\LL(\HH,\omega)$ of the infinite Lagrangian Grassmannian and 
define the function
\bea
\hb:\Bb &\&\rightarrow\left(S_y^\times\right)^{-1}\Cbb[[\yb,\tb']],\cr
&\& \cr
 \hb: \nb'&\& \mapsto h^{\nb'}_{w^0}, \quad \forall \  \nb'\in\Bb
\label{eq:hb_bold_n_prime_def}
\eea
on the infinite lattice $\Bb$, where $h^{\nb}_{w^0}$ is defined by (\ref{eq:LagrangianEvaluationTauFunction}) for any $\nb\in \Ab$.
 Denote the lattice functions with shifted argument by
\bea
\hb_\vb:\Bb&\&\rightarrow\left(S_y^\times\right)^{-1}\Cbb[[\yb,\tb']],\cr
&\& \cr
\hb_{\vb'}: \nb'&\& \mapsto h^{\nb'+\vb'}_{w^0}, \quad \forall \ \nb', \vb' \in \Bb.
\label{eq:hb_bold_n_prime_v_def}
\eea
\begin{proposition}
\label{hexahedron_Lagrangian_36}
 For any ordered triple $i,j,k\in\Nbb^+$ of nonnegative integers
 $i<j<k$, the hyperdeterminantal relations:\footnote{
 To compare with the notations of Proposition 3.19 and Corollary 3.20 in \cite{AHH}, we give the folllowing lexicon
 \be
 h^{\bf 0} = \sigma^{\bf 0}, \quad h^{\bf 0}_{\betab_i} = - \sigma^{\bf 0}_i, \quad h^{\bf 0}_{\betab_i +\betab_j} = - \sigma^{\bf 0}_{ij}, 
 \quad h^{\bf 0}_{\betab_i +\betab_j +\betab_k} = - \sigma^{\bf 0}_{ijk}.
 \ee
}
\begin{equation}
\begin{aligned}
&\hb^2\hb_{\betab_i+\betab_j+\betab_k}^2 +\hb_{\betab_i}^2\hb_{\betab_j+\betab_k}^2 +\hb_{\betab_j}^2\hb_{\betab_i+\betab_k}^2 +\hb_{\betab_k}^2\hb_{\betab_i+\betab_j}^2\\
&-2\hb\hb_{\betab_i+\betab_j+\betab_k}\left(\hb_{\betab_i}\hb_{\betab_j+\betab_k} +\hb_{\betab_j}\hb_{\betab_i+\betab_k} +\hb_{\betab_k}\hb_{\betab_i+\betab_j}\right)\\
&-2\left(\hb_{\betab_i}\hb_{\betab_j}\hb_{\betab_i+\betab_k}\hb_{\betab_j+\betab_k} +\hb_{\betab_j}\hb_{\betab_k}\hb_{\betab_i+\betab_j}\hb_{\betab_i+\betab_k} +\hb_{\betab_i}\hb_{\betab_k}\hb_{\betab_i+\betab_j}\hb_{\betab_j+\betab_k}\right)\\
&+4\hb\hb_{\betab_i+\betab_j}\hb_{\betab_i+\betab_k}\hb_{\betab_j+\betab_k} +4\hb_{\betab_i}\hb_{\betab_j}\hb_{\betab_k}\hb_{\betab_i+\betab_j+\betab_k} =0.
\end{aligned}
\end{equation}
\label{prop:HyperdeterminantalRecurrence}
are satisfied.
\end{proposition}
\begin{proof}
This follows from combining Lemma \ref{lemm:HyperdeterminantalRelationFrobenius} with Theorem \ref{th:LatticeOfPointsInInfiniteLagrangianGrassmannian}.
\end{proof}

 \bigskip
 \bigskip
\noindent
\small{ {\it Acknowledgements.}. This work was partially supported by the Natural Sciences and Engineering Research Council of Canada (NSERC).

 \bigskip
\noindent
\small{ {\it Data sharing.}
Data sharing is not applicable to this article since no new data were created or analyzed in this study.
\bigskip


 \newcommand{\arxiv}[1]{\href{http://arxiv.org/abs/#1}{arXiv:{#1}}}

\end{document}